\documentclass[preprint,3p,times,fleqn]{elsarticle} %final for normal linespace

\pdfoutput=1

\makeatletter
\def\ps@pprintTitle{%
 \let\@oddhead\@empty
 \let\@evenhead\@empty
 \def\@oddfoot{\centerline{\thepage}}%
 \let\@evenfoot\@oddfoot}
\makeatother

\usepackage{amsmath,amssymb,amsthm,mathtools}

\usepackage{lineno}
\usepackage{etoolbox} %% <- for \pretocmd and \apptocmd
\makeatletter %% <- make @ usable in macro names
\newcommand*\linenomathpatch{\@ifstar{\linenomathpatch@AMS}{\linenomathpatch@}}
\newcommand*\linenomathpatch@[1]{
  \expandafter\pretocmd\csname #1\endcsname {\linenomathWithnumbers}{}{}
  \expandafter\pretocmd\csname #1*\endcsname{\linenomathWithnumbers}{}{}
  \expandafter\apptocmd\csname end#1\endcsname {\endlinenomath}{}{}
  \expandafter\apptocmd\csname end#1*\endcsname{\endlinenomath}{}{}
}
\newcommand*\linenomathpatch@AMS[1]{
  \expandafter\pretocmd\csname #1\endcsname {\linenomathWithnumbersAMS}{}{}
  \expandafter\pretocmd\csname #1*\endcsname{\linenomathWithnumbersAMS}{}{}
  \expandafter\apptocmd\csname end#1\endcsname {\endlinenomath}{}{}
  \expandafter\apptocmd\csname end#1*\endcsname{\endlinenomath}{}{}
}
\let\linenomathWithnumbersAMS\linenomathWithnumbers
\patchcmd\linenomathWithnumbersAMS{\advance\postdisplaypenalty\linenopenalty}{}{}{}
\makeatother %% revert @

\linenomathpatch{equation}
\linenomathpatch*{gather}
\linenomathpatch*{multline}
\linenomathpatch*{align}
\linenomathpatch*{alignat}
\linenomathpatch*{split}

\usepackage{bbold}
\usepackage{bbding}

\usepackage{hyperref}

\newtheorem{theorem}{Theorem}
\newtheorem{lemma}[theorem]{Lemma}

\newtheorem{corollary}[theorem]{Corollary}

\newtheorem{example}[theorem]{Example}

\usepackage{tikz}

\usetikzlibrary{arrows,positioning,automata}

\tikzstyle{automaton}=[>=latex',shorten >=1pt,node distance=1.8cm,on grid,auto,
roundnode/.style={circle, draw=black, thick, minimum size=7mm},->]

\newcommand{\N}{\mathbb{N}}
\newcommand{\B}{\mathbb{B}}
\usepackage{xspace}

\newcommand{\mathzero}{\mathbb{0}}
\newcommand{\mathone}{\mathbb{1}}

\newcommand{\mpunkt}{\,\text{.}}
\newcommand{\mkomma}{\,\text{,}}
\newcommand*{\msemikolon}{\,\text{;}}

\newcommand*{\ssig}{S \langle \Sigma \rangle}

\newcommand*{\seps}{S \langle \{\epsilon\} \rangle}
\newcommand*{\ssigstar}{S \langle\langle\Sigma^*\rangle\rangle}
\newcommand*{\ssigomega}{S\langle\langle\Sigma^\omega\rangle\rangle}
\newcommand*{\salg}{S^\text{alg}\langle\langle\Sigma^*\rangle\rangle}
\newcommand*{\snn}{S^{n \times n}}
\newcommand*{\snngam}{(\snn)^{\Gamma^* \times \Gamma^*}}
\newcommand*{\spair}{(\ssigstar,\allowbreak\ssigomega)}
\newcommand*{\stimes}{\ssigstar\times\ssigomega}
\newcommand*{\salgomega}{S^\text{alg}\langle\langle\Sigma^\omega\rangle\rangle}
\newcommand*{\salgtimes}{\salg\times\salgomega}

\newcommand*{\push}{\downarrow}
\newcommand*{\internal}{\#}
\newcommand*{\pop}{\uparrow}

\newcommand{\what}[1]{\,\widehat{\!{#1}}}

\begin{document}

\begin{frontmatter}

\title{Greibach Normal Form for \texorpdfstring{$\omega$}{omega}-Algebraic Systems and\\ Weighted Simple \texorpdfstring{$\omega$}{omega}-Pushdown Automata}

\author[leipzig]{Manfred Droste}
\ead{droste@informatik.uni-leipzig.de}

\author[leipzig,paris]{Sven Dziadek\fnref{dfg}}
\ead{dziadek@informatik.uni-leipzig.de}

\author[wien]{Werner Kuich\fnref{fwf}}
\ead{werner.kuich@tuwien.ac.at}

\address[leipzig]{Institut für Informatik, Universität Leipzig, Germany}
\address[wien]{Institut für Diskrete Mathematik und Geometrie, Technische Unversität Wien, Austria}
\address[paris]{LRDE, EPITA, Kremlin-Bicêtre, France}
\fntext[dfg]{Supported by Deutsche Forschungsgemeinschaft (DFG), Graduiertenkolleg 1763 (QuantLA)}
\fntext[fwf]{Partially supported by Austrian Science Fund (FWF): grant no. I1661 N25}

\begin{abstract}
 In weighted automata theory, many classical results on formal languages have been extended into a quantitative setting. Here, we investigate weighted context-free languages of infinite words, a generalization of $\omega$-context-free languages (Cohen, Gold 1977) and an extension of weighted context-free languages of finite words (Chomsky, Schützenberger 1963).
As in the theory of formal grammars, these weight\-ed context-free languages, or $\omega$-algebraic series, can be represented as solutions of mixed $\omega$-algebraic systems of equations and by weighted $\omega$-pushdown automata.

In our first main result,
we show that (mixed) $\omega$-algebraic systems can be transformed into \emph{Greibach normal form}.
We use the Greibach normal form in our second main result to prove that \emph{simple $\omega$-reset pushdown automata} recognize all $\omega$-algebraic series. Simple $\omega$-reset automata do not use $\epsilon$-transitions
and can change the stack only
by at most one symbol. These results generalize fundamental properties
of context-free languages to weighted context-free languages.
\end{abstract}


\end{frontmatter}

%\linenumbers

\section{Introduction}
\label{sec:introduction}

Context-free languages provide a fundamental concept for
programming languages in computer science. In order to model
quantitative properties, already in 1963, Chomsky and Schützenberger~\cite{chomsky1963}
introduced weighted context-free languages. The theory of weighted pushdown automata has been extensively studied;
for background, we refer the reader to the survey~\cite{kuich_semirings}
and the books \cite{kuich_formal_power_series,MAT,handbook}.
In 1977, Cohen and Gold~\cite{cohen_gold}
investigated context-free languages of infinite words. Weighted
$\omega$-context-free languages, i.e., $\omega$-algebraic series were studied
by Ésik and Kuich~\cite{kuich_infinite}.

The goal of this paper is the investigation of weighted context-free
languages and weighted pushdown automata on infinite words.
As in~\cite{MAT}, the weighted context-free languages of
infinite words are described by solutions of $\omega$-algebraic
systems and mixed $\omega$-algebraic
systems of equations. In our first main result, we show that these
systems can be transformed into a Greibach normal form.
In the literature, Greibach normal forms, central for the theory of
context-free languages of finite words, have been established for
$\omega$-context-free languages (of infinite words), see~\cite{cohen_gold},
and also for algebraic systems of equations for series over finite words~\cite{kuich_formal_power_series,MAT};
this latter result is employed in our proof.
Hence here we extend these classical results to a weighted version for
infinite words.

Recently, Droste, Ésik and Kuich introduced weighted $\omega$-pushdown automata in~\cite{triple-pair,kleene_omega_pda}.
In our second main result, we consider weighted \emph{simple}
$\omega$-pushdown automata that we call simple $\omega$-reset pushdown automata here. 
These automata do not use $\epsilon$-transitions
and utilize only three simple stack commands: popping a symbol,
pushing a symbol or leaving the stack unaltered; moreover, it is only
possible to read the topmost stack symbol by popping it. Observe that
together with the restriction of not allowing $\epsilon$-transitions,
restrictions for the actions on the stack are non-trivial.
In our second main result we show that these \emph{simple}
$\omega$-reset pushdown automata recognize all weighted $\omega$-context-free
languages. For our proof, we use that $\omega$-algebraic systems can be brought into Greibach
normal form by our present first main result. 
Our construction of simple $\omega$-reset pushdown automata
is deduced from the construction used in
a recent corresponding result~\cite{finite_simple}, which states that
simple reset pushdown automata on finite words recognize all algebraic series.

We believe the model of simple $\omega$-reset pushdown automata
to be very natural. Similar expressivity equivalence results
in the unweighted case hold for context-free languages of finite
words, as used in a proof by Blass and Gurevich~\cite{note_nw},
and also for $\omega$-context-free languages, see~\cite{unweighted_logic}.
For a similar automaton model as the simple $\omega$-reset pushdown automata
introduced here, we show a logical characterization in~\cite{weighted_logic}.
Here, we close an important gap showing that in fact all $\omega$-algebraic series
can be converted into a weighted logical formula as described in~\cite{weighted_logic}.

To accomplish our goals, we introduce the following new elements.
We establish a new method to compute the matrix operations ${}^{\omega,t}$.
To prove the existence of the Greibach normal form for $\omega$-algebraic systems,
we introduce a new construction that transforms mixed $\omega$-algebraic systems
into $\omega$-algebraic systems.
For our automaton model, we define and investigate simple reset pushdown matrices.
Pushdown matrices historically are indexed first by the stack and subsequently by the states;
for our transformation from $\omega$-algebraic systems in Greibach normal form to
simple $\omega$-reset pushdown automata, we exploit and refine a notation that reverses this index order.
Finally, we show how the unicity of $l$\textsuperscript{th} canonical solutions can be deployed to show equality of two expressions.

Hereafter, we recall basic definitions in Section~\ref{sec:preliminaries}. There, we also extend our knowledge of the matrix operations~${}^{\omega,t}$.

After the Preliminaries, in Section~\ref{sec:mixed_algebraic_systems}, we introduce $\omega$-algebraic systems and mixed $\omega$-algebraic systems and their canonical solutions.
Then, we characterize $\omega$-algebraic series by a series of equivalent statements.

The main result of Section~\ref{sec:greibach_normal_form} states that each $\omega$-algebraic series is a component of a canonical solution of a mixed $\omega$-algebraic system in Greibach normal form. 

In Section~\ref{sec:greibach_normal_form_unmixed} we specialize the main result of Section~\ref{sec:greibach_normal_form}:  now each $\omega$-algebraic series is a component of a canonical solution of an $\omega$-algebraic system in Greibach normal form. 

We consider simple reset pushdown automata in Section~\ref{sec:simple_reset_pushdown_automata} and recall the result of~\cite{finite_simple} that for each algebraic series $r$ there exists a simple reset pushdown automaton with behavior $r$.
 
Simple $\omega$-reset pushdown automata are introduced in Section~\ref{sec:simple_omega_reset_pushdown_automata}. The main result of this section and of the whole paper is that for each $\omega$-algebraic series $r$ it is possible to construct a simple $\omega$-reset pushdown automaton with behavior $r$.
%Hence, we conclude that for each $\omega$-algebraic series $r$ we can construct a simple $\omega$-reset pushdown automaton with behavior $r$.

A preliminary version of this paper appeared in \cite{ddk_fsttcs}.
In this version, we strengthen the first main result by proving that already $\omega$-algebraic systems can be transformed into Greibach normal form.
In~\cite{ddk_fsttcs}, we only showed the existence of the Greibach normal form for mixed $\omega$-algebraic systems.
The stronger result in this work allows us to generalize the second main result: weighted simple $\omega$-pushdown automata recognize all $\omega$-algebraic series. For this, we needed to adapt the construction such that our simple $\omega$-reset pushdown automata behave exactly like the canonical solutions of $\omega$-algebraic systems.
Furthermore, we add a result (see Theorem~\ref{thm:omegak}) describing $\omega$-powers of matrices considering Büchi-acceptance.
We give complete arguments and further examples for our results.

%%%%%%%%%%%%%%%%%%%%%%%%%%%%%%%%%%%%%%%%%%%%%%%%%%%

\section{Preliminaries}
\label{sec:preliminaries}

For the convenience of the reader, we recall definitions and results from Ésik, Kuich \cite{MAT}.

A monoid $\langle S,+,0\rangle$ is called \emph{complete} if it is equipped with sum operations $\sum_I$ for all families $(a_i \mid i \in I)$ of elements of $S$, where $I$ is an arbitrary index set, such that the following conditions are satisfied
(see Conway \cite{conway_regular}, Eilenberg \cite{eilenberg_automata}, Kuich \cite{kuich_semirings}):
\begin{align*}
\text{(i)} \quad & \sum\limits_{i \in \emptyset} a_i = 0, \qquad
\sum\limits_{i \in \{j\}} a_i = a_j,
\qquad \sum\limits_{i \in \{j,k\}} a_i = a_j + a_k \text{ for } j \neq k \, , \\
\text{(ii)} \quad & \sum\limits_{j \in J}\big(\sum_{i \in I_j} a_i \big) = \sum_{i \in I} a_i \, ,
\text{ if }\ \bigcup_{j \in J}\! I_j = I \ \text{ and }\ I_j \cap I_{j'} = \emptyset \
\text{ for } \ j \neq j' \mpunkt
\intertext{Furthermore, a semiring $\langle S,+,\cdot,0,1\rangle$ is called \emph{complete} if $\langle S,+,0\rangle$ is a complete monoid and if we additionally have}
\text{(iii)} \quad & \sum_{i \in I} (c \cdot a_i) = c \cdot \big( \sum_{i \in I} a_i \big) ,
\qquad \sum_{i \in I} (a_i \cdot c) = \big( \sum_{i \in I } a_i \big) \cdot c \mpunkt
\end{align*}

This means that a semiring $S$ is complete if it has ``infinite sums'' (i) that are an extension of the finite sums, (ii) that are associative and commutative and
(iii) that satisfy the distributivity laws.

A semiring $S$ equipped with an additional unary star operation $^*: S \to S$ is called a \emph{starsemiring}. In complete semirings for each element $a$, the \emph{star} $a^*$
of $a$ is defined by
\begin{equation*}
a^* = \sum_{j\geq 0} a^j \, .
\end{equation*}
Hence, each complete semiring is a starsemiring, called a \emph{complete starsemiring}.

Starsemirings allow us to generalize the star operation to matrices. Let $M\in S^{n\times n}$, then we define $M^*\in S^{n\times n}$ inductively as in Ésik, Kuich~\cite{MAT}, pp. 14--15 as follows. For $n=1$ and $M=(a)$, for $a\in S$, we let $M^*=(a^*)$. Now, for $n>1$, we partition $M$ into submatrices, called \emph{blocks},
\begin{equation}\label{abcd_matrix}
M=\begin{pmatrix}a & b\\c & d\end{pmatrix}\mkomma
\end{equation}
with $a\in S^{1\times 1}$, $b\in S^{1\times (n-1)}$, $c\in S^{(n-1)\times 1}$, $d\in S^{(n-1)\times(n-1)}$, and we define
\begin{equation}
  \label{matrix_star_1}
  M^*=\begin{pmatrix}(a+b d^* c)^* & (a+b d^* c)^* b d^*\\(d+c a^* b)^* c a^* & (d+c a^* b)^*\end{pmatrix}\mpunkt
\end{equation}

Whenever we use a matrix $M$ as defined in~\eqref{abcd_matrix}, the corresponding automaton can be illustrated as follows: %in Figure~\ref{fig:abcd_matrix}.
%\begin{figure}[t]
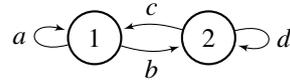
\begin{center}  %\small
  \begin{tikzpicture}[automaton,node distance=1.5cm]
    \node[roundnode] (1) {1};
    \node[roundnode,right=of 1] (2) {2};
    \path[bend angle=17,bend right] (1) edge node[below] {$b$} (2);
    \path[bend angle=17,bend right] (2) edge node[above] {$c$} (1);
    \path (1) edge [loop left] node[left] {$a$} ();
    \path (2) edge [loop right] node[right] {$d$} ();
  \end{tikzpicture}
\end{center}
%  \caption[Depiction of ``abcd'' matrix as automaton]{Depiction of matrix $M=\begin{pmatrix}a & b\\c & d\end{pmatrix}$ as automaton}
%  \label{fig:abcd_matrix}
%\end{figure}

%% selfmade:
A semiring is called \emph{continuous} if it is ordered, each directed subset has a least upper bound and addition and multiplication preserve the least upper bound of directed sets. Any continuous semiring is complete. See Ésik, Kuich~\cite{MAT} for background.

%% From triple-pair:
Suppose that $S$ is a semiring and $V$ is a commutative monoid written additively.
We call $V$ a (left) $S$-semimodule if $V$ is
equipped with a (left) action
\begin{align*}
S \times V & \ \to\ V\\
(s,v) & \ \mapsto \  s v
\end{align*}
subject to the following rules:
\begin{align*}
& s(s' v) = (s s')v \, , \quad (s+s')v  = s v + s' v \, ,  \quad s(v+v')  = s v + s v' \, , \\
& 1v  = v \, , \quad  0v   = 0 \, ,\quad s0  = 0 \, ,
\end{align*}
for all $s,s' \in S$ and $v,v' \in V$.
If $V$ is an $S$-semimodule, we call $(S,V)$ a
\emph{semiring-semimodule pair}.

Suppose that $(S,V)$ is a semiring-semimodule pair
such that $S$ is a starsemiring and $S$ and $V$ are equipped with an omega operation
$^\omega: S \to V$.
Then we call $(S,V)$ a \emph{starsemiring-omegasemimodule pair}.

Ésik, Kuich \cite{iteration_semiring_semimodule_pairs} define a
\emph{complete semiring-semimodule pair} to be a semiring-semimodule pair $(S,V)$ such that $S$ is a complete semiring and $V$ is a complete monoid with
\begin{equation*}
s\Bigl(\sum_{i\in I} v_i \Bigr)  = \sum_{i \in I} s v_i \qquad \text{and} \qquad
\Bigl(\sum_{i\in I} s_i\Bigr) v  = \sum_{i \in I} s_i v \, ,
\end{equation*}
\noindent
for all $s\in S$, $v \in V$, and for all families
$(s_i)_{i \in I} $ over $S$ and $(v_i)_{i \in I} $ over $V$;
moreover, it is required that an \emph{infinite product operation}
\begin{equation*}
S^\omega \ni (s_1, s_2, \ldots) \ \mapsto \ \prod_{j\geq 1} s_j\in V
\end{equation*}
is given mapping infinite sequences over $S$ to $V$ subject to the following three conditions:
\begin{align*}
\text{(i)} \quad & \prod_{i\geq 1} s_i  \ = \ \prod_{i \geq 1} (s_{n_{i-1}+1}\cdot \dots \cdot s_{n_i})\mkomma\\
\text{(ii)} \quad & s_1 \cdot \prod_{i \geq 1} s_{i+1}  \ = \ \prod_{i\geq 1} s_i\mkomma\\
\text{(iii)} \quad & \prod_{j\geq 1} \sum_{i_j \in I_j} s_{i_j}  \ = \
\sum_{(i_1, i_2, \dots)\in I_1 \times I_2 \times \dots} \prod_{j\geq 1} s_{i_j} \mkomma
\end{align*}

\noindent
where in the first equation $0=n_0 \leq n_1 \leq n_2 \leq \dots$ and $I_1, I_2, \dots$ are arbitrary index sets. This means that the left action of the semimodule is distributive and it is required that it has ``infinite products'' mapping infinite sequences over $S$ to $V$ such that the product (i) can be partitioned (an infinite form of associativity), (ii) can be extended from the left and (iii) satisfies an infinite distributivity law.

Suppose that $(S, V)$ is complete.
Then we define
\begin{equation*}
s^*  \ = \ \sum_{i\geq 0} s^i \qquad \text{and} \qquad
s^\omega \  = \ \prod_{i\geq 1} s \, ,
\end{equation*}
\noindent
for all $s \in S$. This turns $(S, V)$ into a starsemiring-omegasemimodule pair.
Observe that, if $(S,V)$ is a complete semiring-semimodule pair,
then $0^\omega = 0$.

A \emph{star-omega semiring} is a semiring $S$ equipped with unary operations $^*$ and $^\omega : S \to S$.
A star-omega semiring $S$ is called \emph{complete} if $(S,S)$ is a complete semiring-semimodule pair,
i.e., if $S$ is complete and is equipped with an infinite product operation that satisfies the  three conditions stated above.
%\TODO{streichen:} A complete star-omega semiring $S$ is called \emph{commutative} if the semiring $S$ is commutative and, for all bijections $\pi\colon \N\to \N$, and $s_j\in S$ for $j\geq 0$, we have $\prod_{j\geq 0}s_{\pi(j)}=\prod_{j\geq0}s_j$.
A complete star-omega semiring $S$ is called \emph{continuous} if the semiring $S$ is continuous.

\begin{example}
Formal languages are covered by our model. Let $\langle\B,+,\cdot,0,1\rangle$ be the Boolean semiring. Then let $0^*=1^*=1$ and take infima as infinite products. This makes $\B$ a continuous star-omega and commutative semiring. It then follows that $\B\langle\langle\Sigma^*\rangle\rangle\times\B \langle\langle\Sigma^\omega\rangle\rangle$ is isomorphic to formal languages of finite and infinite words with the usual operations.

The semiring $\langle\N^\infty,+,\cdot,0,1\rangle$ with $\N^\infty=\N\cup\{\infty\}$ and the natural infinite product operation of numbers is a continuous star-omega and commutative semiring.

The tropical semiring $\langle \N^\infty,\min,+,\infty,0\rangle$ with the usual infinite sum operation as infinite product is a commutative semiring and a continuous star-omega semiring.

Analogously, the arctic semiring $\langle \bar\N,\max,+,-\infty,0\rangle$ with $\bar\N=\N\cup\{-\infty,\infty\}$ and the infinite sum operation as infinite product is a commutative semiring and a continuous star-omega semiring.
\end{example}

A \emph{Conway semiring} (see Conway \cite{conway_regular}, Bloom, Ésik \cite{bloom}) is a starsemiring $S$
satisfying the \emph{sum star identity}
\[(a+b)^* = a^*(b a^*)^*\]
and the \emph{product star identity}
\[(ab)^* = 1+a(b a)^*b\]
for all $a,b \in S$.
Observe that by Ésik, Kuich~\cite{MAT}, Theorem 1.2.24, each complete starsemiring is
a Conway semiring.

Note that from the identities in Conway semirings, it follows
\begin{align}
  \label{eqn:conway_derived}
  \begin{split}
    a^*&=1+a a^*=1+a^* a\mkomma\\
    a (b a)^* &= (a b)^* a\mkomma
  \end{split}
\end{align}
for all $a,b\in S$.

If $S$ is a Conway semiring then so is $\snn$.
Let $M\in\snn$. Assume that $n > 1$ and write $M$ as in~\eqref{abcd_matrix}. Applying the identities of Conway semirings, we get an equivalent definition (cf.\@ Conway~\cite{conway_regular}, pp. 27--28) to~\eqref{matrix_star_1}:
\begin{equation}
  \label{matrix_star_2}
  M^*=\begin{pmatrix}(a+b d^* c)^* & a^* b(d+c a^* b)^*\\ d^* c(a+b d^* c)^* & (d+c a^* b)^*\end{pmatrix}\mpunkt
\end{equation}

Following Bloom, Ésik \cite{bloom}, we call a starsemiring-omegasemimodule pair $(S,V)$ a
\emph{Conway semiring-semimodule pair} if $S$ is a
Conway semiring and if the omega operation satisfies
the \emph{sum omega identity} and the \emph{product omega identity}:
\[(a+b)^\omega  = (a^*b)^\omega + (a^*b)^*a^\omega \qquad \text{and} \qquad
(ab)^\omega  = a(b a)^\omega\mkomma\]
for all $a,b \in S$. By Ésik, Kuich~\cite{iteration_semiring_semimodule_pairs} each complete semiring-semimodule pair is a Conway semiring-semimodule pair.

Observe that the
\emph{omega fixed-point equation} holds, i.e.
\[a a^\omega = a^\omega\mkomma\]
for all $a \in S$.

Consider a starsemiring-omegasemimodule pair $(S, V)$. Following Bloom, Ésik~\cite{bloom}, we define a matrix operation ${}^\omega\colon S^{n\times n} \to V^{n\times 1}$ on a starsemiring-omegasemimodule pair $(S,V)$ as follows. If $n = 0$, $M^\omega$ is the unique element of $V^0$, and if $n = 1$, so that $M = (a)$, for some $a\in S$, $M^\omega = (a^\omega)$. Assume now that $n > 1$ and write $M$ as in~\eqref{abcd_matrix}. Then
\[M^\omega=\begin{pmatrix}(a + b d^* c)^\omega + (a + b d^* c)^* b d^\omega\\(d + c a^* b)^\omega + (d + c a^* b)^* c a^\omega \end{pmatrix}\mpunkt\]
Additionally, the \emph{matrix star identity} is valid for Conway semirings and states that the star of a matrix is independent of the partitioning of the matrix. The \emph{matrix omega identity} is valid for Conway semiring-semimodule pairs and states that the operation ${}^\omega$ is independent of the partitioning of the matrix, i.e., the blocks of~\eqref{abcd_matrix} can have arbitrary sizes: $a\in S^{n_1\times n_1}$, $b\in S^{n_1\times n_2}$, $c\in S^{n_2\times n_1}$, $d\in S^{n_2\times n_2}$ for $n_1+n_2=n$. If $(S,V)$ is a Conway semiring-semimodule pair, then so is $(\snn,V^n)$.
See also Ésik, Kuich~\cite{MAT}, page 106.

Following Ésik, Kuich~\cite{ek_semiring}, we define matrix operations ${}^{\omega,t} \colon S^{n\times n}\to V^{n\times 1}$ for $0 \leq t \leq n$ as follows. Assume that $M\in S^{n\times n}$ is decomposed into blocks $a, b, c, d$ as in~\eqref{abcd_matrix}, but with $a$ of dimension $t\times t$ and $d$ of dimension $(n-t)\times (n-t)$. Then
\begin{equation}
\label{eqn:omega_buchi_1}
M^{\omega,t}=\begin{pmatrix}(a + b d^* c)^\omega\\d^* c (a + b d^* c)^\omega \end{pmatrix}\mpunkt
\end{equation}
Observe that $M^{\omega,0}=0$ and $M^{\omega,n}=M^\omega$.
Intuitively, $M$ can be interpreted as an adjacency matrix of the following automata with $n$ states:
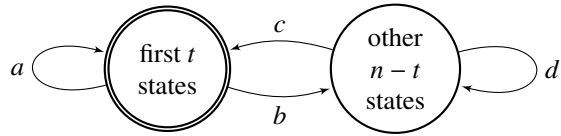
\begin{center}  %\small
  \begin{tikzpicture}[automaton,node distance=3cm]
    \node[roundnode,align=center,minimum size=1.6cm,accepting] (1) {first $t$\\ states};
    \node[roundnode,right=of 1,align=center,minimum size=1.7cm] (2) {other\\ $n-t$\\ states};
    \path[bend angle=17,bend right] (1) edge node[below] {$b$} (2);
    \path[bend angle=17,bend right] (2) edge node[above] {$c$} (1);
    \path (1) edge [loop left] node[left] {$a$} ();
    \path (2) edge [loop right] node[right] {$d$} ();
  \end{tikzpicture}
\end{center} Then $M^{\omega,t}$ are infinite paths where the first $t$ states are repeated states, i.e., states that are Büchi-accepting.

The next theorem states that, in case of a Conway semiring, $M^{\omega,t}$, for $0 \leq t \leq n$, can be computed also in a way different from its definition and, with certain limits, is independent of the partitioning of the matrix $M$.

\begin{theorem}
  \label{thm:omegak}
  \setlength\arraycolsep{4pt}
  Let $S$ be a Conway semiring and $0\leq t\leq k\leq n$.
  Assume $M\in S^{n\times n}$ is decomposed into blocks
  \[M=\begin{pmatrix} a & b \\ c & d\end{pmatrix}\]
  with block $a$ being of dimension $k\times k$ and block $d$ of dimension $(n-k)\times(n-k)$.

  Then we have,% for all $k$ with $t\leq k\leq n$,
  \begin{equation}
    \label{eqn:omega_buchi_2}
    M^{\omega,t}=
    \begin{pmatrix}
      (a+b d^*c)^{\omega,t}\\
      d^* c(a+b d^* c)^{\omega,t}
    \end{pmatrix}\mpunkt
  \end{equation}
\end{theorem}
\begin{proof}
  \setlength\arraycolsep{4pt}
  The proof resembles the proof of the matrix omega identity (cf.~\cite{MAT}, Theorem~5.3.13).
  Assume $M\in S^{n\times n}$ is decomposed into nine blocks
  \[M=\begin{pmatrix}
      f & g & h\\
      i & a & b\\
      j & c & d
  \end{pmatrix}\]
  with dimensions $f\in S^{t\times t}$, $a\in S^{(k-t)\times (k-t)}$ and $d\in S^{(n-k)\times (n-k)}$.
  Consider the following two partitionings:
  %% \[M=\left(
  %% \begin{array}{c|c}
  %%   f & \begin{matrix}
  %%     g & h
  %%   \end{matrix}\\
  %% \hline
  %% \begin{matrix}
  %%   i \\ j
  %% \end{matrix} &
  %% \begin{matrix}
  %%   a & b\\
  %%   c & d
  %% \end{matrix}\\
  %%   \multicolumn{1}{c}{$\upbracefill$} & $\upbracefill$\\
  %%   \multicolumn{1}{c}{t} & n-t\\
  %%   \noalign{\vspace{-2\normalbaselineskip}}
  %% \end{array}
  %% \right)
  %% \vspace{1.5\normalbaselineskip}
  %%     \begin{array}{l}
  %%     \rdelim\}{1}{2mm}[$t$]\\ \rdelim\}{2}{2mm}[$n-t$] \\ \\
  %% \end{array}\hspace{1.5cm}
  %% M'=\left(
  %% \begin{array}{c|c}
  %%   \begin{matrix}
  %%     f & g\\
  %%     i & a
  %%   \end{matrix} &
  %%   \begin{matrix}
  %%     h \\ b
  %%   \end{matrix}\\
  %%   \hline
  %%   \begin{matrix}
  %%     j & c
  %%   \end{matrix} &
  %%   d\\
  %%   \multicolumn{1}{c}{$\upbracefill$} & $\upbracefill$\\
  %%   \multicolumn{1}{c}{k} & n-k\\
  %%   \noalign{\vspace{-2\normalbaselineskip}}
  %% \end{array}
  %% \right)
  %% \vspace{1.5\normalbaselineskip}
  %%     \begin{array}{l}
  %%     \rdelim\}{2}{2mm}[$k$]\\ \\ \rdelim\}{1}{2mm}[$n-k$]
  %% \end{array}\]
  %% old version:
  \[M=\left(
  \begin{array}{c|c}
    f & \begin{matrix}
      g & h
    \end{matrix}\\
  \hline
  \begin{matrix}
    i \\ j
  \end{matrix} &
  \begin{matrix}
    a & b\\
    c & d
  \end{matrix}
  \end{array}
  \right) \hspace{2cm}
  M'=\left(
  \begin{array}{c|c}
    \begin{matrix}
      f & g\\
      i & a
    \end{matrix} &
    \begin{matrix}
      h \\ b
    \end{matrix}\\
    \hline
    \begin{matrix}
      j & c
    \end{matrix} &
    d
  \end{array}
  \right)\]

  Now we need to show that $M^{\omega,t}$, calculated as in~\eqref{eqn:omega_buchi_1}
  \[M^{\omega,t}=
  \begin{pmatrix}
    \alpha\\
    \begin{pmatrix}
      a & b\\ c & d
    \end{pmatrix}^*
    \begin{pmatrix}
      i \\ j
    \end{pmatrix}
    \alpha
  \end{pmatrix}\mkomma\]
  where
  \[\alpha = \left(f +
  \begin{pmatrix}
    g & h
  \end{pmatrix}
  \begin{pmatrix}
    a & b\\ c & d
  \end{pmatrix}^*
  \begin{pmatrix}
    i\\ j
  \end{pmatrix}
  \right)^\omega
  \]
  is equal to $M'^{\omega,t}$, calculated as in~\eqref{eqn:omega_buchi_2}
  \[M'^{\omega,t}=
  \begin{pmatrix}
    \mu\\
    d^*
    \begin{pmatrix}
      j & c
    \end{pmatrix}\mu
  \end{pmatrix}\mkomma\]
  where
  \[\mu=\left(
  \begin{pmatrix}
    f & g\\ i & a
  \end{pmatrix} +
  \begin{pmatrix}
    h\\ b
  \end{pmatrix}d^*
  \begin{pmatrix}
    j & c
  \end{pmatrix}
  \right)^{\omega,t}\mpunkt\]

  In the case $t=k$, we have
  \[M=M'=\begin{pmatrix}
    f & h\\
    j & d
  \end{pmatrix}\mpunkt\]
  It follows that
  \begin{align*}
    \alpha &= (f+h d^* j)^\omega\\
    &=(f+h d^* j)^{\omega,t} = \mu\mkomma
  \end{align*}
  where the second equality is due to $t$ being the full dimension of $f+h d^* j$. The second components of $M^{\omega,t}$ and $M'^{\omega,t}$ then both reduce to $d^* j(f+h d^* j)^\omega$.

  If $k=n$, we have
  \[M=M'=\begin{pmatrix}
    f & g\\
    i & a
  \end{pmatrix}\mpunkt\]
  Now, the second component of $M'^{\omega,t}$ and the second summand of $\mu$ have dimension 0 and thus
  \[M'^{\omega,t} =
      \begin{pmatrix}
        f & g\\
        i & a
      \end{pmatrix}^{\omega,t}= M^{\omega,t}\]

  Hence, in the following, we can assume $t<k<n$.

  First, we compute $M^{\omega,t}$. We denote the blocks of $M^{\omega,t}$ by $(M^{\omega,t})_i$ for $1\leq i\leq 3$. Then we have
  \begin{align*}
    (M^{\omega,t})_1 =\alpha &= \left(f +
  \begin{pmatrix}
    g & h
  \end{pmatrix}
  \begin{pmatrix}
    a & b\\ c & d
  \end{pmatrix}^*
  \begin{pmatrix}
    i\\ j
  \end{pmatrix}
  \right)^\omega\\
  &= \left(f +
  \begin{pmatrix}
    g & h
  \end{pmatrix}
  \begin{pmatrix}
    (a+b d^* c)^* & a^* b(d+c a^* b)^*\\
    d^* c(a+b d^* c)^* & (d+c a^* b)^*
  \end{pmatrix}
  \begin{pmatrix}
    i\\ j
  \end{pmatrix}
  \right)^\omega\\
  &= \left(f +
  \begin{pmatrix}
    g & h
  \end{pmatrix}
  \begin{pmatrix}
    (a+b d^* c)^* i + a^* b(d+c a^* b)^* j\\
    d^* c(a+b d^* c)^* i + (d+c a^* b)^* j
  \end{pmatrix}
  \right)^\omega\\
  &= \big(f +
    g(a+b d^* c)^* i  +g a^* b(d+c a^* b)^* j\\
    &\hspace{1cm}+ h d^* c(a+b d^* c)^* i + h(d+c a^* b)^* j\big)^\omega\mpunkt
  \end{align*}
  Here, we used the star of a matrix in the form shown in~\eqref{matrix_star_2}.
  We will now compute the other two blocks by using the star of a matrix as in~\eqref{matrix_star_1}:
  \begin{align*}
    \begin{pmatrix}
      (M^{\omega,t})_2\\
      (M^{\omega,t})_3
    \end{pmatrix} &=
    \begin{pmatrix}
      a & b\\ c & d
    \end{pmatrix}^*
    \begin{pmatrix}
      i \\ j
    \end{pmatrix}
    \alpha\\
    &=
    \begin{pmatrix}
      (a+b d^* c)^* & (a + b d^* c)^* b d^*\\ (d+c a^* b)^* c a^* & (d+c a^* b)^*
    \end{pmatrix}
    \begin{pmatrix}
      i \\ j
    \end{pmatrix}
    \alpha\\
    &=\begin{pmatrix}
    (a+b d^* c)^* i + (a+b d^* c)^* b d^* j\\ (d+c a^* b)^* c a^* i + (d+c a^* b)^* j
    \end{pmatrix}
    \alpha\\
    &=\begin{pmatrix}
    \big((a+b d^* c)^* i + (a+b d^* c)^* b d^* j\big)\alpha\\
    \big((d+c a^* b)^* c a^* i + (d+c a^* b)^* j\big)\alpha
    \end{pmatrix}
  \end{align*}

  Now, we compute $M'^{\omega,t}$. We denote the blocks of $M'^{\omega,t}$ by $(M'^{\omega,t})_i$ for $1\leq i\leq 3$. Then we have
  \begin{align*}
    \begin{pmatrix}
      (M'^{\omega,t})_1\\
      (M'^{\omega,t})_2
    \end{pmatrix} = \mu &= \left(
  \begin{pmatrix}
    f & g\\ i & a
  \end{pmatrix} +
  \begin{pmatrix}
    h\\ b
  \end{pmatrix}d^*
  \begin{pmatrix}
    j & c
  \end{pmatrix}
  \right)^{\omega,t}\\
  &= \left(
  \begin{pmatrix}
    f & g\\ i & a
  \end{pmatrix} +
  \begin{pmatrix}
    h d^* j & h d^* c\\
    b d^* j & b d^* c
  \end{pmatrix}
  \right)^{\omega,t}\\
  &= \begin{pmatrix}
    f + h d^* j & g + h d^* c\\
    i + b d^* j & a + b d^* c
  \end{pmatrix}^{\omega,t}\\
  &= \begin{pmatrix}
    \delta\\
    (a + b d^* c)^* (i+b d^* j)\delta
  \end{pmatrix}\mkomma
  \end{align*}
  where
  \[\delta=\big(f + h d^* j + (g + h d^* c)(a + b d^* c)^*(i + b d^* j)\big)^\omega\mpunkt\]
  It remains to calculate
  \begin{align*}
    (M'^{\omega,t})_3 &= d^*
    \begin{pmatrix}
      j & c
    \end{pmatrix}\mu\\
    &= d^* \big( j + c(a + b d^* c)^* (i+b d^* j)\big)\delta\mpunkt
  \end{align*}

  The last step is to verify the three equalities $(M^{\omega,t})_i=(M'^{\omega,t})_i$ for $1\leq i\leq 3$.
  The first equality follows basically from Lemma~1.2.16 of~\cite{MAT}. We will mark the use of Lemma~1.2.16 by $\Diamond$ and obtain
  \begin{align*}
    (M^{\omega,t})_1 =\alpha &=\big(f + g(a+b d^* c)^* i  +g a^* b(d+c a^* b)^* j\\
    &\hspace{1cm}+ h d^* c(a+b d^* c)^* i + h(d+c a^* b)^* j\big)^\omega\\
    &\stackrel{\Diamond}{=}\big(f + h d^* j+ g(a + b d^* c)^* i + g(a+b d^* c)^* b d^* j\\
    &\hspace{1cm}+h d^* c(a+b d^* c)^* i+h d^* c(a+b d^* c)^* b d^* j\big)^\omega\\
    &=\big(f + h d^* j+ g(a + b d^* c)^*(i+b d^* j)+h d^* c(a+b d^* c)^*(i+b d^* j)\big)^\omega\\
    &=\big(f + h d^* j+ (g+h d^* c)(a + b d^* c)^*(i+b d^* j)\big)^\omega\\
    &=\delta=(M'^{\omega,t})_1
  \end{align*}
  For the second equality, we have
  \begin{align*}
    (M^{\omega,t})_2 &= \big((a+b d^* c)^* i + (a+b d^* c)^* b d^* j\big)\alpha\\
    &= \big((a+b d^* c)^* (i + b d^* j)\big)\delta\\
    &= (M'^{\omega,t})_2\mpunkt
  \end{align*}
  Now, for the third equality, it suffices to prove
  \[(d+c a^* b)^* c a^* i + (d+c a^* b)^* j = d^* \big( j + c(a + b d^* c)^* (i+b d^* j)\big)\mpunkt\]
  We have
  \begin{align*}
    d^* \big( j + c(a + b d^* c)^* (i+b d^* j)\big)
    &= d^* j + d^* c(a + b d^* c)^* (i+b d^* j)\\
    &= d^* j + d^* c(a + b d^* c)^* i + d^* c(a + b d^* c)^* b d^* j\\
    &= d^* j + d^* c(a^* b d^* c)^* a^* i + d^* c(a^* b d^* c)^* a^* b d^* j\\
    &= d^* j + (d^* c a^* b)^* d^* c a^* i + (d^* c a^* b)^* d^* c a^* b d^* j\\
    &= (d^* c a^* b)^* d^* c a^* i + (d^* c a^* b)^* d^* c a^* b d^* j + d^* j\\
    &= (d^* c a^* b)^* d^* c a^* i + \big((d^* c a^* b)^* d^* c a^* b + 1\big)d^* j\\
    &= (d^* c a^* b)^* d^* c a^* i + (d^* c a^* b)^* d^* j\\
    &= (d + c a^* b)^* c a^* i + (d + c a^* b)^* j\mpunkt
  \end{align*}
  Note that for this calculation, we rely heavily on commutativity of addition, distributivity and the sum star identity and the product star identity of Conway semirings together with their derived identities~\eqref{eqn:conway_derived}.
  This completes the proof.
\end{proof}

%%%%%%%%%%%%%%%%%%%%%%%%%%%%%%%%%%%%%%%%%%%%%%%%%%%%%%%%%%%%%%%%%%%%%%

%% From triple-pair:
%% For the definition of an $S$-algebraic system over a quemiring  $S \times V$
%% we refer the reader to Ésik, Kuich~\cite{MAT}, page 136, and
For a complete definition of quemirings, we refer the reader to~\cite{MAT}, page 110.
Here we note that a quemiring $T$ is isomorphic to a quemiring $S \times V$ determined by the semiring-semimodule pair $(S,V)$; it follows that we can identify every element $t$ of a quemiring $T$ by a pair $(s,v)$ of a semiring-semimodule pair $(S,V)$. A quemiring is an algebraic structure with an addition
given componentwise, i.e.,
\[(s,v)+(s',v') = (s+s',v+v')\mkomma\]
a semidirect product type multiplication (using that $S$ acts on $V$), i.e.,
\[(s,v)\cdot(s',v') = (s s',v+s v')\mkomma\]
and two constants $0=(0,0)$ and $1=(1,0)$ (and a unary operation $\P$, but we will not use it here). A quemiring $S\times V$ satisfies a set of axioms inherited from semiring-semimodule pairs; those axioms make a quemiring \emph{quasi a semiring} (cf.\@ Elgot~\cite{quemirings}, Ésik, Kuich~\cite{MAT}, page 109; in fact, a quemiring is not necessarily distributive from the left and 0 only behaves like a zero from the left).
Also, one can define a natural star operation on $S\times V$, i.e.,
\[(s,v)^\otimes = (s^*,s^\omega+s^* v)\mkomma\]
making it a \emph{generalized starquemiring}, see~\cite{MAT}.

%%%%%%%%%%%%%%%%%%%%%%%%%%%%%%%%%%%%%%%%%%%%%%%%%%%%%%%%%%%%%%%%%%%%%%

For an alphabet $\Sigma$, we call mappings $r$ of $\Sigma^*$ into $S$ \emph{series}. The collection of all such series $r$ is denoted by $S\langle\langle\Sigma^*\rangle\rangle$. We call the set $\text{supp}(r)=\{w \mid (r,w)\neq 0\}$ the \emph{support} of a series $r$. The set of series with finite support $S\langle\Sigma^*\rangle=\{s\in \ssigstar\mid \text{supp}(s) \text{ is finite}\}$ is called the set of \emph{polynomials}. We denote by $S\langle\Sigma\rangle$, $S\langle\{\epsilon\}\rangle$ and $S\langle\Sigma\cup\{\epsilon\}\rangle$ the series with support in $\Sigma$, $\{\epsilon\}$ and $\Sigma\cup\{\epsilon\}$, respectively. Series $s$ with $|\text{supp}(s)|\leq 1$ are called monomials. Note that polynomials are finite sums of monomials.

Mappings of $\Sigma^\omega$ into $S$ are called \emph{$\omega$-series} and their collection is denoted by $\ssigomega$. See~\cite{kuich_formal_power_series,MAT} for more information.
Examples of monomials in $S\langle\Sigma^*\rangle$ for a semiring $\langle S, +, \cdot, 0, 1\rangle$ are $0$, $w$, $s w$ for $s\in S$ and $w\in\Sigma^*$, defined by\\
\indent $(0,w) = 0 \text{ for all }w,$\\
\indent $(w,w) = 1 \text{ and } (w,w')=0\text{ for }w\neq w',$\\
\indent $(s w,w) = s\text{ and } (s w,w')=0\text{ for }w\neq w'.$

%%%%%%%%%%%%%%%%%%%%%%%%%%%%%%%%%%%%%%%%%%%%%%%%%%%%%%%%%%%%%%%%%%%%%%%%%%%%%%%%%%%%%%
%%%%%%%%%%%%%%%%%%%%%%%%%%%%%%%%%%%%%%%%%%%%%%%%%%%%%%%%%%%%%%%%%%%%%%%%%%%%%%%%%%%%%%
%%%%%%%%%%%%%%%%%%%%%%%%%%%%%%%%% Greibach normal form
%%%%%%%%%%%%%%%%%%%%%%%%%%%%%%%%%%%%%%%%%%%%%%%%%%%%%%%%%%%%%%%%%%%%%%%%%%%%%%%%%%%%%%
%%%%%%%%%%%%%%%%%%%%%%%%%%%%%%%%%%%%%%%%%%%%%%%%%%%%%%%%%%%%%%%%%%%%%%%%%%%%%%%%%%%%%%

\section{\texorpdfstring{$\omega$}{omega}-Algebraic Systems}
\label{sec:mixed_algebraic_systems}

This and the next two sections describe the Greibach normal form for (mixed) $\omega$-algebraic systems. Here, we define $\omega$-algebraic systems and mixed $\omega$-algebraic systems.

For this section and the next two sections, Sections~\ref{sec:mixed_algebraic_systems}, \ref{sec:greibach_normal_form} and \ref{sec:greibach_normal_form_unmixed}, \emph{$S$ is a continuous, and therefore complete, star-omega semiring}. Let further $\Sigma$ denote an alphabet. If we consider $\ssigstar$ or $\ssigomega$, then we assume additionally that \emph{the underlying semiring $S$ is commutative}.

By Theorem~5.5.5 of Ésik, Kuich~\cite{MAT}, $\spair$ is a complete semi\-ring-semimodule pair, hence a Conway semiring-semimodule pair, satisfying $\epsilon^\omega=0$. Hence, $\stimes$ is a generalized starquemiring.

In the sequel, $x$, $y$ and $z$ denote vectors of dimension $n$, i.e., $x=(x_1,\ldots,x_n)$, $y=(y_1,\ldots,y_n)$ and $z=(z_1,\ldots,z_n)$. Later, we will also use $z$ of dimension $m$. It is clear by the context whether they are used as row or as column vectors. Similar conventions hold for vectors $p$, $\sigma$, $\omega$ and $\tau$. Moreover, $X$ denotes the set of variables $\{x_1,\ldots,x_n\}$ for $\ssigstar$, while $\{z_1,\ldots,z_n\}$ is the set of variables for $\ssigomega$. The set $Y$ denotes the set of variables $\{y_i,\ldots,y_n\}$ for the quemiring $\stimes$.

We will be working with two different generalizations of $\omega$-context-free grammars, the $\omega$-algebraic systems and the mixed $\omega$-algebraic systems.
Both representations model $\omega$-algebraic series, i.e., weighted $\omega$-context-free languages. The $\omega$-algebraic systems look similar to $\omega$-context-free grammars and are, therefore, of interest.
The mixed $\omega$-algebraic systems distinguish between variables for finite word solutions and variables for infinite word solutions. This division allows us to define and describe canonical solutions that behave similarly to unweighted $\omega$-context-free grammars.
On the other hand, canonical solutions of $\omega$-algebraic systems are derived by first transforming the $\omega$-algebraic systems into mixed $\omega$-algebraic systems.
Thus, this double presentation is needed because we finally work with mixed $\omega$-algebraic systems, but we want to describe series by $\omega$-algebraic systems in the first place.

An \emph{$\omega$-algebraic system over the quemiring} $\stimes$ consists of an algebraic system over $\stimes$
\[y=p(y),\hspace{0.5cm}p\in(S\langle(\Sigma\cup Y)^*\rangle)^{n\times 1}\mpunkt\]
The vector of quemiring elements $\tau\in(\stimes)^n$ is a \emph{solution} of the $\omega$-algebraic system
\[y = p(y)\mkomma\]
if
\[\tau = p(\tau)\mpunkt\]

Note that every $p_i$ is a polynomial, i.e., a finite sum of monomials in $S\langle(\Sigma\cup Y)^*\rangle$. Let $y_i=(x_i,z_i)$, for $1\leq i \leq n$. Now, we can apply the quemiring addition and multiplication to $p$.

Consider a monomial
\[t(y_1,\ldots,y_n)=s w_0 y_{i_1}w_1 \ldots w_{k-1}y_{i_k}w_k\mkomma\]
where $s\in S$ and $w_i\in \Sigma^*$ for $1\leq i\leq k$.
Note that from the quemiring operations, we have
\[t((x_1,z_1),\ldots,(x_n,z_n))=(s w_0 x_{i_1}w_1 \ldots w_{k-1}x_{i_k}w_k,\;s w_0 z_{i_1} + s w_0 x_{i_1} w_1 z_{i_2} + \ldots + s w_0 x_{i_1}w_1 \cdots w_{k-2}x_{i_{k-1}}w_{k-1}z_{i_k})\mpunkt\]
Therefore, following Ésik, Kuich~\cite{MAT}, p.\@ 138, we define
\[t_x(x_1,\ldots,x_n,z_1,\ldots,z_n) = s w_0 z_{i_1} + s w_0 x_{i_1} w_1 z_{i_2} + \ldots + s w_0 x_{i_1}w_1 \cdots w_{k-2}x_{i_{k-1}}w_{k-1}z_{i_k}\mkomma\]
and for a polynomial $p(y_1,\ldots,y_n)=\sum_{1\leq j\leq m}t_j(y_1,\ldots,y_n)$, we let
\[p_x(x_1,\ldots,x_n,z_1,\ldots,z_n)=\sum_{1\leq j\leq m}(t_j)_x(x_1,\ldots,x_n,z_1,\ldots,z_n)\mpunkt\]

For an $\omega$-algebraic system $y=p(y)$ over $\stimes$, we call $x=p(x)$, $z=p_x(x,z)$ the \emph{mixed $\omega$-algebraic system over} $\stimes$ \emph{induced by} $y= p(y)$.

In general, a \emph{mixed $\omega$-algebraic system} over the quemiring $\stimes$ consists of an algebraic system over $\ssigstar$
\[x=p(x),\hspace{0.5cm}p\in(S\langle(\Sigma\cup X)^*\rangle)^{n\times 1}\]
and a linear system over $\ssigomega$
\[z=\varrho(x)z,\hspace{0.5cm}\varrho\in(S\langle(\Sigma\cup X)^*\rangle)^{m\times m}\mpunkt\]
The pair $(\sigma,\omega)\in(\ssigstar)^n\times(\ssigomega)^m$ is a \emph{solution} of the mixed $\omega$-algebraic system
\[x=p(x),\;\;z=\varrho(x)z\mkomma\]
if
\[\sigma=p(\sigma),\;\;\omega=\varrho(\sigma)\omega\mpunkt\]
Observe that, by Theorem~5.5.1 of Ésik, Kuich~\cite{MAT}, $\omega^{(k)}=\varrho(\sigma)^{\omega,k}$ for each $1\leq k\leq n$, is solution for the linear system
\[z=\varrho(\sigma)z\mpunkt\]
A solution $(\sigma_1,\ldots,\sigma_n)$ of the algebraic system $x=p(x)$ is termed \emph{least solution} if
\[\sigma_i \leq \tau_i,\hspace{.5cm}\text{for each }1\leq i\leq n,\]
for all solutions $(\tau_1,\ldots,\tau_n)$ of $x=p(x)$.

% p. 3
If $\sigma$ is the least solution of $x=p(x)$, then $z=\varrho(\sigma)z$ is an $\salg$-linear system and $(\sigma,\omega^{(k)})=(\sigma,\varrho(\sigma)^{\omega,k})$, where $k\in\{0,1,\ldots,m\}$, is called $k$\textsuperscript{th} \emph{canonical solution} of $x=p(x)$, $z=\varrho(x)z$. 
Observe that the $k$\textsuperscript{th} canonical solution is unique by definition.
A solution $(\sigma,\omega)$ is called $canonical$, if there exists a $k$ such that $(\sigma,\omega)$ is the $k$\textsuperscript{th} canonical solution. The $k$\textsuperscript{th} canonical solution of an $\omega$-algebraic system $y=p(y)$ is defined to be the $k$\textsuperscript{th} canonical solution of the mixed $\omega$-algebraic system $x=p(x)$, $z=p_x(x,z)$ induced by $y=p(y)$.

Recall that $\salg$ is the collection of \emph{algebraic series}, i.e., of all components of least solutions of algebraic systems
\[x_i=p_i\hspace{0.5cm}\text{where }p_i\in S\langle(\Sigma\cup X)^*\rangle\text{ for }1\leq i\leq n\mpunkt\]
We define $\salgomega$ to be the collection of all components of vectors $M^{\omega,k}$, where $M\in(\salg)^{n\times n}$, $n\geq 1$, and $k\in\{1,\ldots,n\}$ and call it the collection of \emph{$\omega$-algebraic series}.
%%% Include when including old Thm. 5:
%%Moreover, we define $\Rat$ to be the $\omega$-Kleene closure of (i.e., the generalized starquemiring generated by) $\salg$.

\begin{example}
\label{ex:algebraic_system}
  We consider the following $\omega$-algebraic system over the quemiring $\B \langle\langle\Sigma^*\rangle\rangle\times \B \langle\langle\Sigma^\omega\rangle\rangle$ for the Boolean semiring $\langle \B,+,\cdot,0,1\rangle$
  \begin{align*}
    y_1 &= y_2 y_1 + \epsilon\\
    y_2 &= a y_2 b + \epsilon\mkomma
  \end{align*}
  where $a,b\in\Sigma$.
  This induces the following mixed $\omega$-algebraic system
  \begin{alignat*}{4}
    x_1 &= x_2 x_1 + \epsilon & z_1 &= z_2 + x_2 z_1\\
    x_2 &= a x_2 b + \epsilon  \hspace{2cm} & z_2 &= a z_2\mpunkt
  \end{alignat*}

  Then for the algebraic system $x=p(x)$ over $\B \langle\langle\Sigma^*\rangle\rangle$, we get the least solution $\sigma_2=\sum_{n\geq0}a^n b^n$ and therefore $\sigma_1=(\sum_{n\geq0}a^n b^n)^*$.
For the semimodule part, we can consider the first canonical solution where only $z_1$ is Büchi-accepting and the second canonical solution where both $z_1$ and $z_2$ are Büchi-accepting.
The first canonical solution of the mixed $\omega$-algebraic system $x=p(x)$, $z=p_x(x,z)$ over $\B \langle\langle\Sigma^*\rangle\rangle\times\B \langle\langle\Sigma^\omega\rangle\rangle$ is then $(\sigma_1,\sigma_2;(\sum_{n\geq0}a^n b^n)^\omega,0)$. The second canonical solution would be $(\sigma_1,\sigma_2;(\sum_{n\geq0}a^n b^n)^\omega+(\sum_{n\geq0}a^n b^n)^*a^\omega,a^\omega)$.
%Hence the series $a^n b^n c^\omega \mapsto n$ is $\omega$-algebraic but it is clearly not recognizable by a weighted automaton without stack.
\end{example}
\begin{example}
\label{ex:mixed_algebraic_system}
  We consider the following mixed $\omega$-algebraic system over the quemiring $\N^\infty \langle\langle\Sigma^*\rangle\rangle\times \N^\infty \langle\langle\Sigma^\omega\rangle\rangle$ for the tropical semiring $\langle \N^\infty,\min,+,\infty,0\rangle$
  \begin{alignat*}{4}
    x_1 &= 1a x_1 b + 1 a b \hspace{2cm} & z_1 &= c z_1\\
    & & z_2 &= x_1 z_1 +z_1
  \end{alignat*}
  where $a,b,c\in\Sigma$ and using the natural number $1$.

  Then for the algebraic system $x=p(x)$ over $\N^\infty \langle\langle\Sigma^*\rangle\rangle$, we get the least solution $\sigma=a^n b^n\mapsto n$ for $n\geq 1$.
The first canonical solution of the mixed $\omega$-algebraic system $x=p(x)$, $z=\varrho(x)z$ over $\N^\infty \langle\langle\Sigma^*\rangle\rangle\times\N^\infty \langle\langle\Sigma^\omega\rangle\rangle$ is then $(\sigma,c^\omega\mapsto 0, a^n b^n c^\omega \mapsto n)$ for $n\geq 0$. Hence the series $a^n b^n c^\omega \mapsto n$ is $\omega$-algebraic but it is clearly not recognizable by a weighted automaton without stack. Note for the sake of completeness that in this particular example, the second canonical solution is identical to the first because no other infinite paths are possible.
\end{example}

%% Now we recall the following known characterization for $\omega$-algebraic series for later use. 

%% \begin{theorem}[Ésik, Kuich~\cite{MAT}]
%%     Let $S$ be a continuous complete star-omega semi\-ring with the underlying semiring $S$ being commutative and let $\Sigma$ be an alphabet. Then the following statements are equivalent for $(s,\upsilon)\in\stimes$:
%%   \begin{enumerate}[(i)]
%%   %\item $(s,\upsilon)\in\salg\times\salgomega$,\label{i}
%%   \item $(s,\upsilon)\in\Rat$,\label{ii_old}
%%   \item $(s,\upsilon)=\|\mathfrak{A}\|$, where $\mathfrak{A}$ is a finite $\salg$-automaton over $\stimes$,\label{iii_old}
%%   \item $s\in\salg$ and $\upsilon=\sum_{1\leq j\leq l} s_j t_j^\omega$ for some $l\geq 0$, where $s_j, t_j\in \salg$,\label{iv_old}
%%   %\item $(s,\upsilon)$ is a component of the automata-theoretic solution of an $\salg$-linear system over $\stimes$,\label{v}
%%   \end{enumerate}
%% \end{theorem}
%% \begin{proof}
%%   %The statements \eqref{ii_old}, \eqref{iii_old} and \eqref{iv_old} are equivalent by Theorem~5.4.9 (see also Theorem~5.6.6) of Ésik, Kuich~\cite{MAT}.
%% %The statements \eqref{ii_old}, \eqref{iii_old} and \eqref{iv_old} are equivalent 
%%   This follows by Theorem~5.4.9 (see also Theorem~5.6.6) of Ésik, Kuich~\cite{MAT}.
%% \end{proof}

%%Now we extend this characterization as follows.

Now we have the following characterization of $\omega$-algebraic series.

\begin{theorem}
  \label{thm:3.1}
  Let $S$ be a continuous complete star-omega semi\-ring with the underlying semiring $S$ being commutative and let $\Sigma$ be an alphabet. Then the following statements are equivalent for $(s,\upsilon)\in\stimes$:
  \begin{enumerate}[(i)]
  \item $(s,\upsilon)\in\salg\times\salgomega$,\label{i}
  %\item $(s,\upsilon)\in\Rat$,\label{ii}
  %\item $(s,\upsilon)=\|\mathfrak{A}\|$, where $\mathfrak{A}$ is a finite $\salg$-automaton over $\stimes$,\label{iii}
  \item $s\in\salg$ and $\upsilon=\sum_{1\leq j\leq l} s_j t_j^\omega$ for some $l\geq 0$, where $s_j, t_j\in \salg$,\label{iv}
  %\item $(s,\upsilon)$ is a component of the automata-theoretic solution of an $\salg$-linear system over $\stimes$,\label{v}
  \item $(s,\upsilon)$ is a component of a canonical solution of a mixed $\omega$-algebraic system over $\stimes$.\label{vi}
  \end{enumerate}
  %\TODO{can this be simplified? The reviewer wants to have (ii), (iii) and (iv) removed. At least, (ii) is really not needed here. Should we introduce automata before using them in \eqref{iii}?}
\end{theorem}

% the first two are also used later
\newcommand{\dmatrix}[1]{\begin{pmatrix}#1_1 & & \\& \ddots & \\& & #1_l\end{pmatrix}}
\newcommand{\rmatrix}[1]{\begin{pmatrix}#1_1 & \cdots & #1_l\end{pmatrix}}

\newcommand{\dmatrixx}[1]{\begin{matrix}#1_1 & & \\& \ddots & \\& & #1_l\end{matrix}}
\newcommand{\dmatrixxx}[1]{\begin{matrix}#1_{l+1} & & \\& \ddots & \\& & #1_{l+l}\end{matrix}}

\begin{proof}
  %The statements \eqref{ii}, \eqref{iii} and \eqref{iv} are equivalent by Theorem~5.4.9 (see also Theorem~5.6.6) of Ésik, Kuich~\cite{MAT}.

  %\eqref{iii}$\Rightarrow$\eqref{vi}: Assume that $(s,\upsilon)=\|\mathfrak{A}\|$, where $\mathfrak{A}=(n,I,M,P,k)$ is a finite $\salg$-automaton. Without loss of generality $\mathfrak{A}$ is normalized by Theorem~5.4.2 of Ésik, Kuich~\cite{MAT}; i.e., $I=e_i$ for some $i$.
  %Hence, $(s,\upsilon)=((M^* P)_i,(M^{\omega,k})_i)$ is a component of the $k$\textsuperscript{th} canonical solution of the mixed $\omega$-algebraic system
  %\[x=M x+P,\;\;z=M z\mpunkt\]

  %% \eqref{iii}$\Rightarrow$\eqref{v}: Assume that $(s,\upsilon)=\|\mathfrak{A}\|$, where $\mathfrak{A}=(n,I,M,P,k)$ is a finite $\salg$-automaton. Without loss of generality $\mathfrak{A}$ is normalized by Theorem~5.4.2 of Ésik, Kuich~\cite{MAT}; i.e., $I=e_i$ for some $i$. Consider now the $\salg$-linear system $y=M y+P$. Then $(M^* P,M^{\omega,k})$ is its $k$\textsuperscript{th} automata-theoretic solution and $(s,\upsilon)=((M^* P)_i,(M^{\omega,k})_i)$.

  %% \eqref{v}$\Rightarrow$\eqref{vi}: Assume $(s,\upsilon)$ is a component of an automata-theoretic solution of the $\salg$-linear system $y=M y+P$, where the entries of $M$ and $P$ are in $\salg$. The $k$\textsuperscript{th} automata-theoretic solution of this linear system is given by $(M^* P, M^{\omega,k})$. But $(M^* P, M^{\omega,k})$ is also a canonical solution of the mixed algebraic system
  %% \[x=M x+P,\;\;z=M z\mpunkt\]

  \eqref{vi}$\Rightarrow$\eqref{i}: Assume there exists a mixed $\omega$-algebraic system $x=p(x), z=\varrho(x)z$, with canonical solution $(\sigma, \varrho(\sigma)^{\omega,k})$ such that $(s,\upsilon)=(\sigma_i,(\varrho(\sigma)^{\omega,k})_j)$ for some $i$ and $j$. Since the entries of $\sigma$ and $\varrho(\sigma)$ are in $\salg$, $(s,\upsilon)$ is in $\salgtimes$.

  \eqref{i}$\Rightarrow$\eqref{iv}: Now assume $s\in\salg$ and $\upsilon=(M^{\omega,k})_i$ for some $M\in(\salg)^{n\times n}$, $n\geq 1$, and $i,k\in\{1,\ldots,n\}$. By the definition of $M^{\omega,k}$, each entry of $M^{\omega,k}$ is of the form $\sum_{1\leq j\leq l}s_j t_j^{\omega}$ for some $l\geq 0$, where $s_j, t_j\in\salg$ for $1\leq j\leq l$.

  \eqref{iv}$\Rightarrow$\eqref{vi}: As $s_j, t_j\in \salg$, we can assume that there exist algebraic systems such that all series $s_j$ and $t_j$ ($1\leq j\leq l$) are components of least solutions of one of the algebraic systems. We additionally assume that their variables are distinct and write all algebraic systems together into one algebraic system over $\ssigstar$
  \begin{equation}
    \label{eqn:charact_finite}
    x=p(x)\mkomma
  \end{equation}
  and we order the variables such that for $1\leq j\leq l$, we have that the $j$\textsuperscript{th} component of its least solution is $s_j$ and the $(l+j)$\textsuperscript{th} component of its least solution is $t_j$.

  %% We can assume that for all $1\leq j\leq l$, there exists an algebraic system
  %% \[x^{(j)}_i=p^{(j)}_i(x^{(j)})\mkomma \hspace{2cm}\text{for}\;1\leq i\leq n^{(j)}\mkomma\]
  %% where the first component of the least solution equals to $s_j$. Equally, we assume that for all $1\leq j\leq l$, there exists
  %% \[x'^{(j)}_i=p'^{(j)}_i(x'^{(j)})\mkomma \hspace{2cm}\text{for}\;1\leq i\leq m^{(j)}\mkomma\]
  %% where the first component of the least solution equals to $t_j$.
  %% Now consider the mixed $\omega$-algebraic system consisting of the equations above together with
  %% \begin{align*}
  %%     z_1 &= x'^{(1)}_1 z_1\\
  %%     &\;\;\vdots\\
  %%     z_l &= x'^{(l)}_1 z_l\\
  %%     z_{l+1} &= \sum_{1\leq  j\leq l} x^{(j)}_1 z_j
  %% \end{align*}
  %% The $(l+1)$\textsuperscript{th} component of the $l$\textsuperscript{th} canonical solution of our mixed $\omega$-algebraic system is $\sum_{1\leq j\leq l}s_j t_j^\omega$.

  Now consider the linear system over $\ssigomega$
  \begin{equation}
    \label{eqn:charact_infinite}
    \begin{split}
      z_1 &= x_{l+1} z_1\\
      &\;\;\vdots\\
      z_l &= x_{l+l} z_l\\
      z_{l+1} &= \sum_{1\leq  j\leq l} x_j z_j\mpunkt
    \end{split}
  \end{equation}
  %The last component of the $l$\textsuperscript{th} canonical solution of our mixed $\omega$-algebraic system \eqref{eqn:charact_finite}, \eqref{eqn:charact_infinite} is $\upsilon=\sum_{1\leq j\leq l}s_j t_j^\omega$.

  %\TODO{Ist das schon genug Beweis oder sollen wir die Details auch aufschreiben? $\rightarrow$ \ref{sec:details}}.

  %We infer the $l$\textsuperscript{th} canonical solution of \eqref{eqn:charact_finite}, \eqref{eqn:charact_infinite}.
  We now show that the last component of the $l$\textsuperscript{th} canonical solution of our mixed $\omega$-algebraic system \eqref{eqn:charact_finite}, \eqref{eqn:charact_infinite} is $\upsilon=\sum_{1\leq j\leq l}s_j t_j^\omega$.

  By assumption, we know the first $2l$ components of the least solution of \eqref{eqn:charact_finite}, i.e.,
  \[\sigma=(s_1,\ldots,s_l,t_1,\ldots,t_l,\sigma_{2l+1},\ldots)\mpunkt\]
  
  Now, we write \eqref{eqn:charact_infinite} as $z = M(x) z$ where
  \[M(x)=\left(
  \begin{array}{c|c}
    \dmatrixxx{x} & 0\\
    \hline\\[-1.5\medskipamount]
    \begin{matrix}
      x_1 & \cdots & x_l
    \end{matrix} &
    0
  \end{array}
  \right)\mpunkt\]
  We have
  \begin{align*}
    M(\sigma)^{\omega,l}_{l+1} &= \left(
    \begin{array}{c|c}
      \dmatrixx{t} & 0\\
      \hline\\[-1.5\medskipamount]
      \begin{matrix}
        s_1 & \cdots & s_l
      \end{matrix} &
      0
    \end{array}
    \right)^{\omega,l}_{l+1}\\
    &=
    \begin{pmatrix}
      \dmatrix{t}^\omega\\
      0^*(s_1,\ldots,s_l)\dmatrix{t}^\omega
    \end{pmatrix}_{l+1}\\
    &=(s_1,\ldots,s_l)
    \begin{pmatrix}
      t_1^\omega\\ \vdots\\ t_l^\omega
    \end{pmatrix}\\
    &= \sum_{1\leq j\leq l}s_j t_j^\omega\mpunkt\qedhere
  \end{align*}
\end{proof}

%%%%%%%%%%%%%%%%%%%%%%%%%%%%%%%%%%%%%%%%%%%%%%%%%%%

\section{Greibach Normal Form for Mixed \texorpdfstring{$\omega$}{omega}-Algebraic Systems}
\label{sec:greibach_normal_form}

In this section we show that for any element $(s,\upsilon)$ of $\salgtimes$ there exists a mixed $\omega$-algebraic system in Greibach normal form such that $(s,\upsilon)$ is a component of a solution of this $\omega$-algebraic system. We start by showing this property for \emph{mixed} $\omega$-algebraic systems because Theorem~\ref{thm:3.1}~\eqref{iv} gives us a powerful tool but only for separate $s$ and $\upsilon$, thus we construct equations for $s$ and equations for $\upsilon$ separately---a \emph{mixed} $\omega$-algebraic system.

Similar to the definition for algebraic systems on finite words (cf.\@ also Greibach~\cite{greibach}), a mixed $\omega$-algebraic system
\[x=p(x),\;\; z=\varrho(x)z\]
is in \emph{Greibach normal form} if
\begin{alignat*}{2}
&\text{supp}(p_i(x)) \subseteq \{\epsilon\}\cup\Sigma\cup \Sigma X\cup \Sigma X X,\qquad &&\text{for all }1\leq i\leq n,\hspace{.5cm} \text{ and}\\
&\text{supp}(\varrho_{i j}(x)) \subseteq \Sigma\cup \Sigma X,\;\; &&\text{for all }1\leq i,j\leq m\mpunkt
\end{alignat*}

For the construction of the Greibach normal form we need a corollary to Theorem~\ref{thm:3.1} specializing statement \eqref{iv}.
\begin{corollary}
  \label{cor:4.1}
  The following statement for $(s,\upsilon)\in\stimes$ is equivalent to the statements \eqref{i} to \eqref{vi} of Theorem~\ref{thm:3.1}:\\
$s\in\salg$ and $\upsilon=\sum_{1\leq j\leq l}s_j t_j^\omega$ for some $l\geq 0$, where $s_j,t_j\in\salg$ with $(t_j,\epsilon)=0$; moreover $(s_j,\epsilon)=0$ or $s_j=(s_j,\epsilon)\epsilon$.
\end{corollary}

\begin{proof}
  Assume $(s_j,\epsilon)\neq 0$. Then $s_j=(s_j,\epsilon)\epsilon+s_j'$ where $(s'_j,\epsilon)=0$, and $s_j t_j^\omega =(s_j,\epsilon)t_j^\omega + s_j' t_j^\omega$.

Assume $(t_j,\epsilon)\neq 0$. Then $t_j=(t_j,\epsilon)\epsilon+t_j'$, where $(t_j',\epsilon)=0$. Since $\spair$ is a Conway semiring-semimodule pair satisfying $\epsilon^\omega=0$, we obtain $t_j^\omega=((t_j,\epsilon)^* \epsilon^* t_j')^\omega$ with $(t_j,\epsilon)^*\epsilon^* t_j'\in \salg$, since $((t_j,\epsilon)\epsilon)^\omega = (t_j,\epsilon)^\omega \epsilon^\omega = 0$.
\end{proof}

We now assume that $(s,\upsilon)\in\salg\times \ssigomega$ is given in the form of Corollary~\ref{cor:4.1} with $l=1$. By Theorem~2.4.10 of Ésik, Kuich~\cite{MAT}, there exist algebraic systems in Greibach normal form whose first component of their least solutions equals $s_1$, $t_1$.

Firstly, we deal with the case $(s_1,\epsilon)=0$. Let
\begin{align}
x_i=p_i(x)+\sum_{1\leq j\leq n} p_{i j}(x) x_j,\;\;\text{for each } 1\leq i\leq n,\label{*}\tag{$\ast$}
\end{align}
where $\text{supp}(p_i(x))\subseteq\Sigma\cup \Sigma X$, $\text{supp}(p_{i j}(x))\subseteq\Sigma X$, be the algebraic system in Greibach normal form for $s_1$ and
\begin{align}
  x_i'=p_i'(x')+\sum_{1\leq j\leq m} p'_{i j}(x') x'_j,\;\;\text{for each } 1\leq i\leq m,\label{**}\tag{$\ast\ast$}
\end{align}
where $\text{supp}(p_i'(x'))\subseteq\Sigma\cup \Sigma X'$, $\text{supp}(p_{i j}(x'))\subseteq\Sigma X'$, be the algebraic system in Greibach normal form for $t_1$. Let $\sigma$ and $\sigma'$ with $\sigma_1=s_1$ and $\sigma'_1=t_1$ be the least solutions of (\ref{*}) and (\ref{**}), respectively.

%% solution
Consider now the mixed $\omega$-algebraic system consisting of the algebraic system (\ref{*}), (\ref{**}) over $\ssigstar$ and the linear system over $\ssigomega$
\begin{equation*}
  \begin{alignedat}{2}
    z'' &= p_1'(x')z'' + \sum_{1\leq j\leq m} p'_{1 j}(x') z_j'\mkomma\\
    z_i' &= p_i'(x')z''+\sum_{1\leq j\leq m} p'_{i j}(x') z_j',\hspace{.5cm} &&\text{for }1\leq i\leq m\mkomma\\
    z_i &= p_i(x) z''+\sum_{1\leq j\leq n} p_{i j}(x) z_j,\;&&\text{for }1\leq i\leq n\mpunkt
  \end{alignedat}\label{***}\tag{$\ast\ast\ast$}
\end{equation*}
%% does this solve the problem?
%\refstepcounter{equation}

Observe that the mixed $\omega$-algebraic system is in Greibach normal form. We then order the variables of the mixed $\omega$-algebraic system (\ref{*}), (\ref{**}), (\ref{***}) as $x_1,\ldots, x_n;\allowbreak x_1',\ldots,x_{m}';z'';\allowbreak z_1',\ldots,z'_{m};\allowbreak z_1,\ldots,z_n$.
After an example, we will prove that
\begin{align}
  (\sigma_1,\ldots,\sigma_n;\sigma'_1,\ldots,\sigma'_m;\sigma_1'\sigma_1'^{\omega};\sigma_1'\sigma_1'^\omega,\ldots,\sigma_{m}'\sigma_1'^\omega;\sigma_1\sigma_1'^\omega,\ldots,\sigma_n\sigma_1'^\omega)\label{1}
\end{align}
is a canonical solution of (\ref{*}), (\ref{**}), (\ref{***}). Observe that $\sigma_1'\sigma_1'^\omega=\sigma_1'^\omega$.

\begin{example}
\label{ex:finite_simple}
Consider the quemiring $\N^\infty\! \langle\langle\Sigma^*\rangle\rangle\times \N^\infty\! \langle\langle\Sigma^\omega\rangle\rangle$ for the tropical semiring $\langle \N^\infty\!,\min,+,\infty,0\rangle$. Note that subsequently, $1$ stands for the natural number $1$ and the neutral element of the semiring multiplication is $\mathone=0$.

  We now define algebraic systems in Greibach normal form for $s=a^n b^n\mapsto n$ and $t=((d d)^* c)\mapsto 0$. Let
  \begin{alignat*}{4}
    x_1 &= 1 a x_2 + 1a x_1 x_2\hspace{2cm} & x_1' &= c + d x_2' x_1'\\
    x_2 &= b                                & x_2' &= d
  \end{alignat*}
  Here, $x_1$ is the start variable for $s$ and $x_1'$ is the start variable for $t$. In the proof, these two systems are called~\eqref{*} and \eqref{**}. Now, we construct a mixed $\omega$-algebraic system:
  \begin{alignat*}{4}
    z'' &= c z''+ d x_2' z_1'&&\\
    z_1' &= c z''+ d x_2' z_1' & z_2' &= d z''\\
    z_1 &= 1a x_2 z'' + 1 a x_1 z_2\hspace{1cm} & z_2 &= b z''
  \end{alignat*}
  In the new system (corresponding to~\eqref{***}), variable $z''$ is Büchi-accepting and variable $z_1$ acts as the start variable, i.e., we consider the fourth component (with the ordering $z'', z_1', z_2', z_1, z_2$) of the first canonical solution. The semimodule part of the solution is $s t^\omega= a^n b^n ((dd)^*c)^\omega\mapsto n$. Note that the equation for $z''$ is needed in this example because $z_1'$ is not allowed to be Büchi-accepting to prevent $(dd)^\omega$ as part of the canonical solution.
\end{example}

%% removed (noncanonical) solution here

\begin{lemma}
  \label{lem:canonical_1}
  The tuple (\ref{1}) is the first canonical solution of the mixed $\omega$-algebraic system (\ref{*}), (\ref{**}), (\ref{***}).
\end{lemma}
\begin{proof}
  Let
  \begin{alignat*}{4}
    P'_{1 m}(x') &=\begin{pmatrix}p_{1 1}'(x') &\cdots &p_{1 m}'(x')\end{pmatrix}, \hspace{0.4cm}\\
    P'_{m 1}(x') &= \begin{pmatrix}p_1'(x')\\\vdots\\p'_{m}(x')\end{pmatrix}, &
    P'_{m m}(x') &=\begin{pmatrix}p'_{1 1}(x') & \ldots & p'_{1 m}(x')\\\vdots & & \vdots\\p'_{m 1}(x') & \ldots &p'_{m m}(x')\end{pmatrix},\\
    P_{n 1}(x) &= \begin{pmatrix}p_1(x)\\\vdots\\p_n(x)\end{pmatrix}, &
    P_{n n}(x) &= \begin{pmatrix}p_{1 1}(x) & \ldots & p_{1 n}(x)\\\vdots & & \vdots\\p_{n 1}(x) & \ldots &p_{n n}(x)\end{pmatrix},\\
    z &= \begin{pmatrix}z_1\\\vdots\\ z_n\end{pmatrix},&
    z' &= \begin{pmatrix}z_1'\\\vdots\\ z_{m}'\end{pmatrix},
  \end{alignat*}
and
\begin{equation*}
  M(x,x')=\begin{pmatrix}p'_1(x') & P'_{1 m}(x') & 0\\ P'_{m 1}(x') & P'_{m m}(x') & 0\\ P_{n 1}(x) & 0 & P_{n n}(x)\end{pmatrix}\mpunkt
\end{equation*}

Then the linear system~(\ref{***}) can be written in the form
\begin{equation*}
  \begin{pmatrix}z''\\z'\\z\end{pmatrix}=M(x,x')\begin{pmatrix}z''\\z'\\z\end{pmatrix}\mpunkt
\end{equation*}

Hence, the first canonical solution of (\ref{*}), (\ref{**}), (\ref{***}) is $(\sigma, \sigma', M(\sigma,\sigma')^{\omega,1})$. Before we prove our lemma, we prove three identities.

The system (\ref{*}) can be written in the form
\begin{equation*}
  x=P_{n 1}(x)+P_{n n}(x)x,\quad\text{for } x=(x_1,\ldots,x_n)^{\mathsf T}\mpunkt
\end{equation*}
By the diagonal identity (see Proposition 2.2.11 of Ésik, Kuich~\cite{MAT}) the system
\begin{equation*}
  x= P_{n 1}(\sigma)+P_{n n}(\sigma)x
\end{equation*}
has the same least solution as~(\ref{*}). Hence,
\begin{equation}
  \sigma = P_{n n}(\sigma)^* P_{n 1}(\sigma)\mpunkt\label{2}
\end{equation}

The system (\ref{**}) can be written in the form
\begin{equation*}
  x' = P'_{m 1}(x')+P'_{m m}(x') x', \quad\text{for }x'=(x_1',\ldots,x'_{m})^{\mathsf T}\mpunkt
\end{equation*}
Again, by the diagonal identity (see Proposition 2.2.11 of Ésik, Kuich~\cite{MAT}) the system
\begin{equation*}
  x' = P'_{m 1}(\sigma')+P'_{m m}(\sigma') x'
\end{equation*}
has the same solution. Hence
\begin{align}
  \sigma' = P_{m m}'(\sigma')^*P_{m 1}'(\sigma')\mpunkt\label{3}
\end{align}
It follows for the first component
\begin{align}
  \sigma_1' &= \left(P_{m m}'(\sigma')^*P_{m 1}'(\sigma')\right)_1\notag\\
  &= \left(P_{m 1}'(\sigma')+P_{m m}'(\sigma')^+ P_{m 1}'(\sigma')\right)_1\notag\\
  &= \left(P_{m 1}'(\sigma')+P_{m m}'(\sigma')P_{m m}'(\sigma')^* P_{m 1}'(\sigma')\right)_1\notag\\
  &= p_1'(\sigma')+P_{1 m}'(\sigma')P_{m m}'(\sigma')^* P_{m 1}'(\sigma')\mpunkt\label{4}
\end{align}

We now compute
\begin{align*}
  (M^{\omega,1}(\sigma,\sigma'))_{z''} &= \left[p_1'(\sigma')+\begin{pmatrix}P_{1 m}'(\sigma') & 0\end{pmatrix}\begin{pmatrix} P_{m m}'(\sigma') & 0\\ 0 & P_{n n}(\sigma)\end{pmatrix}^*\begin{pmatrix} P_{m 1}'(\sigma')\\P_{n 1}(\sigma)\end{pmatrix}\right]^\omega\\
  &= \left[p_1'(\sigma')+\begin{pmatrix}P_{1 m}'(\sigma') & 0\end{pmatrix}\begin{pmatrix} P_{m m}'(\sigma')^* & 0\\ 0 & P_{n n}(\sigma)^*\end{pmatrix}\begin{pmatrix} P_{m 1}'(\sigma')\\P_{n 1}(\sigma)\end{pmatrix}\right]^\omega\\
  &= \left[p_1'(\sigma') + P_{1 m}'(\sigma') P_{m m}'(\sigma')^* P_{m 1}'(\sigma')\right]^\omega\\
  &= \sigma_1'^\omega\mpunkt
\end{align*}
The last equality is by~(\ref{4}).

When starting with another variable $z_i$ or $z_j'$ for $1\leq i\leq n$ and $1\leq j\leq m$, we get
\begin{align*}
  (M^{\omega,1}(\sigma,\sigma'))_{(z',z)} &= \begin{pmatrix} P_{m m}'(\sigma') & 0\\ 0 & P_{n n}(\sigma)\end{pmatrix}^*\begin{pmatrix} P_{m 1}'(\sigma')\\P_{n 1}(\sigma)\end{pmatrix}(M^{\omega,1}(\sigma,\sigma'))_{z''}\\
  &= \begin{pmatrix} P_{m m}'(\sigma')^* & 0\\ 0 & P_{n n}(\sigma)^*\end{pmatrix}\begin{pmatrix} P_{m 1}'(\sigma')\\P_{n 1}(\sigma)\end{pmatrix} \sigma_1'^\omega\\
  &= \begin{pmatrix} P_{m m}'(\sigma')^*P_{m 1}'(\sigma')\\ P_{n n}(\sigma)^*P_{n 1}(\sigma)\end{pmatrix} \sigma_1'^\omega
\end{align*}
Thus, by~(\ref{3}), we have, for $1\leq i\leq m$,
\begin{equation*}
  (M^{\omega,1}(\sigma,\sigma'))_{z'_i}=\left[P'_{m m}(\sigma')^*P'_{m 1}(\sigma')\right]_i\sigma_1'^\omega=\sigma_i'\sigma_1'^\omega\mkomma
\end{equation*}
and, by~(\ref{2}), we have, for $1\leq i\leq n$,
\begin{equation*}
  (M^{\omega,1}(\sigma,\sigma'))_{z_i}=\left[P_{n n}(\sigma)^*P_{n 1}(\sigma)\right]_i\sigma_1'^\omega=\sigma_i\sigma_1'^\omega\mpunkt
\end{equation*}
This completes the proof.
\end{proof}

%%%%%%%%%%%

Secondly, we deal with the case $s_1=(s_1,\epsilon)\epsilon$. Consider now the mixed $\omega$-algebraic system consisting of~(\ref{**}) and the linear system over $\ssigomega$

\begin{equation*}
  \begin{aligned}
    z'' &= p_1'(x')z'' + \sum_{1\leq j\leq m} p'_{1 j}(x') z_j'\mkomma\\
    z_i' &= p_i'(x')z''+\sum_{1\leq j\leq m} p'_{i j}(x') z_j',\;1\leq i\leq m\mkomma\\
    z_1 &= (s_1,\epsilon)p_1'(x')z''+(s_1,\epsilon)\sum_{1\leq j\leq m} p'_{1 j}(x') z_j'\mpunkt
  \end{aligned}\label{****}\tag{$\ast\!\ast\!\ast\ast$}
\end{equation*}

%% (noncanonical) solution taken from here

\begin{lemma}
  \label{lem:canonical_2}
  The first canonical solution of the mixed algebraic system (\ref{**}), (\ref{****}) is
  \begin{equation}\label{sol_2}
    (\sigma_1',\ldots,\sigma_m';\sigma_1'\sigma_1'^\omega;\sigma_1'\sigma_1'^\omega,\ldots,\sigma_{m}'\sigma_1'^\omega;(s_1,\epsilon)\sigma_1'^\omega)\mpunkt
  \end{equation}
\end{lemma}
\begin{proof}
  Let
\begin{equation*}
    M_\epsilon(x')=\begin{pmatrix}p'_1(x') & P'_{1 m}(x') & 0\\ P'_{m 1}(x') & P'_{m m}(x') & 0\\ (s_1,\epsilon)p'_1(x') & (s_1,\epsilon)P'_{1 m}(x') & 0\end{pmatrix}\mpunkt
\end{equation*}
Then the linear system~(\ref{****}) can be written in the form
\begin{equation*}
  \begin{pmatrix}z''\\z'\\z_1\end{pmatrix}=M_\epsilon(x')\begin{pmatrix}z''\\z'\\z_1\end{pmatrix}\mpunkt
\end{equation*}

Hence, the first canonical solution of (\ref{**}), (\ref{****}) is $(\sigma', M_\epsilon(\sigma')^{\omega,1})$.
We now compute
\begin{align*}
  (M_\epsilon^{\omega,1}(\sigma'))_{z''} &= \left[p_1'(\sigma')+\begin{pmatrix}P_{1 m}'(\sigma') & 0\end{pmatrix}\begin{pmatrix} P_{m m}'(\sigma') & 0\\ (s_1,\epsilon)P'_{1 m}(\sigma') & 0\end{pmatrix}^*\begin{pmatrix}P_{m 1}'(\sigma')\\(s_1,\epsilon)p'_1(\sigma')\end{pmatrix}\right]^\omega\\
  &= \left[p_1'(\sigma')+\begin{pmatrix}P_{1 m}'(\sigma') & 0\end{pmatrix}\begin{pmatrix} P_{m m}'(\sigma')^* & 0\\(s_1,\epsilon)P'_{1 m}(\sigma')P_{m m}'(\sigma')^* & 1\end{pmatrix}\begin{pmatrix}P_{m 1}'(\sigma')\\(s_1,\epsilon)p'_1(\sigma')\end{pmatrix}\right]^\omega\\
  &= \left[p_1'(\sigma') + P_{1 m}'(\sigma') P_{m m}'(\sigma')^* P_{m 1}'(\sigma')\right]^\omega\\
  &= \sigma_1'^\omega\mpunkt
\end{align*}
The last equality is by~(\ref{4}).

When starting with another variable $z_i'$ or $z_1$ for $1\leq i\leq m$, we get
\begin{align*}
  (M_\epsilon^{\omega,1}(\sigma'))_{(z',z_1)} &= \begin{pmatrix} P_{m m}'(\sigma') & 0\\ (s_1,\epsilon)P'_{1 m}(\sigma') & 0\end{pmatrix}^*\begin{pmatrix}P_{m 1}'(\sigma')\\(s_1,\epsilon)p'_1(\sigma')\end{pmatrix}(M_\epsilon^{\omega,1}(\sigma'))_{z''}\\
  &= \begin{pmatrix} P_{m m}'(\sigma')^* & 0\\(s_1,\epsilon)P'_{1 m}(\sigma')P_{m m}'(\sigma')^* & 1\end{pmatrix}\begin{pmatrix}P_{m 1}'(\sigma')\\(s_1,\epsilon)p'_1(\sigma')\end{pmatrix} \sigma_1'^\omega\\
  &= \begin{pmatrix} P_{m m}'(\sigma')^*P_{m 1}'(\sigma')\\ (s_1,\epsilon)P'_{1 m}(\sigma')P_{m m}'(\sigma')^*P_{m 1}'(\sigma')+(s_1,\epsilon)p'_1(\sigma')\end{pmatrix} \sigma_1'^\omega
\end{align*}
Thus, by~(\ref{3}), we have, for $1\leq i\leq m$,
\begin{equation*}
  (M_\epsilon^{\omega,1}(\sigma'))_{z'_i}=\left[P'_{m m}(\sigma')^*P'_{m 1}(\sigma')\right]_i\sigma_1'^\omega=\sigma_i'\sigma_1'^\omega\mkomma
\end{equation*}
and, by~(\ref{4}), we have
\begin{align*}
  (M_\epsilon^{\omega,1}(\sigma'))_{z_1} &=\big((s_1,\epsilon)P'_{1 m}(\sigma')P_{m m}'(\sigma')^*P_{m 1}'(\sigma')+(s_1,\epsilon)p'_1(\sigma')\big)\sigma_1'^\omega\\
  &=(s_1,\epsilon)\big(P'_{1 m}(\sigma')P_{m m}'(\sigma')^*P_{m 1}'(\sigma')+p'_1(\sigma')\big)\sigma_1'^\omega\\
  &=(s_1,\epsilon)\sigma_1'\sigma_1'^\omega=(s_1,\epsilon)\sigma_1'^\omega\mpunkt\qedhere
\end{align*}
\end{proof}

%  By Corollary~\ref{cor:4.1}, only the two cases $(s_1,\epsilon)=0$ and $s_1=(s_1,\epsilon)\epsilon$ need to be considered. In the first case, the first canonical solution for variable $z_1$ of the mixed algebraic system~(\ref{***}) in Greibach normal form is $(M^{\omega,1}(\sigma,\sigma'))_{z_1}=\sigma_1\sigma_1'^\omega=s_1 t_1^\omega$ by Lemma~\ref{lem:canonical_1}. For the second case, Lemma~\ref{lem:canonical_2} shows that the first canonical solution of the mixed algebraic system~(\ref{****}) in Greibach normal form is $(M_\epsilon^{\omega,1}(\sigma'))_{z_1}=(s_1,\epsilon)\sigma_1'^\omega=s_1 t_1^\omega$.

We now consider general sums of series of the above form. The next lemma shows how to construct a mixed $\omega$-algebraic system whose canonical solution is the sum of the canonical solutions of multiple mixed $\omega$-algebraic systems as given in Lemmas~\ref{lem:canonical_1} and~\ref{lem:canonical_2}.
\begin{lemma}
  \label{lem:canonical_both}
  Let $(s,\upsilon)\in\salg\times \ssigomega$ be given in the form of Corollary~\ref{cor:4.1}. Then there exists a mixed $\omega$-algebraic system in Greibach normal form such that $\upsilon$ is a component of its $l$\textsuperscript{th} canonical solution.
\end{lemma}

\begin{proof}
  Let $\upsilon=\sum_{1\leq i\leq l}s_i t_i^\omega$ as in the statement of Corollary~\ref{cor:4.1} and let $l\geq1$. By Lemmas~\ref{lem:canonical_1} and \ref{lem:canonical_2}, for $1\leq i\leq l$, there exist mixed $\omega$-algebraic systems
\begin{align}\label{sharp}
  x_i&=p_i(x_i),\tag{$\sharp$}\\
  \begin{pmatrix}z_i\\ \bar z_i\end{pmatrix}&=M_i(x_i)\begin{pmatrix}z_i\\ \bar z_i\end{pmatrix},\notag
\end{align}
in Greibach normal form with
\begin{equation*}
  M_i(x_i)=\begin{pmatrix}a_i & b_i\\ c_i & d_i\end{pmatrix},
\end{equation*}
where
\begin{align*}
  a_i &\in (S\langle(\Sigma\cup X)^*\rangle)^{1\times 1}, \\
  b_i &\in (S\langle(\Sigma\cup X)^*\rangle)^{1\times (n_i-1)}, \\
  c_i &\in (S\langle(\Sigma\cup X)^*\rangle)^{(n_i-1)\times 1}, \\
  d_i &\in (S\langle(\Sigma\cup X)^*\rangle)^{(n_i-1)\times (n_i-1)},
\end{align*}
such that $s_i t_i^\omega$ is a component of the first canonical solution of the $i$\textsuperscript{th} system. We will assume without loss of generality that $s_i t_i^\omega$ is the first component of variable $\bar z_i$, i.e.,
\begin{equation}\label{st}
  s_i t_i^\omega= \left[(M_i^{\omega,1})_{\bar z_i}\right]_1=\left[(d_i^* c_i)(a_i+b_i d_i^* c_i)^\omega\right]_1\mpunkt
\end{equation}

Similarly to the case of summation in Theorem 5.4.4 of Ésik, Kuich \cite{MAT}, we consider now the mixed $\omega$-algebraic system consisting of the algebraic systems (\ref{sharp}) over $\ssigstar$ and the linear system over $\ssigomega$
\begin{equation*}\label{ssharp}
 \hat z=M \hat z\mkomma\tag{$\sharp\sharp$}
\end{equation*}
with
\begin{equation*}
  M=\begin{pmatrix}
  \dmatrix{a} & \dmatrix{b} & 0\\
  \dmatrix{c} & \dmatrix{d} & 0\\
  \rmatrix{c} & \rmatrix{d} & 0
  \end{pmatrix},\quad
  \hat z=\begin{pmatrix}z_1\\\vdots\\z_l\\ \bar z_1\\\vdots\\\bar z_l\\z'\end{pmatrix}\mpunkt
\end{equation*}
Note that this system~\ref{ssharp} is still in Greibach normal form.

We order the variables of the mixed $\omega$-algebraic system~(\ref{sharp}), (\ref{ssharp}) as $z_1,\ldots,z_l;\allowbreak \bar z_1,\ldots, \bar z_l;z'$.
We now compute the $l$\textsuperscript{th} canonical solution, starting with variable $z=(z_1, \ldots, z_l)^{\mathsf T}$. Then
\begin{align*}
  (M^{\omega,l})_z &=\left[\dmatrix{a}+\begin{pmatrix}\dmatrix{b} & 0\end{pmatrix}
    \begin{pmatrix}\dmatrix{d} & 0\\\rmatrix{d} & 0\end{pmatrix}^*
    \begin{pmatrix}\dmatrix{c}\\ \rmatrix{c}\end{pmatrix}\right]^\omega\\
  &=\left[\dmatrix{a}+\begin{pmatrix}\dmatrix{b} & 0\end{pmatrix}
    \begin{pmatrix}\dmatrix{d}^* & 0\\\rmatrix{d}\dmatrix{d}^* & 1\end{pmatrix}
    \begin{pmatrix}\dmatrix{c}\\ \rmatrix{c}\end{pmatrix}\right]^\omega\\
  &=\left[\dmatrix{a}+\dmatrix{b}\dmatrix{d}^*\dmatrix{c}\right]^\omega\\
  &=\begin{pmatrix}(a_1+b_1 d_1^* c_1)^\omega\\ \vdots\\ (a_l+b_l d_l^* c_l)^\omega\end{pmatrix}\mpunkt
\end{align*}

When starting with the new variable $z'$, we get a sum of the original solutions:
\begin{align*}
  (M^{\omega,l})_{z'}&=\left[\begin{pmatrix}\dmatrix{d} & 0\\\rmatrix{d} & 0\end{pmatrix}^*
    \begin{pmatrix}\dmatrix{c}\\ \rmatrix{c}\end{pmatrix}(M^{\omega,l})_z\right]_{l+1}\\
  &=\left[\begin{pmatrix}\dmatrix{d}^* & 0\\\rmatrix{d}\dmatrix{d}^* & 1\end{pmatrix}
    \begin{pmatrix}\dmatrix{c}\\ \rmatrix{c}\end{pmatrix}(M^{\omega,l})_z\right]_{l+1}\\
  &=\left[\begin{pmatrix}\dmatrix{d}^*\dmatrix{c}\\\begin{pmatrix}d_1 d_1^* & \cdots & d_l d_l^*\end{pmatrix}\dmatrix{c}+\rmatrix{c}\end{pmatrix}(M^{\omega,l})_z\right]_{l+1}\\
  &=\left[\begin{pmatrix}\begin{pmatrix}d_1^* c_1 & & \\& \ddots & \\& & d_l^* c_l\end{pmatrix}\\\begin{pmatrix}d_1 d_1^* c_1+c_1 & \cdots & d_l d_l^* c_l+c_l\end{pmatrix}\end{pmatrix}\begin{pmatrix}(a_1+b_1 d_1^* c_1)^\omega\\ \vdots\\ (a_l+b_l d_l^* c_l)^\omega\end{pmatrix}\right]_{l+1}\allowdisplaybreaks\\
  &=\begin{pmatrix}d_1^* c_1 (a_1+b_1 d_1^*c_1)^\omega\\\vdots\\d_l^* c_l (a_l+b_l d_l^* c_l)^\omega\\\sum_{1\leq i\leq l}(d_i d_i^* c_i +c_i)(a_i+b_i d_i^* c_i)^\omega\end{pmatrix}_{l+1}\\
  &= \sum_{1\leq i\leq l}(d_i d_i^* c_i +c_i)(a_i+b_i d_i^* c_i)^\omega\\
  &= \sum_{1\leq i\leq l}(d_i^* c_i)(a_i+b_i d_i^* c_i)^\omega
\end{align*}
Thus, the first component is (by identity~(\ref{st}))
\begin{align*}
  \left[(M^{\omega,l})_{z'}\right]_1 &= \left[\sum_{1\leq i\leq l}(d_i^* c_i)(a_i+b_i d_i^* c_i)^\omega\right]_1\\
  &= \sum_{1\leq i\leq l}\left[(d_i^* c_i)(a_i+b_i d_i^* c_i)^\omega\right]_1 = \sum_{1\leq i\leq l}s_i t_i^\omega = \upsilon\mpunkt\qedhere
\end{align*}
\end{proof}

We can now conclude the following theorem.
\begin{theorem}
  \label{thm:greibach}
  The following statement for $(s,\upsilon)\in\stimes$ is equivalent to the statements of Theorem~\ref{thm:3.1}:\\
  $(s,\upsilon)$ is component of a canonical solution of a mixed $\omega$-algebraic system over $\stimes$ in Greibach normal form.
\end{theorem}
\begin{proof}
  The above statement trivially implies statement~\eqref{vi} of Theorem~\ref{thm:3.1}. By Corollary~\ref{cor:4.1} and Lemma~\ref{lem:canonical_both}, the statements of Theorem~\ref{thm:3.1} imply the above statement.
\end{proof}

%% \begin{corollary}
%%   For every pair $(s,\upsilon)\in\salg\times\salgomega$, there exists a weighted grammar $G_k$ in Greibach normal form such that $(\upsilon,w)$ is the sum of the weights of all infinite derivations of $w$ with respect to $G_k$ and at least one of the variables of $\{z_1,\ldots,z_k\}$ appears infinitely often in these infinite derivations.
%% \end{corollary}

%% \begin{proof}
%%   As assumed in the beginning,  $S$ is a complete star-omega semiring. Then by Theorem~5.5.9 of Ésik, Kuich~\cite{MAT} together with the above Theorem~\ref{thm:greibach}, we can define a right-linear grammar $G_k$ with variables $\{z_1,\ldots,z_n\}$ and repeated variables $\{z_1,\ldots,z_k\}$ for $k\leq n$ such that the statement of the corollary holds.
%% \end{proof}

\section{Greibach Normal Form for \texorpdfstring{$\omega$}{omega}-Algebraic Systems}
\label{sec:greibach_normal_form_unmixed}

%\TODO{rephrase and explain that this is of independent interest (extends characterization of Theorem~\ref{thm:3.1}):}
%For the following sections, we need the Greibach normal form not only for mixed $\omega$-algebraic systems but also for $\omega$-algebraic systems.
We show in this section a specialization of Theorem~\ref{thm:greibach} for $\omega$-algebraic systems: already $\omega$-algebraic systems in Greibach normal form are sufficient to describe all $\omega$-algebraic series.
%This section therefore proves that we can specialize the result of the last section thus that each $\omega$-algebraic series is a component of a canonical solution of an $\omega$-algebraic system in Greibach normal form.

We will apply this new result in Section~\ref{sec:simple_omega_reset_pushdown_automata}, but we believe that proving the existence of the Greibach normal form for $\omega$-algebraic systems is of independent interest.

Similar to the definition for mixed $\omega$-algebraic systems, an $\omega$-algebraic system
\[y=p(y)\]
where $\{y_1,\ldots,y_n\}$ is a set of variables for the quemiring $\salgtimes$,
is in \emph{Greibach normal form} if
\begin{alignat*}{2}
&\text{supp}(p_i(y)) \subseteq \{\epsilon\}\cup\Sigma\cup \Sigma Y\cup \Sigma Y Y, \qquad &&\text{for all }1\leq i\leq n\mpunkt
\end{alignat*}

%Now, we can apply the lemma in the next theorem.

Our first main result is the following.

\begin{theorem}
  \label{thm:unmix}
  The following statement for $(s,\upsilon)\in\stimes$ is equivalent to the statements of Theorem~\ref{thm:3.1}:\\
  $(s,\upsilon)$ is component of a canonical solution of an $\omega$-algebraic system over $\stimes$ in Greibach normal form.
\end{theorem}
\begin{proof}
  By Theorem~\ref{thm:greibach}, we can assume that $(s,\upsilon)$ is component of the $t$\textsuperscript{th} canonical solution of a mixed $\omega$-algebraic system over $\stimes$ in Greibach normal form for a $t\in\N$. Let the mixed $\omega$-algebraic system be given in the following form:
  \begin{alignat}{2}
    x_i &= p_i+\sum_{1\leq j\leq n}(p_{i j} x+q_{i j})x_j,&\quad &\text{for }1\leq i\leq n\mkomma\tag{$\divideontimes$}\label{fin}\\
    z_i &= \sum_{1\leq j\leq m}(p'_{i j} x+q'_{i j})z_j,&\quad &\text{for }1\leq i\leq m\mkomma\tag{$\divideontimes\divideontimes$}\label{inf}
  \end{alignat}
  where
  \begin{alignat*}{2}
    &p_{i j}\in \ssig^{1\times n},&\qquad &\text{for }1\leq i,j\leq n\mkomma\\
    &p'_{i j}\in \ssig^{1\times n},&\qquad &\text{for }1\leq i,j\leq m\mkomma
  \end{alignat*}
  and
  \begin{alignat*}{3}
    &\text{supp}(p_i)\subseteq \{\epsilon\}\cup\Sigma,\quad& \text{supp}(p_{i j}x)\subseteq \Sigma X,\quad& \text{supp}(q_{i j})\subseteq \Sigma\mkomma\\
    & &\text{supp}(p'_{i j}x)\subseteq \Sigma X,\quad& \text{supp}(q'_{i j})\subseteq \Sigma\mpunkt
  \end{alignat*}
  Note that
  \[p_{i j} x = \sum_{1\leq k\leq n}(p_{i j})_k x_k\msemikolon\]
  we decided for this notation because of brevity, important especially in matrices.

  For the remainder of the proof, consider integers $k$ and $l$ to be fixed such that the $t$\textsuperscript{th} canonical solution of \eqref{fin}, \eqref{inf} is $(\sigma,\omega)$ with $\sigma_k=s$ and $\omega_{l}=\upsilon$.

  We will later need a simple implication: We can write the linear system \eqref{inf} as
  \[z=P'_{m m}(x)z\mkomma\]
  where
  \[P'_{m m}(x)=
  \begin{pmatrix}
    p'_{1 1} x+q'_{1 1} & \cdots & p'_{1 m} x+q'_{1 m}\\
    \vdots & \ddots & \vdots\\
    p'_{m 1} x+q'_{m 1} & \cdots & p'_{m m} x+q'_{m m}
  \end{pmatrix}\mpunkt\]
  Note that $t\leq m$. It follows that
  \begin{equation}
    \label{eqn:omega}
    \omega=P'_{m m}(\sigma)^{\omega,t}.
  \end{equation}

  Now, we construct from \eqref{fin}, \eqref{inf} an $\omega$-algebraic system \eqref{unmixed} where the variables $x$ are substituted by $\bar y$ and $z$ by $\hat y$. Additionally, we add a new equation and a new variable $\dot y$ to combine the $k$\textsuperscript{th} component of the semiring part and the $l$\textsuperscript{th} component of the semimodule part:
\begin{equation*}
  \begin{aligned}
    \hat y_i &= \sum_{1\leq j\leq m}(p'_{i j} \bar y+q'_{i j})\hat y_j,\qquad \text{for }1\leq i\leq m\mkomma\\
    \bar y_i &= p_i+\sum_{1\leq j\leq n}(p_{i j} \bar y+q_{i j})\bar y_j,\qquad \text{for }1\leq i\leq n\mkomma\\
    \dot y &= p_k+\sum_{1\leq j\leq n}(p_{k j} \bar y+q_{k j})\bar y_j+\sum_{1\leq j\leq m}(p'_{l j} \bar y+q'_{l j})\hat y_j\mpunkt
  \end{aligned}\label{unmixed}\tag{$\divideontimes\!\divideontimes\!\divideontimes$}
\end{equation*}
Note that \eqref{unmixed} is in Greibach normal form.
Moreover, note that we order the equations such that the first equations are those corresponding to the old equations of variables $z_i$. This ensures that the $t$\textsuperscript{th} canonical solution still considers the correct variables as Büchi-accepting.

\emph{Claim}: The $(m+n+1)$\textsuperscript{th} component of the $t$\textsuperscript{th} canonical solution of \eqref{unmixed} is $(\sigma_k,\omega_l)=(s,\upsilon)$.

We now compute this solution. The $t$\textsuperscript{th} canonical solution of the $\omega$-algebraic system \eqref{unmixed} is defined to be the $t$\textsuperscript{th} canonical solution of the mixed $\omega$-algebraic system induced by \eqref{unmixed}. The corresponding induced mixed $\omega$-algebraic system is given by the algebraic system over $\salg$
\begin{align}
  \begin{split}
    \hat x_i &= \sum_{1\leq j\leq m}(p'_{i j} \bar x+q'_{i j})\hat x_j,\qquad \text{for }1\leq i\leq m\mkomma\\
    \bar x_i &= p_i+\sum_{1\leq j\leq n}(p_{i j} \bar x+q_{i j})\bar x_j,\qquad \text{for }1\leq i\leq n\mkomma\\
    \dot x &= p_k+\sum_{1\leq j\leq n}(p_{k j} \bar x+q_{k j})\bar x_j+\sum_{1\leq j\leq m}(p'_{l j} \bar x+q'_{l j})\hat x_j\mkomma
  \end{split}\tag{\#}\label{unmixed_fin}
  \intertext{and the linear system over $\salgomega$}
  \begin{split}
    \hat z_i &= \sum_{1\leq j\leq m}(p'_{i j} \bar x+q'_{i j})\hat z_j+p'_{i j}\bar z,\qquad \text{for }1\leq i\leq m\mkomma\\
    \bar z_i &= \sum_{1\leq j\leq n}(p_{i j} \bar x+q_{i j})\bar z_j+p_{i j}\bar z,\qquad \text{for }1\leq i\leq n\mkomma\\
    \dot z &= \sum_{1\leq j\leq n}(p_{k j} \bar x+q_{k j})\bar z_j+p_{k j}\bar z+\sum_{1\leq j\leq m}(p'_{l j} \bar x+q'_{l j})\hat z_j+p'_{l j}\bar z\mpunkt
  \end{split}\tag{\#\#}\label{unmixed_inf}
\end{align}

\emph{Claim}: $(0,\ldots,0;\sigma;\sigma_k)$ is the least solution of \eqref{unmixed_fin}.

%Intuitively, the first $m$ equations do not have a finite solution and the next $n$ equations are equal to \eqref{fin}.

First, we prove that it is a solution by plugging it into the right sides of the equations.
We have for the first $m$ equations, and for $1\leq i\leq m$,
\begin{align*}
  %\sum_{1\leq j\leq m}(p'_{i j} \bar x+q'_{i j})\hat x_j &=
  \sum_{1\leq j\leq m}(p'_{i j} \sigma+q'_{i j})0 &= 0\mpunkt
\end{align*}
Then for the second set of equations and $1\leq i\leq n$,
\begin{align*}
  %p_i+\sum_{1\leq j\leq n}(p_{i j} \bar x+q_{i j})\bar x_j &=
  p_i+\sum_{1\leq j\leq n}(p_{i j} \sigma+q_{i j})\sigma_j
  &= \sigma_i\msemikolon
\end{align*}
because $\sigma$ is a solution of \eqref{fin}.
Finally, we obtain by the same reason, for the last equation,
\begin{align*}
  %p_k+\sum_{1\leq j\leq n}(p_{k j} \bar x+q_{k j})\bar x_j+\sum_{1\leq j\leq m}(p'_{l j} \bar x+q'_{l j})\hat x_j\\
  p_k+\!\sum_{1\leq j\leq n}(p_{k j} \sigma+q_{k j})\sigma_j+\!\sum_{1\leq j\leq m}(p'_{l j} \sigma+q'_{l j})0_j &=  p_k+\sum_{1\leq j\leq n}(p_{k j} \sigma+q_{k j})\sigma_j+0\\
  &=  \sigma_k\mpunkt
\end{align*}
The algebraic system \eqref{unmixed_fin} is \emph{strict} and therefore has a unique solution. See~\cite{MAT}, p.\@ 62 for a definition and~\cite{MAT}, Theorem~2.4.7 for the unicity. This means that $(0,\ldots,0;\sigma;\sigma_k)$ is also the least solution. This proves the claim.

%%%%%%%

Now consider the linear system \eqref{unmixed_inf}. Let $P'_{m m}(\bar x)$ be defined as above and let further
\[P_{n n}(\bar x)=
\begin{pmatrix}
  p_{1 1}\bar x+q_{1 1} & \cdots & p_{1 n}\bar x+q_{1 n}\\
  \vdots & \ddots & \vdots\\
  p_{n 1}\bar x+q_{n 1} & \cdots & p_{n n}\bar x+q_{n n}
\end{pmatrix}\mkomma\]
\[R_{n n}=
\begin{pmatrix}
  \sum_{1\leq j\leq n}(p_{1 j})_1 & \cdots & \sum_{1\leq j\leq n}(p_{1 j})_n\\
  \vdots & \ddots & \vdots\\
  \sum_{1\leq j\leq n}(p_{n j})_1 & \cdots & \sum_{1\leq j\leq n}(p_{n j})_n
\end{pmatrix}\mkomma\]
\[R'_{m n}=
\begin{pmatrix}
  \sum_{1\leq j\leq m}(p'_{1 j})_1 & \cdots & \sum_{1\leq j\leq m}(p'_{1 j})_n\\
  \vdots & \ddots & \vdots\\
  \sum_{1\leq j\leq m}(p'_{m j})_1 & \cdots & \sum_{1\leq j\leq m}(p'_{m j})_n
\end{pmatrix}\mpunkt\]

Note that for \eqref{unmixed_inf} and for $1\leq i\leq m$, we have
\begin{align*}
  \sum_{1\leq j\leq m} p'_{i j}\bar z &= \sum_{1\leq j\leq m}\sum_{1\leq k\leq n} (p'_{i j})_k \bar z_k\\
  &= \sum_{1\leq k\leq n}\sum_{1\leq j\leq m} (p'_{i j})_k \bar z_k\\
  &= \Big(\sum_{1\leq j\leq m} (p'_{i j})_1,\cdots,\sum_{1\leq j\leq m}(p'_{i j})_n\Big) \bar z\\
  &= (R'_{m n})_i \bar z\mpunkt
\end{align*}
Analogously, we can prove $\sum_{1\leq j\leq n}p_{i j}\bar z=(R_{n n})_i\bar z$.
We let
\[M(\hat x,\bar x,x)=
\begin{pmatrix}
  P'_{m m}(\bar x) & R'_{m n} & 0\\
  0 & P_{n n}(\bar x)+R_{n n} & 0\\
  (P'_{m m}(\bar x))_l & (P_{n n}(\bar x))_k+(R_{n n})_k+(R'_{m n})_l & 0
\end{pmatrix}\mkomma\]
then the linear system \eqref{unmixed_inf} can be written as
\[\begin{pmatrix}\hat z\\ \bar z\\ z\end{pmatrix}=M(\hat x,\bar x,x)\begin{pmatrix}\hat z\\ \bar z\\ z\end{pmatrix}\mpunkt\]

Now, we can plug the semiring part $(0,\sigma,\sigma_k)$ of the solution into $M$. By Theorem~\ref{thm:omegak}, the semimodule part of the canonical solution of \eqref{unmixed_fin}, \eqref{unmixed_inf} is
\begin{align*}
  M(0,\sigma,\sigma_k)^{\omega,t} = \begin{pmatrix}
    \xi^{\omega,t}\\
    \begin{pmatrix}
    P_{n n}(\sigma)+R_{n n} & 0\\
    \chi & 0
  \end{pmatrix}^*
  \begin{pmatrix}
    0\\ (P'_{m m}(\sigma))_l
  \end{pmatrix}
  \xi^{\omega,t}
  \end{pmatrix}
\end{align*}
with
\[\chi = (P_{n n}(\sigma))_k+(R_{n n})_k+(R'_{m n})_l\]
and
\begin{align*}
  \xi &=P'_{m m}(\sigma)+\begin{pmatrix}R'_{m n} & 0\end{pmatrix}
  \begin{pmatrix}
    P_{n n}(\sigma)+R_{n n} & 0\\
    \chi & 0
  \end{pmatrix}^*
  \begin{pmatrix}
    0\\ (P'_{m m}(\sigma))_l
  \end{pmatrix}\\
  &= P'_{m m}(\sigma)+\begin{pmatrix}R'_{m n} & 0\end{pmatrix}
  \begin{pmatrix}
    (P_{n n}(\sigma)+R_{n n})^* & 0\\
    \chi(P_{n n}(\sigma)+R_{n n})^* & 1
  \end{pmatrix}
  \begin{pmatrix}
    0\\ (P'_{m m}(\sigma))_l
  \end{pmatrix}\\
  &= P'_{m m}(\sigma)+\begin{pmatrix}R'_{m n}(P_{n n}(\sigma)+R_{n n})^* & 0\end{pmatrix}
  \begin{pmatrix}
    0\\ (P'_{m m}(\sigma))_l
  \end{pmatrix}\\
  &= P'_{m m}(\sigma)+0 = P'_{m m}(\sigma)\mpunkt
\end{align*}
It follows that
\begin{align*}
  M(0,\sigma,\sigma_k)^{\omega,t} &=
  \begin{pmatrix}
    P'_{m m}(\sigma)^{\omega,t}\\
    \begin{pmatrix}
      P_{n n}(\sigma)+R_{n n} & 0\\
      \chi & 0
    \end{pmatrix}^*
    \begin{pmatrix}
      0\\ (P'_{m m}(\sigma))_l
    \end{pmatrix}
    P'_{m m}(\sigma)^{\omega,t}
  \end{pmatrix}\\
  &=
  \begin{pmatrix}
    P'_{m m}(\sigma)^{\omega,t}\\
    \begin{pmatrix}
      (P_{n n}(\sigma)+R_{n n})^* & 0\\
      \chi(P_{n n}(\sigma)+R_{n n})^* & 1
    \end{pmatrix}
    \begin{pmatrix}
      0\\ (P'_{m m}(\sigma))_l
    \end{pmatrix}
    P'_{m m}(\sigma)^{\omega,t}
  \end{pmatrix}\\
  &=
  \begin{pmatrix}
    P'_{m m}(\sigma)^{\omega,t}\\
    \begin{pmatrix}
      0\\ (P'_{m m}(\sigma))_l
    \end{pmatrix}
    P'_{m m}(\sigma)^{\omega,t}
  \end{pmatrix}\\
  &=
  \begin{pmatrix}
    P'_{m m}(\sigma)^{\omega,t}\\ 0\\ (P'_{m m}(\sigma))_l P'_{m m}(\sigma)^{\omega,t}
  \end{pmatrix}\mpunkt
\end{align*}

Now, we have for the last component
\begin{align*}
  \left(M(0,\sigma,\sigma_k)^{\omega,t}\right)_{m+n+1} &= (P'_{m m}(\sigma))_l P'_{m m}(\sigma)^{\omega,t}\\
  &= \left(P'_{m m}(\sigma) P'_{m m}(\sigma)^{\omega,t}\right)_l\\
  &= \left(P'_{m m}(\sigma)^{\omega,t}\right)_l=\omega_l\mkomma
\end{align*}
where the third equality is by Theorem~5.5.1 of~\cite{MAT} and the last equality is by~\eqref{eqn:omega}. In summary, the $(n+m+1)$\textsuperscript{th} component of the $t$\textsuperscript{th} canonical solution of \eqref{unmixed_fin}, \eqref{unmixed_inf} is $(\sigma_k,\omega_l)=(s,\upsilon)$. As defined for $\omega$-algebraic systems, it then follows that also the $t$\textsuperscript{th} canonical solution of \eqref{unmixed} is $(s,\upsilon)$.
\end{proof}

As the mixed $\omega$-algebraic system in the preceding proof does not depend on the previous discussion and since we proved that we can construct the Greibach normal form when needed, we infer the following.

\begin{corollary}
  Let $(s,\upsilon)$ be a component of a canonical solution of a mixed $\omega$-algebraic system over $\stimes$.

  Then we can construct an $\omega$-algebraic system over $\stimes$ (in Greibach normal form) where $(s,\upsilon)$ is a component of a canonical solution.
\end{corollary}

%%%%%%%%%%%%%%%%%%%%%%%%%%%%%%%%%%%%%%%%%%%%%%%%%%%%%%%%%%%%%%%%%%%%%%%%%%%%%%%%%%%%%%
%%%%%%%%%%%%%%%%%%%%%%%%%%%%%%%%%%%%%%%%%%%%%%%%%%%%%%%%%%%%%%%%%%%%%%%%%%%%%%%%%%%%%%
%%%%%%%%%%%%%%%%%%%%%%%%%%%%%%%%% Simple PDA
%%%%%%%%%%%%%%%%%%%%%%%%%%%%%%%%%%%%%%%%%%%%%%%%%%%%%%%%%%%%%%%%%%%%%%%%%%%%%%%%%%%%%%
%%%%%%%%%%%%%%%%%%%%%%%%%%%%%%%%%%%%%%%%%%%%%%%%%%%%%%%%%%%%%%%%%%%%%%%%%%%%%%%%%%%%%%

\section{Simple Reset Pushdown Automata}
\label{sec:simple_reset_pushdown_automata}

Now that we have proved the existence of the Greibach normal form for every $\omega$-algebraic system and every mixed $\omega$-algebraic system, we want to use it in the second part of the paper to show that each $\omega$-algebraic series can be represented as the behavior of a simple $\omega$-reset pushdown automaton. The next section will prove that result. For the proof, we will need the corresponding result for finite words as an intermediate step. We have shown in~\cite{finite_simple} that for every algebraic series $r$ (of finite words), there exists a simple reset pushdown automaton with behavior $r$. We recall the construction of the simple reset pushdown automata here for the convenience of the reader, as variants of these automata will
be used in Section~\ref{sec:simple_omega_reset_pushdown_automata} for $\omega$-algebraic series.

Following Kuich, Salomaa \cite{kuich_formal_power_series} and Kuich \cite{kuich_semirings}, we introduce pushdown transitions matrices. These matrices can be considered as adjacency matrices of graphs representing automata. A special form, the reset pushdown matrices, is used for pushdown automata starting with an empty stack and allowing the automaton to push onto the empty stack. Here, we are interested in simple reset pushdown matrices, introduced in~\cite{finite_simple}. This simple form allows the automaton only to push one symbol, to pop one symbol or to ignore the stack. The corresponding automata, the simple reset pushdown automata are a generalization of the unweighted automata used in~\cite{unweighted_logic}. They do not use $\epsilon$-transitions and do not allow the inspection of the topmost stack symbol.

%In this paper, we are interested in automata for finite and infinite words. For the convenience of the reader, here we recall the finite part from Droste, Dziadek, Kuich~\cite{finite_simple}.

A matrix $M\in (S^{n\times n})^{\Gamma^*\times \Gamma^*}$ is called \emph{row-finite} if $\{\pi' \mid M_{\pi,\pi'}\neq 0\}$ is finite for all $\pi\in\Gamma^*$.

Let $\Gamma$ be an alphabet, called \emph{pushdown alphabet} and let $n\geq 1$. A matrix $\bar M \in (S^{n \times n})^{\Gamma^* \times \Gamma^*}$ is called a \emph{pushdown matrix} (with \emph{pushdown alphabet} $\Gamma$ and \emph{stateset} $\{1,\dots,n\}$) if
  \begin{enumerate}[(i)]
  \item $\bar M$ is row-finite;\label{itm:pushdown_matrix_1}
%for each $p \in \Gamma$ there exist only finitely many blocks $\bar M_{p,\pi}$, $\pi \in \Gamma^*$, that are unequal to $0$;
  \item for all $\pi_1, \pi_2 \in \Gamma^*$,
  \[\bar M_{\pi_1,\pi_2} = \begin{cases}
      \bar M_{p,\pi}, & \text{if there exist } p \in \Gamma, \pi,\pi' \in \Gamma^* \text{ with }\pi_1 = p\pi' \text{ and } \pi_2 = \pi \pi', \\
      0, & \text{otherwise.}
    \end{cases}\]\label{itm:pushdown_matrix_2}
  \end{enumerate}
Intuitively, here \eqref{itm:pushdown_matrix_2} means that the infinite pushdown matrix $\bar M$ is fully represented already by the blocks $\bar M_{p,\pi}$ where $p\in\Gamma$, $\pi\in\Gamma^*$, and \eqref{itm:pushdown_matrix_1} means that only finitely many such blocks are nonzero.

Let $\Gamma$ be a pushdown alphabet and $\{1,\dots,n\}$, $n\geq1$, be a set of states.
A \emph{reset matrix} $M_R\in\snngam$ is a row-finite matrix such that
\[(M_R)_{\pi_1,\pi_2}=0\hspace{1cm}\text{for }\pi_1,\pi_2\in\Gamma^*\text{ with }\pi_1\neq\epsilon\mpunkt\]

A \emph{reset pushdown matrix} $M\in\snngam$ is the sum of a reset matrix $M_R$ and a pushdown matrix $\bar M$,
\[M=M_R+\bar M\mpunkt\]
Intuitively, a reset pushdown matrix is similar to a pushdown matrix with the additional possibility to push onto the empty stack, i.e., $M_{\epsilon,\pi}$ is allowed to be nonzero. Note that the entries of reset pushdown matrices are determined by finitely many values because it is row-finite and property \eqref{itm:pushdown_matrix_2} of pushdown matrices ensures that the value of $M_{p \pi',\pi \pi'}$ is equal to (and therefore can be derived from) $M_{p,\pi}$.

%\section{Simple Reset Pushdown Matrices}

%For the rest of this paper, \emph{the complete starsemiring $S$ is additionally assumed to be commutative}; and $\Sigma$ denotes an alphabet.

A reset pushdown matrix $M$ is called \emph{simple} if, $M\in\big((\ssig)^{n\times n}\big)^{\Gamma^*\times\Gamma^*}$ for some $n\geq 1$, and for all $p, p_1\in\Gamma$,
\[M_{p,\epsilon},\;M_{p,p}=M_{\epsilon,\epsilon}\text{ and }M_{p,p_1 p}=M_{\epsilon,p_1},\]
are the only blocks $M_{\pi,\pi'}$, where $\pi\in\{\epsilon,p\}$ and $\pi'\in\Gamma^*$, that may be unequal to the zero matrix $0$.

Hence, a simple reset pushdown matrix $M$ is defined by its blocks $M_{\epsilon,\epsilon}$ and $M_{p,\epsilon}$, $M_{\epsilon,p}$ ($p\in\Gamma$). Intuitively, the automata will only be allowed to ignore the stack (modeled by $M_{\epsilon,\epsilon}$), pop one symbol ($M_{p,\epsilon}$) or push one symbol ($M_{\epsilon,p}$). Note also that the matrix $M\in((\ssig)^{n\times n})^{\Gamma^*\times\Gamma^*}$ forbids $\epsilon$-transitions. Moreover, the equalities $M_{p,p}=M_{\epsilon,\epsilon}$ and $M_{p,p_1p}=M_{\epsilon,p_1}$ imply that the next transition does not depend on the topmost symbol of the stack except when popping it (modeled by $M_{p,\epsilon}$). %An example of a simple reset pushdown matrix can be found in Example~\ref{ex:automaton}.

%%%%%%%%%%%%%%%%%%%%%%%%%%%%%%%%%%

%% p. 10
A \emph{reset pushdown automaton} (with input alphabet $\Sigma$) $\mathfrak{A}=(n,\Gamma,I,M,P)$ is given by
\begin{itemize}
\item a \emph{set of states} $\{1,\dots,n\}$, $n\geq1$,
\item a \emph{pushdown alphabet} $\Gamma$,
\item a reset pushdown matrix $M\in((S \langle \Sigma\cup\{\epsilon\} \rangle)^{n\times n})^{\Gamma^*\times\Gamma^*}$ called \emph{transition matrix},
\item a row vector $I\in(\seps)^{1\times n}$, called \emph{initial state vector},
\item a column vector $P\in(\seps)^{n\times1}$, called \emph{final state vector}.
\end{itemize}
The \emph{behavior} $\|\mathfrak{A}\|$ of a reset pushdown automaton $\mathfrak{A}$ is defined by
\[\|\mathfrak{A}\|=I(M^*)_{\epsilon,\epsilon}P\mpunkt\]

%% p. 7, part 2
A reset pushdown automaton $\mathfrak{A}=(n,\Gamma,I,M,P)$ is called \emph{simple} if $M$ is a simple reset pushdown matrix.

Example~\ref{ex:simple_automaton} will show a simple reset pushdown automaton and the corresponding simple reset pushdown matrix.

%% p. 11
Given a series $r\in \salg$, we want to construct a simple reset pushdown automaton with behavior $r$. By Theorems~5.10 and 5.4 of \cite{kuich_semirings}, $r$ is a component of the unique solution of a strict algebraic system in Greibach normal form.

We only consider the algebraic series $r$ with $(r,\epsilon)=0$; cf.\@ \cite{finite_simple} for the other case. So we assume without loss of generality that $r$ is the $x_1$-component of
%% no empty line
%% p. 12
the unique solution of the algebraic system \eqref{eqn:2} with variables $x_1,\dots,x_n$
\[ x_i = p_i,\;\;1\leq i\leq n,\]
of the form
\begin{align}
  \tag{$\blacklozenge$}
  \label{eqn:2}
    x_i= \sum_{1\leq j,k\leq n}\sum_{a\in\Sigma}(p_i,a x_j x_k)a x_j x_k+\sum_{1\leq j\leq n}\sum_{a\in\Sigma}(p_i,a x_j)a x_j+\sum_{a\in\Sigma}(p_i,a)a\mpunkt
\end{align}

%% p. 13
As in~\cite{finite_simple}, we now construct the simple reset pushdown automaton $\mathfrak{A}_m=(n+1,\Gamma,I_m,M,P)$, $1\leq m\leq n$, with $r=\|\mathfrak{A}_1\|$ as follows:\\
We let $\Gamma=\{x_1,\dots,x_n\}$; we also denote the state $n+1$ by $f$; the entries of $M$ of the form $(M_{x_k,x_k})_{i,j}$, $(M_{x_k,\epsilon})_{i,j}$, $(M_{\epsilon,x_k})_{i,j}$, $(M_{\epsilon,\epsilon})_{i,j}$, $(M_{\epsilon,\epsilon})_{i,f}$, $1\leq i,j,k\leq n$, that may be unequal to $0$ are
\begin{align*}
  %(M_{x_t,x_k x_t})_{i j} &= \sum_{a\in\Sigma}(p_i,a x_j x_k)a\mkomma\\
  %(M_{x_t,x_t})_{i j} &= \sum_{a\in\Sigma}(p_i,a x_j)a\mkomma\\
  %(M_{x_t,\epsilon})_{i t} &= \sum_{a\in\Sigma}(p_i,a)a\mkomma\\
  (M_{\epsilon,x_k})_{i,j} &= \sum_{a\in\Sigma}(p_i,a x_j x_k)a\mkomma\\
  (M_{x_k,x_k})_{i,j} = (M_{\epsilon,\epsilon})_{i,j} &= \sum_{a\in\Sigma}(p_i,a x_j)a\mkomma\\
  (M_{x_k,\epsilon})_{i,k} = (M_{x_k,x_k})_{i,f} = (M_{\epsilon,\epsilon})_{i,f} &= \sum_{a\in\Sigma}(p_i,a)a\msemikolon
\end{align*}
we further put $(I_m)_m=\epsilon, (I_m)_i=0$ for $1\leq i\leq m-1$ and $m+1\leq i \leq n+1$; finally let $P_f=\epsilon$ and $P_j=0$ for $1\leq j\leq n$;

The following motivation will be essential for our later construction for $\omega$-pushdown automata.
Intuitively, the variables in the algebraic system are simulated by states in the simple reset pushdown automaton $\mathfrak{A}_m$. By the Greibach normal form, only two variables on the right-hand side are allowed. The first is modeled directly by changing the state, the second is pushed to the pushdown tape and the state is changed to it later when the variable is popped again. The special final state $f$ will only be used as the last state.

Note that $(M_{x_k,x_k})_{i,f}$ allows the automaton to change to the final state with a non-empty pushdown tape. This is an artificial addition to fit the definition of simple reset pushdown matrices. If the simple reset automaton is not popping a symbol from the pushdown tape, it cannot distinguish between different pushdown states. Even though the automaton can enter the final state too early, it can not continue from there as it is a sink.

Observe that $\|\mathfrak{A}_m\|=((M^*)_{\epsilon,\epsilon})_{s,f}$ for all $1\leq m\leq n$.

This simple reset pushdown matrix $M$ is called the simple pushdown matrix \emph{induced} by the Greibach normal form~\eqref{eqn:2}.
%% no empty line
%% p. 14
The simple reset pushdown automata $\mathfrak{A}_m$, $1\leq m\leq n$, are called the simple reset pushdown automata \emph{induced} by the Greibach normal form~\eqref{eqn:2}.

The following (main) theorem of~\cite{finite_simple} states that the behavior of the simple reset pushdown automata induced by the Greibach normal form~\eqref{eqn:2} is the unique solution of the original algebraic system~\eqref{eqn:2}.
\begin{theorem}[Theorem 11 of~\cite{finite_simple}]
  \label{thm:8}The unique solution of the algebraic system~\eqref{eqn:2} is
  \[(\|\mathfrak{A}_1\|,\dots,\|\mathfrak{A}_n\|)=(((M^*)_{\epsilon,\epsilon})_{1,f},\dots,((M^*)_{\epsilon,\epsilon})_{n,f})\mpunkt\]
\end{theorem}
%% \begin{proof}
%%  We obtain, for $1\leq i\leq n$, by substituting into the right sides of \eqref{eqn:2}
%%   \begin{align*}
%%     &\sum_{1\leq j,k\leq n}\sum_{a\in\Sigma}(p_i a x_j x_k)a((M^*)_{\epsilon,\epsilon})_{j,f}((M^*)_{\epsilon,\epsilon})_{k,f}+\\
%%     &\hspace{4cm}\sum_{1\leq j\leq n}\sum_{a\in\Sigma}(p_i,a x_j)a((M^*)_{\epsilon,\epsilon})_{j,f}+\!\sum_{a\in\Sigma}(p_i,a)a\\
%%     &=\!\!\sum_{1\leq j,k\leq n} (M_{\epsilon,x_k})_{i,j}((M^*)_{x_k,\epsilon})_{j,f}+\!\!\sum_{1\leq j\leq n}(M_{\epsilon,\epsilon})_{i,j}((M^*)_{\epsilon,\epsilon})_{j,f}+(M_{\epsilon,\epsilon})_{i,f}\\
%% %% p. 17
%%     &=\!\!\sum_{1\leq k\leq n}(M_{\epsilon,x_k}(M^*)_{x_k,\epsilon})_{i,f}+(M_{\epsilon,\epsilon}(M^*)_{\epsilon,\epsilon})_{i,f}\\
%%     &= ((M^+)_{\epsilon,\epsilon})_{i,f} = ((M^*)_{\epsilon,\epsilon})_{i,f}\mpunkt
%%   \end{align*}
%% Here in the second equality, we have replaced $(M_{\epsilon,\epsilon})_{i,f}$ by $(M_{\epsilon,\epsilon})_{i,f}((M^*)_{\epsilon,\epsilon})_{f,f}$, since $((M^*)_{\epsilon,\epsilon})_{f,f}=\epsilon$.
%% Since the algebraic system~\eqref{eqn:2} is strict, it has a unique solution.
%% \end{proof}

\begin{corollary}[Corollary 12 of~\cite{finite_simple}]
  \label{cor:9}
  Let $r\in\salg$. Then there exists a simple reset pushdown automaton with behavior $r$.
\end{corollary}

\begin{figure}[t]
  \centering\small
  \begin{tikzpicture}[automaton,node distance=2cm]
    \node[roundnode] (S) {$S$};
    \node[right=of S] (dummy) {};
    \node[roundnode,right=of dummy] (P) {$Q$};
    \node[roundnode,below=4cm of S] (Q) {$R$};
    \node[roundnode,below=4cm of P] (T) {$T$};
    \node[roundnode,below=of dummy] (B) {$B$};
    \node[roundnode,right=4cm of B] (F) {$F$};
    %\node[right=of F] (cont1) {1};
    %\node[below=of cont1] (cont2) {2};
    \draw (F) -- ++ (0.8,0);
    \draw[<-] (S) -- ++ (-0.8,0);
    \path[bend angle=10,bend right] (S) edge node[below] {$a\push S$} (P);
    \path[bend angle=10,bend right] (P) edge node[above] {$b\pop S$} (S);
    \path (S) edge node[left,align=left] {$1a \push T$\\ $1a\internal$} (Q);
    \path (T) edge node[right,pos=0.1,align=left] {$a\push T$\\ $a\internal$} (P);
    %\draw (T) .. controls (cont2) and (cont1) .. (P) node[pos=0.3,above] {$a\push T$};
    \path (Q) edge node[below] {$b \pop T$} (T);
    \path (B) edge node[left=.1cm,near start] {$b \pop S$} (S);
    \path (B) edge node[right,pos=.3] {$b \pop T$} (T);
    \path (Q) edge node[left,near end] {$b \pop B$} (B);
    \path (P) edge node[right=.1cm,near end] {$b \pop B$} (B);
    \path (P) edge node[right] {$b \internal$} (F);
    \path (B) edge node[below,pos=.8] {$b \internal$} (F);
    \path (P) edge [loop right] node[right] {$a \push B$} ();
    \path (Q) edge [loop left] node[left] {$1a \push B$} ();
    \path (B) edge [loop below] node[below] {$b\pop B$} ();
  \end{tikzpicture}
  \caption{Example~\ref{ex:simple_automaton}: Simple reset pushdown automaton, where $\push X$ means push symbol $X$, $\pop X$ means pop $X$, and $\internal$ leaves the stack unaltered. All shown transitions have a weight equal to the natural number 0 except the three transitions going to state $Q$, which have weight 1. All other possible transitions have weight $-\infty$.}
  \label{fig:example_simple_automaton}
\end{figure}

\begin{example}
  \label{ex:simple_automaton}
  Consider the semiring $\bar\N \langle\langle\Sigma^*\rangle\rangle$ for the arctic semiring $\langle \bar\N,\allowbreak\max,\allowbreak+,\allowbreak-\infty,0\rangle$ with $\bar \N=\N\cup\{-\infty,\infty\}$.
Analogously to Example~\ref{ex:finite_simple}, we let $\mathzero = -\infty$ and $\mathone=0$ and we note that in the following, $1$ stands for the natural number $1$.

  We define the algebraic system
  \begin{alignat*}{4}
    S &= a Q S + 1a R T + 1a R\hspace{1cm} & T &= a Q T + a Q\\
    Q &= b + a Q B &  R &= b + 1a R B\\
    B &= b
  \end{alignat*} with the variables $S, T, Q, R, B$. These variables facilitate reading the equations, but for comparison with equation~\eqref{eqn:2}, consider the variable mapping $x_1=T, x_2=S, x_3=R, x_4=Q, x_5=B$.

Now, the variable $Q$ derives a string $a^n b^{n+1}$ for $n\in\N$. The variable $R$ does the same but at the same time produces the weight $n$. The variables $S$ and $T$ add another $a$.

Let $L=\{a^n b^n \mid n\geq 1\}$. %and $(s,w)=0$ if $w\in L$ and $(s,w)=\infty$ if $w\notin L$. Let further $(t,a^n b^n)=n$ and $(t,w)=\infty$ if $w\notin L$.
In total, the second component (i.e., with $S$ being the start variable) of the least solution is $u$ with $(u,a^{n_1}b^{n_1}a^{n_2}b^{n_2}\ldots a^{n_k}b^{n_k})=\max n_i$ for $k\geq 1$ and $(u,w)=-\infty$ for $w\notin L^+$.

From this, we can construct a simple reset pushdown automaton $\mathfrak{A}_2=(n,\Gamma,I,M,P)$ as shown in Figure~\ref{fig:example_simple_automaton}. Thus, we have $n=6$, $\Gamma=\{T,S,R,Q,B\}$. The initial state vector is $I_2=\epsilon$ and $I_i=0$ for $i\neq 2$. The final state vector is $P_6=\epsilon$ and $P_i=0$ for $i\neq 6$. The simple reset pushdown matrix is defined as
\[M=
\begin{pmatrix}
  M_{\epsilon,\epsilon} & M_{\epsilon,T} & M_{\epsilon,S} & M_{\epsilon,R} & M_{\epsilon,Q} & M_{\epsilon,B} & \cdots\\
  M_{T,\epsilon} & M_{\epsilon,\epsilon} & \mathzero & \mathzero & \mathzero & \mathzero & \cdots\\
  M_{S,\epsilon} & \mathzero & M_{\epsilon,\epsilon} & \mathzero & \mathzero & \mathzero  & \cdots\\
  M_{R,\epsilon} & \mathzero & \mathzero & M_{\epsilon,\epsilon} & \mathzero & \mathzero  & \cdots\\
  M_{Q,\epsilon} & \mathzero &  \mathzero & \mathzero & M_{\epsilon,\epsilon} & \mathzero & \cdots\\
  M_{B,\epsilon} & \mathzero &  \mathzero & \mathzero & \mathzero & M_{\epsilon,\epsilon} & \cdots\\
  \vdots & \vdots & \vdots & \vdots & \vdots & \vdots & \ddots
\end{pmatrix}\mkomma\]
with, for instance
\[M_{\epsilon,\epsilon}=
\begin{pmatrix}
  \mathzero & \mathzero & \mathzero & a & \mathzero & \mathzero\\
  \mathzero & \mathzero & 1a & \mathzero & \mathzero & \mathzero\\
  \mathzero & \mathzero & \mathzero & \mathzero & \mathzero & \mathzero\\
  \mathzero & \mathzero & \mathzero & \mathzero & \mathzero & b\\
  \mathzero & \mathzero & \mathzero & \mathzero & \mathzero & b\\
  \mathzero & \mathzero & \mathzero & \mathzero & \mathzero & \mathzero
\end{pmatrix}\text{ and }
M_{\epsilon,B}=
\begin{pmatrix}
  \mathzero & \mathzero & \mathzero & \mathzero & \mathzero & \mathzero\\
  \mathzero & \mathzero & \mathzero & \mathzero & \mathzero & \mathzero\\
  \mathzero & \mathzero & 1a & \mathzero & \mathzero & \mathzero\\
  \mathzero & \mathzero & \mathzero & a & \mathzero & \mathzero\\
  \mathzero & \mathzero & \mathzero & \mathzero & \mathzero & \mathzero\\
  \mathzero & \mathzero & \mathzero & \mathzero & \mathzero & \mathzero
\end{pmatrix}\mpunkt
\]
The rest of the matrix $M$ can be inferred by the rules of pushdown matrices. The behavior $\|\mathfrak{A}_2\|$ is equal to the second component of the least solution of the algebraic system above.
\end{example}

%%%%%%%%%%%%%%%%%%%%%%%%%%%%%%%%%%%%%%%%%%%%%%%%%%%%%%%%%%%%%%%%%%%%%%%%%%%%%%%%%%%%%%
%%%%%%%%%%%%%%%%%%%%%%%%%% new material
%%%%%%%%%%%%%%%%%%%%%%%%%%%%%%%%%%%%%%%%%%%%%%%%%%%%%%%%%%%%%%%%%%%%%%%%%%%%%%%%%%%%%%

\section{Simple \texorpdfstring{$\omega$}{omega}-Reset Pushdown Automata}
\label{sec:simple_omega_reset_pushdown_automata}

In this section, we will prove that for every $\omega$-algebraic series $r$, there exists a simple $\omega$-reset pushdown automaton with behavior $r$. We first introduce some notation and prove an important equality for infinite applications of reset pushdown matrices. Then we introduce simple $\omega$-reset pushdown automata, and the main theorem will show that they can recognize all $\omega$-algebraic series.

%% new:
In the sequel, $(S,V)$ \emph{is a complete semiring-semimodule pair}.

%% p. 18.1
%% From Triple-Pair:
We will use sets $P_l$ comprising infinite sequences over $\{1,\ldots,n\}$ as defined in \cite{triple-pair}:
\[P_l = \{ (j_1, j_2, \dots) \in \{1, \dots, n\}^\omega \mid j_t \leq l \text{ for infinitely many } t \geq 1 \}\mpunkt\]

We obtain, for a reset pushdown matrix $M\in\snngam$, $\pi\in \Gamma^+$ and for $1\leq j\leq n$,
\begin{equation}\label{eqn:pl_identity}
  ((M^{\omega,l})_\pi)_j =
  \sum_{\pi_1,\pi_2,\dots \in \Gamma^*} \sum_{(j_1, j_2, \dots) \in P_l}
  (M_{\pi,\pi_1})_{j,j_1}(M_{\pi_1,\pi_2})_{j_1,j_2}(M_{\pi_2,\pi_3})_{j_2,j_3} \cdots \mpunkt
\end{equation}

Observe the following summation identity: Assume that $M_1, M_2, \dots$ are matrices in $\snn$. Then for $0 \leq l \leq n$, $1 \leq j \leq n$, and $m \geq 1$, we have
\[\sum_{\mathclap{(j_1,j_2, \dots) \in P_l}}\;(M_1)_{j,j_1}(M_2)_{j_1,j_2}\dots=
\;\sum_{\mathclap{1 \leq j_1, \dots, j_m \leq n}}\; (M_1)_{j,j_1}\cdots (M_m)_{j_{m-1},j_m} \sum_{\mathclap{(j_{m+1},j_{m+2}, \dots)\in P_l}}\;(M_{m+1})_{j_m,j_{m+1}} \cdots \mpunkt\]

%% taken mini prelimiaries from here %%%

By Theorem 5.5.1 of Ésik, Kuich \cite{MAT} we obtain, for a finite matrix $M$ and for $0 \leq l \leq n$, the equality $MM^{\omega, l} = M^{\omega,l}$. By Theorem~6 of Droste, Ésik, Kuich~\cite{triple-pair}, we have a similar result for pushdown matrices.
We will now show the same equality for a reset pushdown matrix $M$.

\begin{theorem}
  \label{thm:triple-pair:6}
  Let $(S,V)$ be a complete semiring-semimodule pair and let further $M \in \snngam$ be a reset pushdown transition matrix. Then, for $0 \leq l \leq n$,
  \[M^{\omega, l}= MM^{\omega, l} \mpunkt\]
\end{theorem}

\begin{proof}
  We obtain for $\pi_0 \in \Gamma^*$ and $1 \leq j_0 \leq n$,
  \begin{align*}
    ((MM^{\omega,l})_{\pi_0})_{j_0} &= \sum_{\pi \in \Gamma^*} \sum_{1 \leq j \leq n}(M_{\pi_0,\pi})_{j_0,j}\sum_{{\pi_1,\pi_2, \dots}\in \Gamma^*}\sum_{(j_1,j_2\dots)\in P_l} (M_{\pi,\pi_1})_{j,j_1}(M_{\pi_1,\pi_2})_{j_1,j_2} \dots \\
    &= {\sum_{{\pi,\pi_1,\pi_2 \dots} \in \Gamma^*} \sum_{(j,j_1,j_2,\dots)\in P_l}(M_{\pi_0,\pi})_{j_0,j}(M_{\pi,\pi_1})_{j,j_1}(M_{\pi_1,\pi_2})_{j_1,j_2} \dots}\\
    &= ((M^{\omega,l})_{\pi_0})_{j_0}.\qedhere
  \end{align*}
\end{proof}

%%%%% taken appendix from here

%% p. 27
%We will show that a solution of the forthcoming system~\eqref{eqn:5} is given by the behaviors of certain simple $\omega$-reset pushdown automata.
Next, an $\omega$-reset pushdown automaton
\[\mathfrak{A}=(n,\Gamma,I,M,P,l)\]
is given by a reset pushdown automaton $(n,\Gamma,I,M,P)$ and an integer $l$ with $0\leq l\leq n$, which indicates that $1,\dots,l$ are the repeated states of $\mathfrak{A}$. The behavior $\|\mathfrak{A}\|$ of this $\omega$-reset pushdown automaton $\mathfrak{A}$ is defined by
\[\|\mathfrak{A}\|=I(M^*)_{\epsilon,\epsilon}P+I(M^{\omega,l})_\epsilon\mpunkt\]
The $\omega$-reset pushdown automaton $\mathfrak{A}=(n,\Gamma,I,M,P,l)$ is called \emph{simple} if $M$ is a simple reset pushdown matrix.

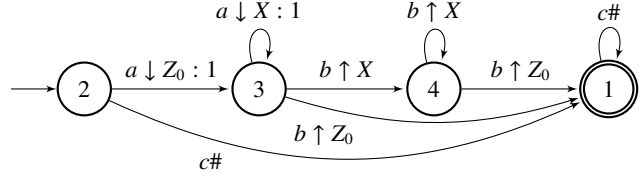
\begin{figure}[t]
  \centering\small
  \begin{tikzpicture}[automaton,node distance=2.3cm]
    \node[roundnode] (2) {2};
    \node[roundnode,right=of 2] (3) {3};
    \node[roundnode,right=of 3] (4) {4};
    \node[roundnode,accepting,right=of 4] (1) {1};
    %\draw (S) -- ++ (0,-1);
    \draw[<-] (2) -- ++ (-1,0);
    \path (2) edge node[above] {$a\push Z_0:1$} (3);
    \path (3) edge node[above] {$b\pop X$} (4);
    \path (4) edge node[above] {$b\pop Z_0$} (1);
    \path (3) edge [loop above] node[above] {$a\push X:1$} ();
    \path (4) edge [loop above] node[above] {$b\pop X$} ();
    \path (1) edge [loop above] node[above] {$c\internal$} ();
    \path[bend angle=25,bend right] (2) edge node[below left,near start] {$c\internal$} (1);
    \path[bend angle=17,bend right] (3) edge node[below left,near start] {$b\pop Z_0$} (1);
  \end{tikzpicture}
  \caption{Example~\ref{ex:automaton}: Simple $\omega$-reset pushdown automaton, where, as above, $\push X$ means push symbol $X$, $\pop X$ means pop $X$, and $\internal$ leaves the stack unaltered. All transitions shown have a weight equal to the natural number 0 except the two transitions going to state $3$ and reading letter $a$, which have weight 1. All other possible transitions have weight $\infty$.}
  \label{fig:example}
\end{figure}

\begin{example}
  \label{ex:automaton}
  Figure~\ref{fig:example} shows a simple $\omega$-reset pushdown automaton $\mathcal{A}=(4,\Gamma,I,M,P,1)$ over the quemiring $\N^\infty \langle\langle\Sigma^*\rangle\rangle\times \N^\infty \langle\langle\Sigma^\omega\rangle\rangle$ for the tropical semiring $\langle \N^\infty,\min,\allowbreak +,\allowbreak\mathzero=\infty,\mathone=0\rangle$ with $\Sigma=\{a,b,c\}$, $\Gamma=\{Z_0,X\}$, $I_2=0$, $I_i=\infty$ for $i\neq 2$ and $P_i=\infty$ for all $1\leq i\leq 4$. Then the adjacency matrix $M$ of the automaton shown in Figure~\ref{fig:example} is a simple reset pushdown matrix. As an indication, $M$ is defined with $(M_{\epsilon,\epsilon})_{1,1}=0c$, $(M_{\epsilon,\epsilon})_{2,1}=0c$, $(M_{\epsilon,Z_0})_{2,3}=1a$, etc., resulting in e.g.,
\[M_{\epsilon,\epsilon}=
\begin{pmatrix}
  0c & \mathzero & \mathzero & \mathzero\\  0c & \mathzero & \mathzero & \mathzero\\ \mathzero & \mathzero & \mathzero & \mathzero\\ \mathzero & \mathzero & \mathzero & \mathzero
\end{pmatrix} \text{ and finally }M=
\begin{pmatrix}
  M_{\epsilon,\epsilon} & M_{\epsilon,Z_0} & M_{\epsilon,X} & \cdots\\
  M_{Z_0,\epsilon} & M_{\epsilon,\epsilon} & \mathzero & \cdots\\
  M_{X,\epsilon} & \mathzero & M_{\epsilon,\epsilon}  & \cdots\\
  \vdots & \vdots & \vdots & \ddots
\end{pmatrix}\mkomma\]
where the excluded part of $M$ can be derived from the rules of pushdown and simple reset pushdown matrices.
The automaton $\mathcal{A}$ has the behavior $a^n b^n c^\omega\mapsto n$, similar to the mixed $\omega$-algebraic system in Example~\ref{ex:mixed_algebraic_system}.
\end{example}

\begin{example}
  \label{ex:omega_automaton}
  Reconsider Example~\ref{ex:simple_automaton}. We define the simple $\omega$-reset pushdown automaton $\mathfrak{A}_2=(6,\Gamma,I,\allowbreak M,P,1)$ where we define the state ordering $T,S,Q,P,B,F$ to make state $T$ Büchi-accepting. The behavior in the semiring part is equal to before; the behavior in the semimodule part is $u$ with $(u,a^{n_1}b^{n_1}\allowbreak a^{n_2}b^{n_2}\ldots)=\max n_i$ and $(u,w)=-\infty$ for $w\notin \{a^n b^n \mid n\geq 1\}^\omega$.
\end{example}

\begin{example}
  %  We now reinterpret the algebraic system in Example~\ref{ex:simple_automaton} as an $\omega$-algebraic system and are interested in the first canonical solution, i.e., with variable $T$ Büchi-accepting. The corresponding simple $\omega$-reset pushdown automaton is the one described in Example~\ref{ex:omega_automaton}, i.e., the same as in Figure~\ref{fig:example_simple_automaton} but with state $T$ Büchi-accepting.

  Consider the $\omega$-algebraic system
  \begin{equation}\label{eqn:example_unmixed_system}
    \begin{aligned}
      y_1 &= a + c y_1\\
      y_2 &= a y_1 y_2 + a y_1\mpunkt
    \end{aligned}
  \end{equation}
  We will consider the second component of the first canonical solution, i.e., variable $y_1$ is Büchi-accepting and variable $y_2$ is the start variable.

  The $\omega$-algebraic system induces the following mixed $\omega$-algebraic system
  \begin{equation}\label{eqn:example_mixed_system}
    \begin{alignedat}{4}
      x_1 &= a + c x_1                &    z_1 &= c z_1\\
      x_2 &= a x_1 x_2 + a x_1\qquad  &    z_2 &= a z_1 + a x_1 z_2\mpunkt
    \end{alignedat}
  \end{equation}

  The least solution of $x=p(x)$ is
  \[\sigma=
  \begin{pmatrix}
    c^* a\\ (a c^* a)^+
  \end{pmatrix}\mpunkt\]

  Now, we write the linear system $z = \varrho(\sigma) z$ in the matrix form and compute the first canonical solution.
  \begin{align*}
    \varrho(\sigma)^{\omega,1}&=
    \begin{pmatrix}
      c & 0\\
      a & a (c^* a)
    \end{pmatrix}^{\omega,1}\\
    &= \begin{pmatrix}
      (c+0(a a)^* a)^\omega\\ (a c^* a)^* a (c+0(a a)^* a)^\omega
    \end{pmatrix}\\
    &= \begin{pmatrix}
      c^\omega\\ (a c^* a)^* a c^\omega
    \end{pmatrix}\eqqcolon
    \begin{pmatrix}
      \omega^{(1)}_1\\ \omega^{(1)}_2
    \end{pmatrix}=
    \omega^{(1)}
  \end{align*}
  Note that the second component, $\omega^{(1)}_2$, does not contain the $\omega$-words $(a c^* a)^\omega$ even though for an unweighted $\omega$-context-free grammar corresponding to \eqref{eqn:example_unmixed_system}, the derivation
  \[y_2\to a y_1 y_2\to (a a) y_2 \to (a a) a y_1 y_2 \to (a a)^2 y_2 \to^\omega a^\omega\]
  would be successful even with only $y_1$ Büchi-accepting. The difference is due to the fact that $y_1$ is not \emph{significant} in the $\omega$-algebraic system above, i.e., $y_1$ in \eqref{eqn:example_unmixed_system} is exchanged by $x_1$ in the mixed $\omega$-algebraic system \eqref{eqn:example_mixed_system} and can therefore no longer be considered as Büchi-accepting variable in $\varrho(\sigma)^{\omega,1}$ (for more information, see~\cite{MAT} pp.\@ 140 ff.).

  Now, we look at the simple $\omega$-reset pushdown automaton induced by $\omega$-algebraic system \eqref{eqn:example_unmixed_system}:
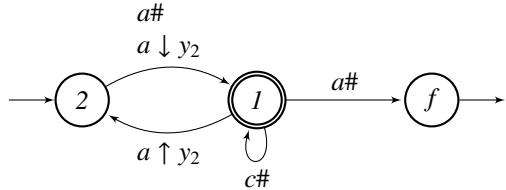
  \begin{center}
    \begin{tikzpicture}[automaton,node distance=2.3cm]
      \node[roundnode] (2) {2};
      \node[roundnode,accepting,right=of 2] (1) {1};
      \node[roundnode,right=of 1] (f) {f};
      \draw (f) -- ++ (1,0);
      \draw[<-] (2) -- ++ (-1,0);
      \path[bend left] (2) edge node[above,align=left] {$a\internal$\\$a\push y_2$} (1);
      \path[bend left] (1) edge node[below] {$a\pop y_2$} (2);
      \path (1) edge node[above] {$a\internal$} (f);
      \path (1) edge [loop below] node[below] {$c\internal$} ();
    \end{tikzpicture}
  \end{center}
  The behavior of this automaton is
  \begin{alignat*}{8}
    &(&((M^*)_{\epsilon,\epsilon})_{1,f},&\;&((M^*)_{\epsilon,\epsilon})_{2,f};&\;&((M^{\omega,1})_\epsilon)_1,&\;&((M^{\omega,1})_\epsilon)_2)&\\
    =\,&(&c^* a,&&(a c^* a)^+;&&c^\omega,&\quad&(a c^* a)^* a c^\omega+(a c^* a)^\omega)
  \end{alignat*}
  Here, the first two components are equal to $\sigma$, as desired. But the last component differs from $\omega^{(1)}_2$; the last component is however equal to the behavior of unweighted $\omega$-context-free grammars.

  Note that the desired component $\omega^{(1)}_2=(a c^* a)^* a c^\omega$ is not recognized by this automaton, even when changing the Büchi-accepting states. If no states are Büchi-accepting, the behavior is 0, if all of them are Büchi-accepting, we have the same behavior as above. If only state 2 is Büchi-accepting (can be achieved by renaming), we only recognize $(a c^* a)^\omega$.

  We now propose a different construction; this new construction models exactly the canonical solutions of mixed $\omega$-algebraic systems. The following is the simple $\omega$-reset pushdown automaton induced by the \emph{mixed} $\omega$-algebraic system~\eqref{eqn:example_mixed_system}; this new construction will be defined after the example. Basically, the construction is similar to the old construction but it differentiates between variables $x$ and $z$; it therefore uses the states $x_1,\ldots,x_n,z_1,\ldots,z_n$:
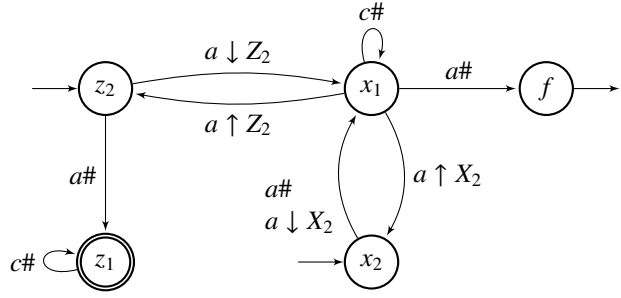
  \begin{center}
    \begin{tikzpicture}[automaton,node distance=2.3cm]
      \node[roundnode] (z2) {$z_2$};
      \node[roundnode,accepting,below=of z2] (z1) {$z_1$};
      \node[roundnode,right=3.5cm of z2] (x1) {$x_1$};
      \node[roundnode,below=of x1] (x2) {$x_2$};
      \node[roundnode,right=of x1] (f) {$f$};
      \draw (f) -- ++ (1,0);
      \draw[<-] (x2) -- ++ (-1,0);
      \draw[<-] (z2) -- ++ (-1,0);
      \path[bend left] (x2) edge node[left,align=left,near start] {$a\internal$\\$a\push X_2$} (x1);
      \path[bend left] (x1) edge node[right] {$a\pop X_2$} (x2);
      \path (x1) edge node[above] {$a\internal$} (f);
      \path (x1) edge [loop above] node[above] {$c\internal$} ();
      \path (z2) edge node[left] {$a\internal$} (z1);
      \path[bend left=10] (z2) edge node[above] {$a\push Z_2$} (x1);
      \path[bend left=10] (x1) edge node[below] {$a\pop Z_2$} (z2);
      \path (z1) edge [loop left] node[left] {$c\internal$} ();
    \end{tikzpicture}
  \end{center}
  This simple $\omega$-reset pushdown automaton has exactly the behavior $(\sigma,\omega^{(1)})$. This means, if only $z_1$ is Büchi-accepting, then the automaton does not allow the run $(a c^* a)^\omega$.

  The rest of the paper will show that in general, the $l$\textsuperscript{th} canonical solution of a mixed $\omega$-algebraic system $x=p(x), z=\varrho(x)z$ is exactly the behavior of the simple $\omega$-reset pushdown automaton induced by $x=p(x), z=\varrho(x)z$.
\end{example}

\mbox{}\phantomsection
\label{pagemark_discussions}
%% my additions
Given a series $r\in \salgtimes$, we want to construct a simple $\omega$-reset pushdown automaton with behavior $r$. %By Theorems~5.10 and 5.4 of \cite{kuich_semirings}, $r$ is a component of the unique solution of a strict algebraic system in Greibach normal form.
%% p. 28
By Theorem~\ref{thm:unmix} and Theorem~\ref{thm:3.1}, $r$ is a component of a canonical solution of an $\omega$-algebraic system~\eqref{eqn:5new} (compare this to the algebraic system~\eqref{eqn:2}) in Greibach normal form over the quemiring $\stimes$,
\begin{align}
  \label{eqn:5new}
    y_i= \sum_{1\leq j,k\leq n}\sum_{a\in\Sigma}(p_i,a y_j y_k)a y_j y_k+\sum_{1\leq j\leq n}\sum_{a\in\Sigma}(p_i,a y_j)a y_j+\sum_{a\in\Sigma}(p_i,a)a\mpunkt
\end{align}

The variables of this system are $y_i$, ($1\leq i\leq n$); they are variables for $\stimes$. The system~\eqref{eqn:5new} induces the following mixed $\omega$-algebraic system:
\begin{align}
  \label{eqn:2againnew}
    x_i &= \sum_{1\leq j,k\leq n}\sum_{a\in\Sigma}(p_i,a y_j y_k)a x_j x_k+\sum_{1\leq j\leq n}\sum_{a\in\Sigma}(p_i,a y_j)a x_j+\sum_{a\in\Sigma}(p_i,a)a\mpunkt
\intertext{and}
  \label{eqn:6new}
  z_i &= \sum_{1\leq j,k\leq n}\sum_{a\in\Sigma}(p_i, a y_j y_k)a(z_j+x_j z_k)+\sum_{1\leq j\leq n}\sum_{a\in\Sigma}(p_i,a y_j)a z_j
\end{align}

But this system hides information, for instance, $y_j y_k$ will never be derived by two consecutive variables $z_j z_k$ of $\ssigomega$. Our new construction is therefore based on the following mixed $\omega$-algebraic system:
\begin{align}
  %\tag{\ref{eqn:2}} TODO was that nice?
  \tag{$\lozenge$}
  \label{eqn:new_sys1}
    x_i &= \sum_{1\leq j,k\leq n}\sum_{a\in\Sigma}(p_i,a x_j x_k)a x_j x_k+\sum_{1\leq j\leq n}\sum_{a\in\Sigma}(p_i,a x_j)a x_j+\sum_{a\in\Sigma}(p_i,a)a\mpunkt
\intertext{and}
  \tag{$\lozenge\lozenge$}
  \label{eqn:new_sys2}
  z_i &= \sum_{1\leq j,k\leq n}\sum_{a\in\Sigma}(p_i, a x_j z_k)a x_j z_k+\sum_{1\leq j\leq n}\sum_{a\in\Sigma}(p_i,a z_j)a z_j
\end{align}
The new system \eqref{eqn:new_sys1}, \eqref{eqn:new_sys2} can be gained from the last system \eqref{eqn:2againnew}, \eqref{eqn:6new} by renaming; all new coefficients can easily be transferred except one: for all $1\leq i,j\leq n$, we set $(p_i,a z_j) = (p_i,a y_j)+\sum_{1\leq k\leq n} (p_i,a y_j y_k)$.

Note that the algebraic systems \eqref{eqn:new_sys1} and \eqref{eqn:2} are equivalent. %\TODO{did reuse equation number before..}

%\TODO{Instead, we could start the presentation directly with Eq.\@ \eqref{eqn:new_sys1}, \eqref{eqn:new_sys2}. Then, we only need the Greibach NF for mixed systems.}
Also note that we could start the presentation directly with the system \eqref{eqn:new_sys1}, \eqref{eqn:new_sys2} by applying Theorem~\ref{thm:greibach} instead of starting with system \eqref{eqn:5new}, \eqref{eqn:2againnew} and applying Theorem~\ref{thm:unmix}. We decided for this presentation because the mixed $\omega$-algebraic systems do not have a counterpart in unweighted automata theory and therefore, we believe it more natural to start by an $\omega$-algebraic system and constructing our simple $\omega$-reset pushdown automaton from there.

%% p. 13
We now want to construct a simple $\omega$-reset pushdown automaton. Here, we introduce our new construction. Let $\mathfrak{A}_m^l=(2n+1,\Gamma,I_m,M,\allowbreak P,l)$, $1\leq m\leq n$, $0\leq l\leq n$, be defined as follows:\\
We let $\Gamma=\{X_1,\dots,X_n,Z_1,\dots,Z_n\}$; we denote the states $1,\dots,2n+1$ by $z_1,\dots,z_n,x_1,\dots,x_n,f$; the entries of $M$ of the form $(M_{\pi,\pi'})_{v,v'}$ for $1\leq v,v'\leq 2n+1$ and for $\pi,\pi'\in\Gamma^*$ with $|\pi|,|\pi'|\leq 1$
%$(M_{x_k,x_k})_{x_i,x_j}$, $(M_{x_k,\epsilon})_{x_i,x_j}$, $(M_{\epsilon,x_k})_{x_i,x_j}$, $(M_{\epsilon,\epsilon})_{x_i,x_j}$, $(M_{\epsilon,\epsilon})_{x_i,f}$, $(M_{z_k,\epsilon})_{x_i,z_j}$, $(M_{\epsilon,\epsilon})_{z_i,z_j}$, $(M_{z_k,z_k})_{z_i,z_j}$, $(M_{\epsilon,z_k})_{z_i,z_j}$, $1\leq i,j,k\leq n$, (\TODO{can that be simplified?})
that may be unequal to $0$ are
\begin{align*}
  (M_{\epsilon,X_k})_{x_i,x_j} &= \sum_{a\in\Sigma}(p_i,a x_j x_k)a\mkomma\\
  (M_{Z_k,Z_k})_{x_i,x_j} = (M_{X_k,X_k})_{x_i,x_j} = (M_{\epsilon,\epsilon})_{x_i,x_j} &= \sum_{a\in\Sigma}(p_i,a x_j)a\mkomma\\
  (M_{Z_k,\epsilon})_{x_i,z_k} = (M_{X_k,\epsilon})_{x_i,x_k} = (M_{Z_k,Z_k})_{x_i,f} = (M_{X_k,X_k})_{x_i,f} = (M_{\epsilon,\epsilon})_{x_i,f} &= \sum_{a\in\Sigma}(p_i,a)a\mkomma\\
  (M_{Z_k,Z_k})_{z_i,z_j} = (M_{X_k,X_k})_{z_i,z_j} = (M_{\epsilon,\epsilon})_{z_i,z_j} &= \sum_{a\in\Sigma}(p_i,a z_j)a\mkomma\\
  (M_{\epsilon,Z_k})_{z_i,x_j} &= \sum_{a\in\Sigma}(p_i,a x_j z_k)a\mkomma
\end{align*}
for $1\leq i, j, k\leq n$; we further put $(I_m)_{x_m}=(I_m)_{z_m}=\epsilon$, and  $(I_m)_{x_i}=(I_m)_{z_i}=0$ for $1\leq i\leq m-1$ and $m+1\leq i \leq n$ and $(I_m)_f=0$; finally let $P_f=\epsilon$ and $P_j=0$ for $1\leq j\leq 2n$;

In the following, we assume that $r\in \salgtimes$ is the $m$\textsuperscript{th} component of the $l$\textsuperscript{th} canonical solution of~\eqref{eqn:5new}.
We want to show that for the $l$\textsuperscript{th} canonical solution $\tau=(\sigma,\omega)$ of \eqref{eqn:new_sys1}, \eqref{eqn:new_sys2}, and therefore also of \eqref{eqn:5new}, we have $\tau_m=\sigma_m+\omega_m=\|\mathfrak{A}_m^l\|$.

This simple reset pushdown matrix $M$ is called the simple reset pushdown matrix \emph{induced} by the Greibach normal form~\eqref{eqn:new_sys1}, \eqref{eqn:new_sys2}.
The simple $\omega$-reset pushdown automata $\mathfrak{A}_m^l$ ($1\leq m\leq n$, $0\leq l\leq n$) are called the simple $\omega$-reset pushdown automata \emph{induced} by the Greibach normal form~\eqref{eqn:new_sys1}, \eqref{eqn:new_sys2}.

For the rest of the paper, we will use the following notation (cf.\@ \cite{kuich_formal_power_series}, page 179).
Note that $M\in(S^{k \times k})^{\Gamma^* \times \Gamma^*}$ for $k=2n+1$. By isomorphism, we can transform this into $\what M\in(S^{\Gamma^* \times \Gamma^*})^{k \times k}$. We then have $(M_{\pi,\pi'})_{v,v'}=(\what M_{v,v'})_{\pi,\pi'}$ for $\pi,\pi'\in\Gamma^*$ and $1\leq v,v'\leq 2n+1$. (By the notation $1\leq v\leq 2n+1$, we mean $v$ can be any of the states $z_1,\dots,z_n,x_1,\dots,x_n,f$.) 

\begin{example}
This notation allows us to add up matrices with suitable pushdown indexes while still keeping the information of the states. For instance, note that
\[\sum_{1\leq k\leq n}\sum_{\pi\in\Gamma^*}(\what M_{z_i,x_k})_{\epsilon,\pi} (\what M_{x_{k},z_j})_{\pi,\epsilon}=\sum_{1\leq k\leq n}\big(\what M_{z_i,x_k} \what M_{x_k,z_j}\big)_{\epsilon,\epsilon}\mpunkt\]
Now consider the term
\[\sum_{1\leq k\leq n}\sum_{\pi\in\Gamma^*}(M_{\epsilon,\pi})_{z_i,x_k} (M_{\pi,\epsilon})_{x_k,z_j}\mkomma\]
which cannot be simplified because $\sum_{\pi\in\Gamma^*}(M_{\epsilon,\pi}M_{\pi,\epsilon})_{z_i,z_j}$
does no longer hold the information that the path passes only through states $x_i$, i.e., it contains also the path $(M_{\epsilon,\pi})_{z_i,z_k} (M_{\pi,\epsilon})_{z_k,z_j}$ (for all $1\leq k\leq n$). In the proofs below, we will specifically need to distinguish paths that pass through states $x_i$ and those that pass through states $z_i$ as in the mixed $\omega$-algebraic system, we also distinguish between variables $x_i$ for finite derivations and variables $z_i$ for infinite derivations.
\end{example}

\begin{lemma}
  \label{lem:hat_star}
  Let $M\in(S^{k \times k})^{\Gamma^* \times \Gamma^*}$ be a reset pushdown matrix. Then,
  \[\widehat {M^*}={\what M}^*\mpunkt\]
\end{lemma}
\begin{proof}
  For $1\leq v,v'\leq k$ and for $\pi,\pi'\in\Gamma^*$, we obtain
  \begin{align*}
    ((\widehat {M^*})_{v,v'})_{\pi,\pi'} &= ((M^*)_{\pi,\pi'})_{v,v'}\\
    &=\sum_{n\geq 0}((M^n)_{\pi,\pi'})_{v,v'}\\
    &=\sum_{n\geq 0}((\what M^n)_{v,v'})_{\pi,\pi'}\\
    &=((\what M^*)_{v,v'})_{\pi,\pi'}\mpunkt\qedhere
  \end{align*}
\end{proof}

Similarly, we need the above result for another operator.
\begin{lemma}
  \label{lem:hat_omega}
  Let $M\in(S^{k \times k})^{\Gamma^* \times \Gamma^*}$ be a reset pushdown matrix. Then, for $1\leq l\leq k$,
  \[\widehat {M^{\omega,l}}={\what M}^{\omega,l}\mpunkt\]
\end{lemma}
\begin{proof}
  For $1\leq v\leq k$ and for $\pi\in\Gamma^*$, we obtain
  \begin{align*}
    ((\widehat {M^{\omega,l}})_v)_\pi &= ((M^{\omega,l})_\pi)_v\\
    &=\sum_{\pi_1,\pi_2,\ldots\in \Gamma^*}\sum_{(v_1,v_2,\ldots)\in P_l}(M_{\pi,\pi_1})_{v,v_1}(M_{\pi_1,\pi_2})_{v_1,v_2}(M_{\pi_2,\pi_3})_{v_2,v_3}\cdots\\
    &=\sum_{\pi_1,\pi_2,\ldots\in \Gamma^*}\sum_{(v_1,v_2,\ldots)\in P_l}(\what M_{v,v_1})_{\pi,\pi_1}(\what M_{v_1,v_2})_{\pi_1,\pi_2}(\what M_{v_2,v_3})_{\pi_2,\pi_3}\cdots\\
    &=((\what M^{\omega,l})_v)_\pi\mpunkt\qedhere
  \end{align*}
\end{proof}

Let $M$ be a simple reset pushdown matrix induced by the Greibach normal form~\eqref{eqn:new_sys1}, \eqref{eqn:new_sys2}. We define some blocks of the matrix $\what M$ to make the following argumentation easier. We take the idea of the above-mentioned isomorphism and divide $\what M$ like
\begin{equation}\what M=
  \begin{pmatrix}
    \what M_{z,z} & \what M_{z,x} & 0\\
    \what M_{x,z} & \what M_{x,x} & \what M_{x,f}\\
    0 & 0 & 0
  \end{pmatrix}\mkomma\label{eqn:induced_blocks}
\end{equation}
where the respective blocks are defined as
\begin{alignat*}{4}
  \what M_{z,z}=\begin{pmatrix}
  \what M_{z_1,z_1} & \cdots & \what M_{z_1,z_n}\\
  \vdots & \ddots & \vdots\\
  \what M_{z_n,z_1} & \cdots & \what M_{z_n,z_n}
  \end{pmatrix}\mkomma\quad & 
  \what M_{z,x}=\begin{pmatrix}
  \what M_{z_1,x_1} & \cdots & \what M_{z_1,x_n}\\
  \vdots & \ddots & \vdots\\
  \what M_{z_n,x_1} & \cdots & \what M_{z_n,x_n}
  \end{pmatrix}\mkomma&&\\
  \what M_{x,z}=\begin{pmatrix}
  \what M_{x_1,z_1} & \cdots & \what M_{x_1,z_n}\\
  \vdots & \ddots & \vdots\\
  \what M_{x_n,z_1} & \cdots & \what M_{x_n,z_n}
  \end{pmatrix}\mkomma\quad &
  \what M_{x,x}=\begin{pmatrix}
  \what M_{x_1,x_1} & \cdots & \what M_{x_1,x_n}\\
  \vdots & \ddots & \vdots\\
  \what M_{x_n,x_1} & \cdots & \what M_{x_n,x_n}
  \end{pmatrix}\mkomma\quad &&
  \what M_{x,f}=\begin{pmatrix}
  \what M_{x_1,f}\\
  \vdots\\
  \what M_{x_n,f}
  \end{pmatrix}\mkomma
\end{alignat*}
and where each $\what M_{v,v'}\in S^{\Gamma^* \times \Gamma^*}$ for $1\leq v,v'\leq 2n+1$. For notational convenience, we also set
\begin{alignat*}{4}
  \what M_{z_i,x}=
  \begin{pmatrix}
    \what M_{z_i,x_1} & \cdots & \what M_{z_i,x_n}
  \end{pmatrix}\mkomma\qquad &&
  \what M_{x,z_i}=
  \begin{pmatrix}
    \what M_{x_1,z_i}\\
    \vdots\\
    \what M_{x_n,z_i}
  \end{pmatrix}\mpunkt
\end{alignat*}

Note that we have not defined the blocks $\what M_{z,f}$, $\what M_{f,z}$, $\what M_{f,x}$ and $\what M_{f,f}$ as they would all be zero by our construction for simple reset pushdown matrices induced by the Greibach normal form~\eqref{eqn:new_sys1}, \eqref{eqn:new_sys2}.

Analogously, let $M_{z,z}, M_{z,x}, M_{x,z}, M_{x,x},\allowbreak M_{x,f} \in (S^{(2n+1) \times (2n+1)})^{\Gamma^* \times \Gamma^*}$ be the isomorphic copy of $\what M_{z,z},\allowbreak \what M_{z,x}, \what M_{x,z}, \what M_{x,x},\allowbreak \what M_{x,f}$, respectively. Then, for $u,v\in\{x,z\}$ and for $\pi, \pi'\in\Gamma^*$, the matrix $(M_{u,v})_{\pi, \pi'}$ is $M_{\pi, \pi'}$ restricted to the variables $u_i,v_j$ (for $1\leq i,j\leq n$). Similarly, $M_{x,f}$ is $M$ restricted to variables $x_i,f$ (for $1\leq i\leq n$).
For instance, $(\what M_{x,x})^*$ and equally $(M_{x,x})^*$ consider only paths passing through states $x_i$ and no paths through $z_i$ or $f$ (for $1\leq i\leq n$). Their only difference is the order of indexes.

The following theorem computes the behavior of induced simple $\omega$-reset pushdown automata.

\begin{theorem}
  \label{thm:induced_omega}
  Let $M$ be a simple reset pushdown matrix induced by the Greibach normal form~\eqref{eqn:new_sys1}, \eqref{eqn:new_sys2}. Then, for all $1\leq i\leq n$ and $0\leq l\leq n$,
  \[\bigl((M^{\omega,l})_\epsilon\bigr)_{x_i}=\bigl((M^{\omega,l})_\epsilon\bigr)_f=0\mkomma\]
  and
  \[\bigl((M^{\omega,l})_\epsilon\bigr)_{z_i}=\Bigl(\bigl(\bigl(M_{z,z}+M_{z,x} (M^*)_{x,x} M_{x,z}\bigr)^{\omega,l}\bigr)_\epsilon\Bigr)_i\mpunkt\]
\end{theorem}
\begin{proof}
  For the matrix $M$, we have, by above notation~\eqref{eqn:induced_blocks},
  \begin{align*}
    \what M &=
    \begin{pmatrix}
      \what M_{z,z} & \what M_{z,x} & 0\\
      \what M_{x,z} & \what M_{x,x} & \what M_{x,f}\\
      0 & 0 & 0
    \end{pmatrix} =\left(
    \begin{array}{c|c}
      \what M_{z,z} & \begin{matrix}
        \what M_{z,x} & 0
      \end{matrix}\\
      \hline\\[-1.5\medskipamount]
      \begin{matrix}
        \what M_{x,z} \\ 0
      \end{matrix} &
      \begin{matrix}
        \what M_{x,x} & \what M_{x,f}\\
        0 & 0
      \end{matrix}
    \end{array}
    \right)\mpunkt
  \end{align*}
Thus, by Theorem~\ref{thm:omegak} and Lemma~\ref{lem:hat_omega}, we obtain
\[M^{\omega,l} = \begin{pmatrix}
      \alpha^{\omega,l}\\
      \begin{pmatrix}
        M_{x,x} & M_{x,f}\\
        0 & 0
      \end{pmatrix}^*
      \begin{pmatrix}
        M_{x,z}\\
        0
      \end{pmatrix}
      \alpha^{\omega,l}
    \end{pmatrix}\mkomma\]
where
  \begin{align}
    \alpha &= M_{z,z}+
    \begin{pmatrix}
      M_{z,x} & 0
    \end{pmatrix}
    \begin{pmatrix}
      M_{x,x} & M_{x,f}\notag\\
      0 & 0
    \end{pmatrix}^*
    \begin{pmatrix}
      M_{x,z}\\
      0
    \end{pmatrix}\notag\\
    &= M_{z,z}+
    \begin{pmatrix}
      M_{z,x} & 0
    \end{pmatrix}
    \begin{pmatrix}
      (M_{x,x})^* & (M_{x,x})^* M_{x,f}\\
      0 & 1
    \end{pmatrix}
    \begin{pmatrix}
      M_{x,z}\\
      0
    \end{pmatrix}\notag\\
    &= M_{z,z}+
    \begin{pmatrix}
      M_{z,x} (M_{x,x})^* & M_{z,x}(M_{x,x})^* M_{x,f}
    \end{pmatrix}
    \begin{pmatrix}
      M_{x,z}\\
      0
    \end{pmatrix}\notag\\
    &= M_{z,z}+M_{z,x} (M_{x,x})^* M_{x,z}\mpunkt\label{eqn:star_induced_matrix}
  \end{align}

  Now, we continue with the term from before and get
  \begin{align*}
    M^{\omega,l}
    &= \begin{pmatrix}
      \alpha^{\omega,l}\\
      \begin{pmatrix}
        M_{x,x} & M_{x,f}\\
        0 & 0
      \end{pmatrix}^*
      \begin{pmatrix}
        M_{x,z}\\
        0
      \end{pmatrix}
      \alpha^{\omega,l}\\
    \end{pmatrix}\\
    &= \begin{pmatrix}
      (M_{z,z}+M_{z,x} (M_{x,x})^* M_{x,z})^{\omega,l}\\
      \begin{pmatrix}
        M_{x,x} & M_{x,f}\\
        0 & 0
      \end{pmatrix}^*
      \begin{pmatrix}
        M_{x,z}\\
        0
      \end{pmatrix}
      (M_{z,z}+M_{z,x} (M_{x,x})^* M_{x,z})^{\omega,l}
    \end{pmatrix}\\
    &= \begin{pmatrix}
      (M_{z,z}+M_{z,x} (M_{x,x})^* M_{x,z})^{\omega,l}\\
      \begin{pmatrix}
        (M_{x,x})^* & (M_{x,x})^*M_{x,f}\\
        0 & 1
      \end{pmatrix}
      \begin{pmatrix}
        M_{x,z}\\
        0
      \end{pmatrix}
      (M_{z,z}+M_{z,x} (M_{x,x})^* M_{x,z})^{\omega,l}
    \end{pmatrix}\\
    &= \begin{pmatrix}
      (M_{z,z}+M_{z,x} (M_{x,x})^* M_{x,z})^{\omega,l}\\
      \begin{pmatrix}
        (M_{x,x})^*M_{x,z}\\
        0
      \end{pmatrix}
      (M_{z,z}+M_{z,x} (M_{x,x})^* M_{x,z})^{\omega,l}
    \end{pmatrix}\\
    &= \begin{pmatrix}
      \big(M_{z,z}+M_{z,x} (M_{x,x})^* M_{x,z}\big)^{\omega,l}\\
      \big((M_{x,x})^*M_{x,z}\big) \big(M_{z,z}+M_{z,x} (M_{x,x})^* M_{x,z}\big)^{\omega,l}\\
      0
    \end{pmatrix}\mpunkt
  \end{align*}

  Then, we start the run of the automaton with an empty stack and get
  \begin{align*}
    (M^{\omega,l})_\epsilon
    &= \begin{pmatrix}
      \big(M_{z,z}+M_{z,x} (M_{x,x})^* M_{x,z}\big)^{\omega,l}\\
      \big((M_{x,x})^*M_{x,z}\big) \big(M_{z,z}+M_{z,x} (M_{x,x})^* M_{x,z}\big)^{\omega,l}\\
      0
    \end{pmatrix}_\epsilon\\
    &= \begin{pmatrix}
      \bigl(\bigl(M_{z,z}+M_{z,x} (M_{x,x})^* M_{x,z}\bigr)^{\omega,l}\bigr)_\epsilon\\
      \Big(\big((M_{x,x})^*M_{x,z}\big) \big(M_{z,z}+M_{z,x} (M_{x,x})^* M_{x,z}\big)^{\omega,l}\Big)_\epsilon\\
      0
    \end{pmatrix}\allowdisplaybreaks\\
    &= \begin{pmatrix}
      \bigl(\bigl(M_{z,z}+M_{z,x} (M_{x,x})^* M_{x,z}\bigr)^{\omega,l}\bigr)_\epsilon\\
      \sum_{\pi\in\Gamma^*}\big((M_{x,x})^*M_{x,z}\big)_{\epsilon,\pi} \bigl(\bigl(M_{z,z}+M_{z,x} (M_{x,x})^* M_{x,z}\bigr)^{\omega,l}\bigr)_\pi\\
      0
    \end{pmatrix}\\
    &\stackrel{4}{=} \begin{pmatrix}
      \bigl(\bigl(M_{z,z}+M_{z,x} (M_{x,x})^* M_{x,z}\bigr)^{\omega,l}\bigr)_\epsilon\\
      \sum_{\pi\in\Gamma^*}0\; \bigl(\bigl(M_{z,z}+M_{z,x} (M_{x,x})^* M_{x,z}\bigr)^{\omega,l}\bigr)_\pi\\
      0
    \end{pmatrix}\\
    &= \begin{pmatrix}
      \bigr(\bigr(M_{z,z}+M_{z,x} (M_{x,x})^* M_{x,z}\bigr)^{\omega,l}\bigr)_\epsilon\\
      0\\
      0
    \end{pmatrix}
  \end{align*}
  where the fourth equality uses the fact that $((M_{x,x})^*M_{x,z})_{\epsilon,\pi}=0$, which is because $(M_{x,z})_{\pi,\pi'}=0$ for all $\pi\neq Z_k\pi''$ ($1\leq k\leq n$ and $\pi''\in\Gamma^*$) and at the same time, $((M_{x,x})^*)_{\epsilon,Z_k\pi''}=0$ because only $(M_{z,x})_{\epsilon,Z_k}\neq 0$ by construction.

  The vector $(M^{\omega,l})_\epsilon$ is indexed by $z_1,\dots,z_n,x_1,\dots,x_n,f$, thus completing the proof.
\end{proof}

We want to apply the results from Section~\ref{sec:simple_reset_pushdown_automata}.
The following three lemmas investigate the star operation applied to simple reset pushdown matrices $M$ induced by the Greibach normal form~\eqref{eqn:new_sys1}, \eqref{eqn:new_sys2}.
The lemmas state that in a computation $(M^*)_{\epsilon,\epsilon}$, the new states $z_k$ are never reached when starting in a state $x_i$ and therefore, these computations are equivalent to the computations $(M'^*)_{\epsilon,\epsilon}$ for $M'$ being induced by the Greibach normal form~~\eqref{eqn:2}, i.e., for $M'$ built by the old construction.

\begin{lemma}
  \label{lem:induced_star_x}
  Let $M$ be a simple reset pushdown matrix induced by the Greibach normal form~\eqref{eqn:new_sys1}, \eqref{eqn:new_sys2}. Then, for all $1\leq i,k\leq n$,
  \[((M^*)_{\epsilon,\epsilon})_{x_k,x_i}=(((M_{x,x})^*)_{\epsilon,\epsilon})_{x_k,x_i}\mpunkt\]
\end{lemma}
\begin{proof}
  Let $\Delta=\{X_1,\ldots,X_n\}$. We have
  \begin{align*}
    ((M^*)_{\epsilon,\epsilon})_{x_k,x_i} &= \sum_{t\geq 0}((M^t)_{\epsilon,\epsilon})_{x_k,x_i}\\
    &= \sum_{t\geq 0}\sum_{\pi_1,\ldots,\pi_{t-1}\in\Gamma^*} \Bigl(M_{\epsilon,\pi_1}M_{\pi_1,\pi_2}\cdots M_{\pi_{t-1},\epsilon}\Bigr)_{x_k,x_i}\\
    &= \sum_{t\geq 0}\sum_{\substack{\pi_1\in\Delta^*\\ \pi_2,\ldots,\pi_{t-1}\in\Gamma^*}}\sum_{1\leq j_1\leq n} (M_{\epsilon,\pi_1})_{x_k,x_{j_1}}\Bigl(M_{\pi_1,\pi_2}\cdots M_{\pi_{t-1},\epsilon}\Bigr)_{x_k,x_i}\\
    &= \sum_{t\geq 0}\sum_{\pi_1,\ldots,\pi_{t-1}\in\Delta^*}\sum_{1\leq j_1, \ldots, j_{t-1}\leq n} (M_{\epsilon,\pi_1})_{x_k,x_{j_1}}(M_{\pi_1,\pi_2})_{x_{j_1},x_{j_2}}\cdots (M_{\pi_{t-1},\epsilon})_{x_{j_{t-1}},x_i}\\
    &= \bigl(\bigl(\sum_{t\geq 0} (M_{x,x})^t\bigr){}_{\epsilon,\epsilon}\bigr)_{x_k,x_i} = (((M_{x,x})^*)_{\epsilon,\epsilon})_{x_k,x_i}\mkomma
  \end{align*}
  where the third equality (and similarly the fourth equality) is by definition of induced pushdown matrices;
the blocks $(M_{\epsilon,X_k})_{x_i,x_j}$, $(M_{X_k,X_k})_{x_i,x_j}$ and $(M_{\epsilon,\epsilon})_{x_i,x_j}$ are the only non-null blocks that describe a step in the matrix starting from a state $x_i$ and having $\epsilon$ or $X_k$ as the topmost stack symbol.
\end{proof}

\begin{lemma}
  \label{lem:induced_star_x_extra}
  Let $M$ be a simple reset pushdown matrix induced by the Greibach normal form~\eqref{eqn:new_sys1}, \eqref{eqn:new_sys2}. Then, we have
  \[((M_{x,x}+M_{x,z}(M_{z,z})^*M_{z,x})^*)_{\epsilon,\epsilon} = ((M_{x,x})^*)_{\epsilon,\epsilon}\mpunkt\]
\end{lemma}
\begin{proof}
  Let $\Delta=\{X_1,\ldots,X_n\}$. In some sense similar to the proof of Lemma~\ref{lem:induced_star_x}, we have
  \begin{align*}
    & ((M_{x,x}+M_{x,z}(M_{z,z})^*M_{z,x})^*)_{\epsilon,\epsilon}\\
    &= \Bigl(\sum_{t\geq 0}(M_{x,x}+M_{x,z}(M_{z,z})^*M_{z,x})^t\Bigr)_{\epsilon,\epsilon}\\
    &= \sum_{t\geq 0}\sum_{\pi_1,\ldots,\pi_{t-1}\in\Gamma^*}(M_{x,x}+M_{x,z}(M_{z,z})^*M_{z,x})_{\epsilon,\pi_1}\cdots (M_{x,x}+M_{x,z}(M_{z,z})^*M_{z,x})_{\pi_{t-1},\epsilon}\allowdisplaybreaks\\
    &= \sum_{t\geq 0}\sum_{\pi_1,\ldots,\pi_{t-1}\in\Gamma^*}\Bigl((M_{x,x})_{\epsilon,\pi_1}+\bigl(\sum_{\pi,\pi'\in\Gamma^*} (M_{x,z})_{\epsilon,\pi}((M_{z,z})^*)_{\pi,\pi'}(M_{z,x})_{\pi,\pi_1}\bigr)\Bigr)\cdots (M_{x,x}+M_{x,z}(M_{z,z})^*M_{z,x})_{\pi_{t-1},\epsilon}\\
    &\stackrel{4}{=} \sum_{t\geq 0}\hspace{-.3cm}\sum_{\substack{\pi_1\in\Delta^*\\\pi_2,\ldots,\pi_{t-1}\in\Gamma^*}}\hspace{-.3cm}(M_{x,x})_{\epsilon,\pi_1}(M_{x,x}+M_{x,z}(M_{z,z})^*M_{z,x})_{\pi_1,\pi_2}\cdots (M_{x,x}+M_{x,z}(M_{z,z})^*M_{z,x})_{\pi_{t-1},\epsilon}\\
    &= \sum_{t\geq 0}\hspace{-.3cm}\sum_{\substack{\pi_1\in\Delta^*\\\pi_2,\ldots,\pi_{t-1}\in\Gamma^*}}\hspace{-.3cm}(M_{x,x})_{\epsilon,\pi_1} \Bigl((M_{x,x})_{\pi_1,\pi_2}+\bigl(\sum_{\pi,\pi'\in\Gamma^*} (M_{x,z})_{\pi_1,\pi}((M_{z,z})^*)_{\pi,\pi'}(M_{z,x})_{\pi,\pi_2}\bigr)\Bigr)\cdots (M_{x,x}+M_{x,z}(M_{z,z})^*M_{z,x})_{\pi_{t-1},\epsilon}\\
    &\stackrel{6}{=} \sum_{t\geq 0}\sum_{\pi_1,\ldots,\pi_{t-1}\in\Delta^*}\hspace{-.3cm} (M_{x,x})_{\epsilon,\pi_1} (M_{x,x})_{\pi_1,\pi_2}\cdots (M_{x,x})_{\pi_{t-1},\epsilon} = ((M_{x,x})^*)_{\epsilon,\epsilon}\mkomma
  \end{align*}
  where the fourth equality is because $(M_{x,z})_{\epsilon,\pi}=0$ for all $\pi\in\Gamma^*$. Similarly, for the sixth equality, we use the fact that $(M_{x,z})_{\pi_i,\pi}=0$ for all $\pi_i\in\Delta^*$ (and $\pi\in\Gamma^*$).
\end{proof}

\begin{lemma}
  \label{lem:induced_star}
  Let $M$ be a simple reset pushdown matrix induced by the Greibach normal form~\eqref{eqn:new_sys1}, \eqref{eqn:new_sys2} and $M'$ be induced by the Greibach normal form~~\eqref{eqn:2}.
  Then, for all $1\leq i\leq n$,
  \[((M^*)_{\epsilon,\epsilon})_{x_i,f} = ((M'^*)_{\epsilon,\epsilon})_{i,f}\mpunkt\]
\end{lemma}
\begin{proof}
  Note that by construction, we have
  \[\what M'=
  \begin{pmatrix}
    \what M_{x,x} & \what M_{x,f}\\
    0 & 0
  \end{pmatrix}\mpunkt\]

  By applying Lemma~\ref{lem:hat_star}, we infer
  \[M'^*=\begin{pmatrix}
  (M_{x,x})^* & (M_{x,x})^* M_{x,f}\\
  0 & 1
  \end{pmatrix}\mkomma\]
  and we get
  \begin{equation}\label{eqn:old_M}
    ((M'^*)_{\epsilon,\epsilon})_{i,f} = \bigl(((M_{x,x})^* M_{x,f})_{\epsilon,\epsilon}\bigr)_i\mpunkt
  \end{equation}
  
  At the same time, we have
  \begin{align*}
    \what M &= \begin{pmatrix}
      \what M_{z,z} & \what M_{z,x} & 0\\
      \what M_{x,z} & \what M_{x,x} & \what M_{x,f}\\
      0 & 0 & 0
    \end{pmatrix}\\
    &=\left(
    \begin{array}{c|c}
      \what M_{z,z} & \begin{matrix}
        \what M_{z,x} & 0
      \end{matrix}\\
      \hline\\[-1.5\medskipamount]
      \begin{matrix}
        \what M_{x,z} \\ 0
      \end{matrix} &
      \begin{matrix}
        \what M_{x,x} & \what M_{x,f}\\
        0 & 0
      \end{matrix}
    \end{array}
    \right)\mpunkt
  \end{align*}

  By Lemma~\ref{lem:hat_star}, we obtain
  \[M^*=\begin{pmatrix}
      \alpha^* & \alpha^*
      \begin{pmatrix}
        M_{z,x} & 0
      \end{pmatrix}\\
      \beta^* \begin{pmatrix}
        M_{x,z}\\ 0
      \end{pmatrix} & \beta^*
    \end{pmatrix}\mkomma\]
  with
  \begin{align*}
    \alpha &= M_{z,z}+
    \begin{pmatrix}
      M_{z,x} & 0
    \end{pmatrix}
    \begin{pmatrix}
      M_{x,x} & M_{x,f}\\
      0 & 0
    \end{pmatrix}^*
    \begin{pmatrix}
      M_{x,z}\\
      0
    \end{pmatrix}\\
    &= M_{z,z}+M_{z,x} (M_{x,x})^* M_{x,z}\mkomma
  \end{align*}
  by \eqref{eqn:star_induced_matrix} in the proof of Theorem~\ref{thm:induced_omega} and
  \begin{align}
    \beta^* &= \left(\begin{pmatrix}
      M_{x,x} & M_{x,f}\\
      0 & 0
    \end{pmatrix}+
    \begin{pmatrix}
      M_{x,z}\\ 0
    \end{pmatrix}
    (M_{z,z})^*
    \begin{pmatrix}
      M_{z,x} & 0
    \end{pmatrix}\right)^*\notag\\
    &= \left(\begin{pmatrix}
      M_{x,x} & M_{x,f}\\
      0 & 0
    \end{pmatrix}+
    \begin{pmatrix}
      M_{x,z}(M_{z,z})^*M_{z,x} & 0\\ 0 & 0
    \end{pmatrix}\right)^*\notag\\
    &= \begin{pmatrix}
      M_{x,x}+M_{x,z}(M_{z,z})^*M_{z,x} & M_{x,f}\\ 0 & 0
    \end{pmatrix}^*\notag\notag\\
    &= \begin{pmatrix}
      (M_{x,x}+M_{x,z}(M_{z,z})^*M_{z,x})^* & (M_{x,x}+M_{x,z}(M_{z,z})^*M_{z,x})^* M_{x,f}\\ 0 & 1
    \end{pmatrix}\mpunkt\label{eqn:star_induced_matrix_beta}
  \end{align}

  We deduce that
  \begin{align*}
    ((M^*)_{\epsilon,\epsilon})_{x_i,f}
    &= \bigl(\bigl((M_{x,x}+M_{x,z}(M_{z,z})^*M_{z,x})^* M_{x,f}\bigr)_{\epsilon,\epsilon}\bigr)_i\\
    &= \bigl(((M_{x,x}+M_{x,z}(M_{z,z})^*M_{z,x})^*)_{\epsilon,\epsilon} (M_{x,f})_{\epsilon,\epsilon}\bigr)_i\\
    &= \bigl(((M_{x,x})^*)_{\epsilon,\epsilon} (M_{x,f})_{\epsilon,\epsilon}\bigr)_i\\
    &= \bigl(((M_{x,x})^* M_{x,f})_{\epsilon,\epsilon}\bigr)_i\\
    &= ((M'^*)_{\epsilon,\epsilon})_{i,f}\mkomma
  \end{align*}
  where the third equality is by Lemma~\ref{lem:induced_star_x_extra} and the last equality is by~\eqref{eqn:old_M}. This concludes the proof.
\end{proof}

The following lemma investigates the final state $f$ in infinite paths. It states that a finite run of induced simple $\omega$-reset pushdown automata is equivalent to another path only through states $x$ and with symbol $Z_j$ initially on the pushdown tape and ending in state $z_j$ with an empty pushdown tape.
%The following lemma is an adaptation of Lemma~\ref{lem:run_without_f}.

\begin{lemma}
  \label{lem:newrun_without_f}
  Let $M$ be a simple reset pushdown matrix induced by the Greibach normal form~\eqref{eqn:new_sys1}, \eqref{eqn:new_sys2}.
  Then, for all $1\leq j,k\leq n$,
  \[((M^*)_{\epsilon,\epsilon})_{x_k,f}=\bigl(((M_{x,x})^*)_{Z_j,Z_j}M_{Z_j,\epsilon}\bigr)_{x_k,z_j}\mpunkt\]
\end{lemma}
\begin{proof}
  The beginning of the proof is similar to the proof of Lemma~10 of~\cite{finite_simple}. We obtain
    \vspace{-.2cm}
    \begin{align*}
    ((M^*)_{\epsilon,\epsilon})_{x_k,f} &= ((M^+)_{\epsilon,\epsilon})_{x_k,f} = ((M^*M)_{\epsilon,\epsilon})_{x_k,f}\\
    &= \!\!\sum_{1\leq v_1\leq 2n+1}\!((M^*)_{\epsilon,\epsilon})_{x_k,v_1}(M_{\epsilon,\epsilon})_{v_1,f}+
    \sum_{1\leq v_1\leq 2n+1}\sum_{P\in\Gamma} ((M^*)_{\epsilon,P})_{x_k,v_1}(M_{P,\epsilon})_{v_1,f}\\
    &\stackrel{4}{=} \!\!\sum_{1\leq v_1\leq 2n+1}\!((M^*)_{\epsilon,\epsilon})_{x_k,v_1}(M_{\epsilon,\epsilon})_{v_1,f}\\
    &\stackrel{5}{=} \sum_{1\leq i\leq n}((M^*)_{\epsilon,\epsilon})_{x_k,x_i}(M_{\epsilon,\epsilon})_{x_i,f}\\
    &\stackrel{6}{=} \sum_{1\leq i\leq n}(((M_{x,x})^*)_{\epsilon,\epsilon})_{x_k,x_i}(M_{\epsilon,\epsilon})_{x_i,f}\\
    &\stackrel{7}{=} \sum_{1\leq i\leq n}(((M_{x,x})^*)_{Z_j,Z_j})_{x_k,x_i}(M_{Z_j,\epsilon})_{x_i,z_j} = (((M_{x,x})^*)_{Z_j,Z_j}M_{Z_j,\epsilon})_{x_k,z_j}\mkomma
  \end{align*}
  where the fourth equality is since $(M_{P,\epsilon})_{v_1,f}=0$ for all $1\leq v_1\leq 2n+1$ and $P\in\Gamma$ by our construction. In the fifth equality, we use the fact that $(M_{\epsilon,\epsilon})_{v_1,f}=0$ for $v_1\neq x_i$ ($1\leq i\leq n$). The sixth equality is by Lemma~\ref{lem:induced_star_x}. The seventh equality is also by construction and by the definition of pushdown matrices.
\end{proof}

We now discuss the behaviors of our constructed simple $\omega$-reset pushdown automata.
\begin{lemma}
  Let the simple $\omega$-reset pushdown automata $\mathfrak{A}_m^l=(2n+1,\Gamma,I_m,M,P,l)$, for $1\leq m\leq n$ and $0\leq l\leq n$, be induced by the Greibach normal form~\eqref{eqn:new_sys1}, \eqref{eqn:new_sys2}.
  We then have
  \[\|\mathfrak{A}_m^l\|=((M^*)_{\epsilon,\epsilon})_{x_m,f}+((M^{\omega,l})_\epsilon)_{z_m}\mpunkt\]
\end{lemma}
\begin{proof}
  Let $1\leq m\leq n$ and $0\leq l\leq n$. We obtain
  \begin{align*}
  \|\mathfrak{A}_m^l\|&=I(M^*)_{\epsilon,\epsilon}P+I(M^{\omega,l})_\epsilon\\
  &=((M^*)_{\epsilon,\epsilon})_{x_m,f}+((M^*)_{\epsilon,\epsilon})_{z_m,f}+((M^{\omega,l})_\epsilon)_{x_m}+((M^{\omega,l})_\epsilon)_{z_m}\mkomma\\
  &=((M^*)_{\epsilon,\epsilon})_{x_m,f}+((M^*)_{\epsilon,\epsilon})_{z_m,f}+((M^{\omega,l})_\epsilon)_{z_m}\mpunkt
\end{align*}
  where the last equality is by Theorem~\ref{thm:induced_omega}.

It remains to show that $((M^*)_{\epsilon,\epsilon})_{z_m,f}=0$.
We have
\[\what M=\left(
\begin{array}{c|c}
  \what M_{z,z} & \begin{matrix}
    \what M_{z,x} & 0
  \end{matrix}\\
  \hline\\[-1.5\medskipamount]
  \begin{matrix}
    \what M_{x,z} \\ 0
  \end{matrix} &
  \begin{matrix}
    \what M_{x,x} & \what M_{x,f}\\
    0 & 0
  \end{matrix}
\end{array}
\right)\mpunkt\]

Now let
\[M^*=
\begin{pmatrix}
  \alpha & \beta\\
  \gamma & \delta
\end{pmatrix}\mkomma\]
where we are only interested in the second component of $\beta$. By lemma~\ref{lem:hat_star} and by \eqref{eqn:star_induced_matrix_beta} in the proof of Lemma~\ref{lem:induced_star}, we have
\begin{align*}
  \beta
  &=(M_{z,z})^*
  \begin{pmatrix}
    M_{z,x} & 0
  \end{pmatrix}
  \left[\begin{pmatrix}
      M_{x,x} & M_{x,f}\\
      0 & 0
    \end{pmatrix}+
    \begin{pmatrix}
      M_{x,z}\\ 0
    \end{pmatrix}(M_{z,z})^*
    \begin{pmatrix}
      M_{z,x} & 0
    \end{pmatrix}\right]^*\\
  %% &=\begin{pmatrix}
  %%   (M_{z,z})^* M_{z,x} & 0
  %% \end{pmatrix}
  %% \left[\begin{pmatrix}
  %%     M_{x,x} & M_{x,f}\\
  %%     0 & 0
  %%   \end{pmatrix}+
  %%   \begin{pmatrix}
  %%     M_{x,z} (M_{z,z})^* M_{z,x} & 0\\
  %%     0 & 0
  %%   \end{pmatrix}\right]^*\\
  %% &=\begin{pmatrix}
  %%   (M_{z,z})^* M_{z,x} & 0
  %% \end{pmatrix}
  %% \begin{pmatrix}
  %%     M_{x,x} + M_{x,z} (M_{z,z})^* M_{z,x} & M_{x,f}\\
  %%     0 & 0
  %%   \end{pmatrix}^*\\
  &=\begin{pmatrix}
    (M_{z,z})^* M_{z,x} & 0
  \end{pmatrix}
  \begin{pmatrix}
      (M_{x,x} + M_{x,z} (M_{z,z})^* M_{z,x})^* & (M_{x,x} + M_{x,z} (M_{z,z})^* M_{z,x})^* M_{x,f}\\
      0 & 1
    \end{pmatrix}\\
  &=\begin{pmatrix}
    (M_{z,z})^* M_{z,x}(M_{x,x} + M_{x,z} (M_{z,z})^* M_{z,x})^*, & (M_{z,z})^* M_{z,x}(M_{x,x} + M_{x,z} (M_{z,z})^* M_{z,x})^* M_{x,f}
  \end{pmatrix}\mpunkt
\end{align*}

Now, we obtain
\begin{align}
  ((M^*)_{\epsilon,\epsilon})_{z_m,f} &=\Bigl(\Bigl((M_{z,z})^* M_{z,x}(M_{x,x} + M_{x,z} (M_{z,z})^* M_{z,x})^* M_{x,f}\Bigr)_{\!\!\epsilon,\epsilon}\,\Bigr)_m\notag\\
  &=\Bigl(\bigl((M_{z,z})^* M_{z,x}(M_{x,x} + M_{x,z} (M_{z,z})^* M_{z,x})^*\bigr)_{\epsilon,\epsilon} (M_{x,f})_{\epsilon,\epsilon} \Bigr)_m\notag\\
  &=\Bigl(((M_{z,z})^*)_{\epsilon,\epsilon} \bigl(M_{z,x}(M_{x,x} + M_{x,z} (M_{z,z})^* M_{z,x})^*\bigr)_{\epsilon,\epsilon} (M_{x,f})_{\epsilon,\epsilon} \Bigr)_m\notag\\
  &=\sum_{1\leq i\leq n}\Bigl(((M_{z,z})^*)_{\epsilon,\epsilon} (M_{z,x})_{\epsilon,Z_i}\bigl((M_{x,x} + M_{x,z} (M_{z,z})^* M_{z,x})^*\bigr)_{Z_i,\epsilon} (M_{x,f})_{\epsilon,\epsilon} \Bigr)_m\mkomma\label{eq:continue_here}
\end{align}
where in the second equality, we have $(M_{x,f})_{\pi,\epsilon}=0$ for $\pi\neq \epsilon$. The third equality uses that $(M_{z,z})^*)_{\epsilon,\pi}=0$ for $\pi\neq \epsilon$. In the fourth equality, we have $(M_{z,x})_{\epsilon,\pi}=0$ for $\pi\notin \{Z_i \mid 1\leq i\leq n\}$.

We concentrate on the factor in the center, where we have
\begin{align*}
  &\bigl((M_{x,x} + M_{x,z} (M_{z,z})^* M_{z,x})^*\bigr)_{Z_i,\epsilon}\\
  &= \sum_{t\geq 0}\sum_{\pi_1,\ldots,\pi_{t-1}\in\Gamma^*}\hspace{-.4cm}(M_{x,x} + M_{x,z} (M_{z,z})^* M_{z,x})_{Z_i,\pi_1}\cdots (M_{x,x} + M_{x,z} (M_{z,z})^* M_{z,x})_{\pi_{t-1},\epsilon}\\
  &= \sum_{t\geq 0}\sum_{\pi_1,\ldots,\pi_{t-2}\in\Gamma^*}\hspace{-.4cm}(M_{x,x} + M_{x,z} (M_{z,z})^* M_{z,x})_{Z_i,\pi_1}\cdots (M_{x,x} + M_{x,z} (M_{z,z})^* M_{z,x})_{\pi_{t-2},\epsilon}(M_{x,x})_{\epsilon,\epsilon}\\
  &= \sum_{t\geq 0}(M_{x,x})_{Z_i,\epsilon}\cdots (M_{x,x})_{\epsilon,\epsilon}(M_{x,x})_{\epsilon,\epsilon}= 0\mkomma
\end{align*}
where in the second (and similarly in the third) equality we have $(M_{x,x} + M_{x,z} (M_{z,z})^* M_{z,x})_{\pi_{t-1},\epsilon}=(M_{x,x})_{\epsilon,\epsilon}$ because $(M_{x,x})_{\pi_{t-1},\epsilon}=0$ for $\pi_{t-1}\neq \epsilon$ and because
\[(M_{x,z} (M_{z,z})^* M_{z,x})_{\pi_{t-1},\epsilon}=\sum_{\pi,\pi'\in\Gamma^*}(M_{x,z})_{\pi_{t-1},\pi}((M_{z,z})^*)_{\pi,\pi'}(M_{z,x})_{\pi',\epsilon}=0\]
as $(M_{z,x})_{\pi',\epsilon}=0$ for all $\pi'$.
In the last equality, $(M_{x,x})_{Z_i,\epsilon}=0$.

We now plug this into~\eqref{eq:continue_here} and obtain
\begin{align*}
  ((M^*)_{\epsilon,\epsilon})_{z_m,f} &=\sum_{1\leq i\leq n}\Bigl(((M_{z,z})^*)_{\epsilon,\epsilon} (M_{z,x})_{\epsilon,Z_i}\bigl((M_{x,x} + M_{x,z} (M_{z,z})^* M_{z,x})^*\bigr)_{Z_i,\epsilon} (M_{x,f})_{\epsilon,\epsilon} \Bigr)_m\\
  &=\sum_{1\leq i\leq n}\Bigl(((M_{z,z})^*)_{\epsilon,\epsilon} (M_{z,x})_{\epsilon,Z_i} 0 (M_{x,f})_{\epsilon,\epsilon} \Bigr)_m=0\mpunkt
\end{align*}
This completes the proof.
\end{proof}

The following theorem compares the behavior of induced simple $\omega$-reset pushdown automata with the solutions of system~\eqref{eqn:5new} %. %Note that in the semimodule part $(M^{\omega,l})_\epsilon$ of the behavior, state $f$ will never be reached.
%The following theorem extends the previous theorem 
by stating that $(\|\mathfrak{A}_1^l\|,\dots,\|\mathfrak{A}_n^l\|)$ is a canonical solution of~\eqref{eqn:5new}.
%failed for the old construction (Theorem~\ref{thm:simple_canonical_sol}).

\begin{theorem}
  \label{thm:newsimple_canonical_sol}
  Let $(S,V)$ be a complete semiring-semimodule pair.
  Let the simple $\omega$-reset pushdown automata $\mathfrak{A}_m^l$, for $1\leq m\leq n$ and $0\leq l\leq n$, be induced by the Greibach normal form~\eqref{eqn:new_sys1}, \eqref{eqn:new_sys2}.

Then, for $0\leq l\leq n$,
\[(\|\mathfrak{A}_1^l\|,\dots,\|\mathfrak{A}_n^l\|)=
\big(((M^*)_{\epsilon,\epsilon})_{x_1,f}+((M^{\omega,l})_\epsilon)_{z_1},\dots,((M^*)_{\epsilon,\epsilon})_{x_n,f}+((M^{\omega,l})_\epsilon)_{z_n}\big)\]
is the $l$\textsuperscript{th} canonical solution of~\eqref{eqn:5new}.
\end{theorem}

\begin{proof}
  We show that
  \[(((M^*)_{\epsilon,\epsilon})_{x_1,f},\dots,((M^*)_{\epsilon,\epsilon})_{x_n,f}) \qquad \text{and} \qquad (((M^{\omega,l})_\epsilon)_{z_1},\dots,((M^{\omega,l})_\epsilon)_{z_n})\]
  is the $l$\textsuperscript{th} canonical solution of the mixed $\omega$-algebraic system~\eqref{eqn:new_sys1}, \eqref{eqn:new_sys2}.

  Let $M'$ be induced by the Greibach normal form~~\eqref{eqn:2}. Then, by Theorem~\ref{thm:8}, $(((M'^*)_{\epsilon,\epsilon})_{1,f},\allowbreak \dots,\allowbreak ((M'^*)_{\epsilon,\epsilon})_{n,f})$ is the unique (and therefore least) solution of~\eqref{eqn:2}. By Lemma~\ref{lem:induced_star} and by equality of \eqref{eqn:2} and \eqref{eqn:new_sys1}, we can conclude that $\sigma=(((M^*)_{\epsilon,\epsilon})_{x_1,f},\allowbreak \dots,\allowbreak ((M^*)_{\epsilon,\epsilon})_{x_n,f})$ is also the least solution of~\eqref{eqn:new_sys1}.

  Fix $l$ with $1\leq l\leq n$ for the remainder of the proof.
  It remains to show that for the system \eqref{eqn:new_sys2}, written as $z=\varrho(x)z$, we have
  \[\varrho(\sigma)^{\omega,l}=(((M^{\omega,l})_\epsilon)_{z_1},\dots,((M^{\omega,l})_\epsilon)_{z_n})\]%=((M^{\omega,l})_\epsilon)_z\mpunkt\]

  We start with the right side of equation~\eqref{eqn:new_sys2}. We have, for $1\leq i\leq n$,
  \begin{align*}
      \varrho(\sigma)_i z &= \sum_{1\leq j,k\leq n}\sum_{a\in\Sigma}(p_i, a x_j z_k)a \sigma_j z_k+\sum_{1\leq j\leq n}\sum_{a\in\Sigma}(p_i,a z_j)a z_j\\
      &= \sum_{1\leq j,k\leq n}\sum_{a\in\Sigma}(p_i, a x_k z_j)a \sigma_k z_j+\sum_{1\leq j\leq n}\sum_{a\in\Sigma}(p_i,a z_j)a z_j\\
      &= \sum_{1\leq j\leq n}\Big(\sum_{1\leq k\leq n}\sum_{a\in\Sigma}(p_i, a x_k z_j)a \sigma_k +\sum_{a\in\Sigma}(p_i,a z_j)a \Big)z_j\\
      &= \sum_{1\leq j\leq n}\Big(\sum_{1\leq k\leq n}(M_{\epsilon,Z_j})_{z_i,x_k} ((M^*)_{\epsilon,\epsilon})_{x_k,f} +(M_{\epsilon,\epsilon})_{z_i,z_j} \Big)z_j\\
      &\stackrel{5}{=} \sum_{1\leq j\leq n}\Big(\sum_{1\leq k\leq n}(M_{\epsilon,Z_j})_{z_i,x_k} (((M_{x,x})^*)_{Z_j,Z_j}M_{Z_j,\epsilon})_{x_k,z_j} +(M_{\epsilon,\epsilon})_{z_i,z_j} \Big)z_j\\
      &= \sum_{1\leq j\leq n}\Big(\sum_{1\leq k,k'\leq n}(M_{\epsilon,Z_j})_{z_i,x_k} (((M_{x,x})^*)_{Z_j,Z_j})_{x_k,x_{k'}}(M_{Z_j,\epsilon})_{x_{k'},z_j} +(M_{\epsilon,\epsilon})_{z_i,z_j} \Big)z_j\\
      &= \sum_{1\leq j\leq n}\Big(\sum_{1\leq k,k'\leq n}(\what M_{z_i,x_k})_{\epsilon,Z_j} (((\what M_{x,x})^*)_{x_k,x_{k'}})_{Z_j,Z_j}(\what M_{x_{k'},z_j})_{Z_j,\epsilon} +(\what M_{z_i,z_j})_{\epsilon,\epsilon} \Big)z_j\\
      &\stackrel{8}{=} \sum_{1\leq j\leq n}\Big(\sum_{1\leq k,k'\leq n}\sum_{P\in\Gamma}(\what M_{z_i,x_k})_{\epsilon,P} (((\what M_{x,x})^*)_{x_k,x_{k'}})_{P,P}(\what M_{x_{k'},z_j})_{P,\epsilon} +(\what M_{z_i,z_j})_{\epsilon,\epsilon} \Big)z_j\\
      &= \sum_{1\leq j\leq n}\Big(\sum_{1\leq k,k'\leq n} \big(\what M_{z_i,x_k} ((\what M_{x,x})^*)_{x_k,x_{k'}}\what M_{x_{k'},z_j}\big)_{\epsilon,\epsilon} +(\what M_{z_i,z_j})_{\epsilon,\epsilon} \Big)z_j\\
    &= \sum_{1\leq j\leq n}\Big(\what M_{z_i,x} (\what M_{x,x})^*\what M_{x,z_j} + \what M_{z_i,z_j}\Big)_{\epsilon,\epsilon} z_j\mkomma
  \end{align*}
  where the fifth equality is by Lemma~\ref{lem:newrun_without_f}. The eighth equality is because for $P\neq Z_j$, we have $(\what M_{x_{k'},z_j})_{P,\epsilon}=0$.

  Now for $\varrho$ of the system $z=\varrho(x)z$, we obtain
  \begin{align*}
    \varrho(\sigma)
    &= \begin{pmatrix}\big(\what M_{z_1,x} (\what M_{x,x})^*\what M_{x,z_1} + \what M_{z_1,z_1}\big)_{\epsilon,\epsilon} & \cdots & \big(\what M_{z_1,x} (\what M_{x,x})^*\what M_{x,z_n} + \what M_{z_1,z_n}\big)_{\epsilon,\epsilon}\\
      \vdots & \ddots & \vdots\\
      \big(\what M_{z_n,x} (\what M_{x,x})^*\what M_{x,z_1} + \what M_{z_n,z_1}\big)_{\epsilon,\epsilon} & \cdots & \big(\what M_{z_n,x} (\what M_{x,x})^*\what M_{x,z_n} + \what M_{z_n,z_n}\big)_{\epsilon,\epsilon}
    \end{pmatrix}\\
    &= \Big(M_{z,x} (M_{x,x})^*M_{x,z} + M_{z,z}\Big)_{\epsilon,\epsilon}{}_{\mpunkt}
  \end{align*}

  Then, we apply the identity \eqref{eqn:pl_identity} and we get
  \begin{align}
    \bigl(\varrho(\sigma)^{\omega,l}\bigr)_{j}
    &= \Bigl(\bigl( \big(M_{z,x} (M_{x,x})^*M_{x,z} + M_{z,z}\big)_{\epsilon,\epsilon}\bigr)^{\omega,l}\Big)_j\notag\allowdisplaybreaks\\
    %&=\hspace{-.4cm}\sum_{(j_1,j_2,\ldots)\in P_l}\hspace{-.4cm} \Bigl(\big(M_{z,x} (M_{x,x})^*M_{x,z} + M_{z,z}\big)_{\epsilon,\epsilon}\Bigr)_{j,j_1} \Bigl(\big(M_{z,x} (M_{x,x})^* M_{x,z} + M_{z,z}\big)_{\epsilon,\epsilon}\Bigr)_{j_1,j_2} \cdots\notag\allowdisplaybreaks\\
    &=\hspace{-.4cm}\sum_{(j_1,j_2,\ldots)\in P_l}\hspace{-.4cm} \big(\what M_{z_j,x} (\what M_{x,x})^*\what M_{x,z_{j_1}} + \what M_{z_j,z_{j_1}}\big)_{\epsilon,\epsilon} \big(\what M_{z_{j_1},x} (\what M_{x,x})^*\what M_{x,z_{j_2}} + \what M_{z_{j_1},z_{j_2}}\big)_{\epsilon,\epsilon} \cdots\notag\\
    &\stackrel{4}{=}\hspace{-.4cm}\sum_{(j_1,j_2,\ldots)\in P_l}\sum_{\mathrlap{\hspace{-.2cm}\pi_1,\pi_2,\ldots\in\Gamma^*}}\;\big(\what M_{z_j,x} (\what M_{x,x})^*\what M_{x,z_{j_1}} + \what M_{z_j,z_{j_1}}\big)_{\epsilon,\pi_1} \big(\what M_{z_{j_1},x} (\what M_{x,x})^*\what M_{x,z_{j_2}} + \what M_{z_{j_1},z_{j_2}}\big)_{\pi_1,\pi_2} \cdots\notag\\
    &=\Bigl(\bigl((M_{z,x}(M_{x,x})^*M_{x,z}+M_{z,z})^{\omega,l}\bigr)_{\!\epsilon}\,\Bigr)_{\!j}\mkomma \label{solution_system}
  \end{align}
  where the fourth equality uses the fact that $\big(\what M_{z_i,x}(\what M_{x,x})^*\what M_{x,z_j} + \what M_{z_i,z_j}\big)_{\epsilon,\pi}=0$ for $\pi\neq \epsilon$, which is because $(\what M_{z_i,z_j})_{\epsilon,\pi} = 0$ for $\pi\neq \epsilon$ by definition and because, by our construction, we have
  \[M_{z,x}(M_{x,x})^*M_{x,z}=\sum_{1\leq j\leq n}(M_{z,x})_{\epsilon,Z_j}((M_{x,x})^*)_{Z_j,Z_j}(M_{x,z})_{Z_j,\epsilon}\mpunkt\]
  Inductively, the above argument can be applied to all factors $\big(\what M_{z_{j_i},x} (\what M_{x,x})^*\what M_{x,z_{j_{i+1}}} + \what M_{z_{j_i},z_{j_{i+1}}}\big)_{\pi_i,\pi_{i+1}}$ because we learn from the preceding factor that $\pi_i=\epsilon$.

  Now, we proceed from the other direction. 
  From Theorem~\ref{thm:induced_omega}, we know that for the simple $\omega$-reset pushdown automaton $\mathfrak{A}_m^l$ and a variable $z_j$, we have
  \begin{align*}
    ((M^{\omega,l})_\epsilon)_{z_j} &= \Bigl(\bigl(\bigl(M_{z,z}+M_{z,x} (M_{x,x})^* M_{x,z}\bigr)^{\omega,l}\bigr)_\epsilon\Bigr)_j\\
    &=\varrho(\sigma)_{j}^{\omega,l}\mkomma
  \end{align*}
  where the last equality is by \eqref{solution_system}. This completes the proof.
\end{proof}

We now combine our previous discussion and Theorem~\ref{thm:newsimple_canonical_sol} to get our second main result.

\begin{corollary}
  Let $S$ be a continuous star-omega semiring with the underlying semiring $S$ being commutative and let $r\in \salgtimes$.

  Then there exists a simple $\omega$-reset pushdown automaton with behavior $r$.
\end{corollary}

\begin{proof}
  Let $r\in\salgtimes$.
  As discussed on page~\pageref{pagemark_discussions}, by Theorem~\ref{thm:unmix} (and Theorem~\ref{thm:3.1}), $r$ is a component of a canonical solution of an $\omega$-algebraic system in Greibach normal form over $\stimes$.
  Let \eqref{eqn:5new} be such a system and assume that the $m$\textsuperscript{th} component of the $l$\textsuperscript{th} canonical solution of \eqref{eqn:5new} is $r$, i.e., assume $\tau_m=r$ for the $l$\textsuperscript{th} canonical solution $\tau$.
  %  \[\tau_i = p_i(\tau)\]

  Now, we can construct the simple $\omega$-reset pushdown automata $\mathfrak{A}_m^l$ induced by the Greibach normal form~\eqref{eqn:new_sys1}, \eqref{eqn:new_sys2}, for which, by Theorem~\ref{thm:newsimple_canonical_sol}, $(\|\mathfrak{A}_1^l\|,\dots,\|\mathfrak{A}_n^l\|)$ is the $l$\textsuperscript{th} canonical solution of \eqref{eqn:5new}. As the $l$\textsuperscript{th} canonical solution is unique, we can conclude that
  \[\|\mathfrak{A}_m^l\|=\tau_m=r\mpunkt\qedhere\]
\end{proof}

\section{Discussion}
We have extended the characterization of $\omega$-algebraic series so that we can use the $\omega$-Kleene closure to transfer the property of Greibach normal form from algebraic systems to mixed $\omega$-algebraic systems. This generalizes a fundamental property from context-free languages.

We believe that the same technique can be used to transfer other properties of algebraic systems to infinite words. Cohen, Gold~\cite{cohen_gold} use this technique also for the elimination of chain rules, for the Chomsky normal form and for effective decision methods of emptiness, finiteness and infiniteness.

The second part of this paper applies the Greibach normal form for the construction of $\omega$-pushdown automata. Simple $\omega$-reset pushdown automata do not use $\epsilon$-transitions; in the literature, this is also called a \emph{realtime} pushdown automaton. Realtime pushdown automata read a symbol of the input word in every transition---exactly like context-free grammars in Greibach normal form generate a letter in every derivation step. Additionally, each derivation step of context-free grammars in Greibach normal form increases the number of non-terminals in the sentential form by at most one. We showed that for realtime pushdown automata it suffices to handle at most one stack symbol per transition. Here the Greibach normal form provides exactly the properties needed to construct simple $\omega$-reset pushdown automata.

For our proof in the second part of the paper, we exploit the following connections.
The $l$\textsuperscript{th} canonical solutions are by definition unique. This allows us to perform the following proof method in Section~\ref{sec:simple_omega_reset_pushdown_automata}: The proof that each of two expressions is the $m$\textsuperscript{th} component of the $l$\textsuperscript{th} canonical solution implies the equality of these two expressions. (Compare this with the proof method in continuous semirings: The proof that each of two expressions is the $m$\textsuperscript{th} component of the least solution of an algebraic system implies the equality of these two expressions.)
In our proof, we consider an $\omega$-algebraic series that is the $m$\textsuperscript{th} component of the $l$\textsuperscript{th} canonical solution of an $\omega$-algebraic system in Greibach normal form and construct a simple $\omega$-reset pushdown automaton whose moves depend only on the coefficients of this Greibach normal form. We prove that the behavior of this simple $\omega$-reset pushdown automaton equals the $m$\textsuperscript{th} component of the $l$\textsuperscript{th} canonical solution of this Greibach normal form.

The model of simple $\omega$-reset pushdown automata seems to be very natural. They occur when applying general homomorphisms to nested-word automata~\cite{note_nw,unweighted_logic}. Their unweighted counterparts have been used for a Büchi-type logical characterization of timed pushdown languages of finite words~\cite{perevoshchikov_logic} and $\omega$-context-free languages~\cite{unweighted_logic}.
Also in the weighted setting, simple reset pushdown automata of finite words have been used in~\cite{perevoshchikov_weighted}.

We use a similar automaton model as simple $\omega$-reset pushdown automata for a Büchi-type logical characterization in \cite{weighted_logic}. There, we introduce a weighted logic and prove its expressive equivalence to the new automaton model. Restricted to the weight structure used in the current paper, we can therefore extend our result there by stating that every $\omega$-algebraic series can be converted to a formula of our weighted logic.

%%%%%%%%%%%%%%%%%%%%%%%%%%%%%%

\paragraph{Acknowledgment}
We thank the anonymous reviewers for their valuable feedback.

%%%%%%%%%%%%%%%%%%%%%%%%%%%%%%

\bibliography{main}

%%%%%%%%%%%%%%%%%%%%%%%%%%%%%%%%%%%%%%%%%%%%%%%%%%%%%%%%%%%%%%%%%%%%%%%%%%%%%%%%

\appendix

\end{document}